\documentclass[review]{siamart251216}

\usepackage{amssymb}
\usepackage{subcaption}
\usepackage{bm}
\graphicspath{{./}{images/}}
\usepackage{enumerate}
\usepackage{algorithmic}
\usepackage{soul}
\usepackage{mathrsfs}
\usepackage{tikz}
\usetikzlibrary{shapes.geometric, positioning, calc, arrows.meta}

\newsiamthm{assumption}{Assumption}
\newsiamremark{remark}{Remark}

\headers{G. Bao, Y. Zhang, and J. Yu}
{$c$-Transform for High-Contrast Reflector Design}

\title{A $c$-Transform Optimal Transport Method for High-Contrast Freeform Reflector Design}
\author{Gang Bao\thanks{School of Mathematical Sciences, Zhejiang University,
Hangzhou 310027, China (\email{baog@zju.edu.cn}).}
\and Yixuan Zhang\thanks{School of Mathematical Sciences, Peking University,
Beijing 100871, China (\email{yixuan@math.pku.edu.cn}). Corresponding author.}
\and Jiaqi Yu\thanks{School of Mathematical Sciences, Zhejiang University,
Hangzhou 310027, China (\email{12135024@zju.edu.cn}).}}

\begin{document}
\nolinenumbers
\maketitle

\begin{abstract}
High-contrast target distributions pose a particular challenge in freeform reflector design, as the presence of zero-intensity regions leads to degeneracy in the associated Monge--Amp\`ere-type equation. A common remedy is to add a positive artificial background to the target intensity, which improves the regularity of the equation. However, this regularization inevitably reduces the achievable illumination contrast and introduces nonzero intensity into regions that are intended to remain dark.
We propose a fast c-transform method based on optimal transport duality for high-contrast freeform reflector design in the far field without artificial background regularization. The method integrates repeated discrete c-transforms into a measure-based dual optimization scheme to maintain the $c$-concavity of the reflector potential throughout the iteration. This preserves the admissibility of the iterates and enables stable convergence for targets with zero-intensity regions. To efficiently compute the discrete $c$-transforms, we develop a localized search algorithm that exploits their contact structure and prove its exactness and linear complexity. Numerical experiments demonstrate that the proposed method accurately realizes high-contrast illumination targets with zero background.
\end{abstract}



\section{Introduction}
Freeform optical elements are used to redistribute optical energy according to a prescribed illumination pattern. 
In this work, we consider the inverse design of a far-field reflector. In this optical system, the source rays emanate in directions 
$\boldsymbol{x}\in\Omega\subset\mathbb{S}_{-}^{2}$, are reflected by the surface
\begin{equation}\label{eq:radial-surface}
\Gamma=\{\rho(\boldsymbol{x})\boldsymbol{x}:\boldsymbol{x}\in\Omega\},
\end{equation}
and then propagate to the target domain 
$\boldsymbol{y}=T(\boldsymbol{x})\in\Omega^*\subset\mathbb{S}_{+}^{2}$, where $T$ is the reflective optical map following the reflection law; see Fig.~\ref{fig:far_field_reflector_system}. 
Here, $\mathbb{S}_{-}^{2}$ and $\mathbb{S}_{+}^{2}$ denote the lower and upper unit hemispheres, respectively. The source intensity is denoted by $f(\boldsymbol{x})$, $\boldsymbol{x}\in\Omega$, and the prescribed far-field intensity by $g(\boldsymbol{y})$, $\boldsymbol{y}\in\Omega^*$.

In the absence of loss, the optical map $T$ induced by $\Gamma$ is required to preserve energy \cite{CGH2008Reflector,OT2}. This condition can be written as $T_{\#}f=g$, i.e.,
\begin{equation}\label{eq:measure-preserving}
    \int_{T^{-1}(B)} f(\boldsymbol{x}) \,\mathrm{d}\boldsymbol{x}
    =
    \int_B g(\boldsymbol{y}) \,\mathrm{d}\boldsymbol{y},
    \quad \forall \text{ Borel set } B \subset \Omega^*,
\end{equation}
where $T_{\#}f$ denotes the pushforward of the source intensity measure by $T$. The inverse problem is therefore to determine $\Gamma$ realizing this redistribution.

\begin{figure}[H]
    \centering
    \includegraphics[width=0.70\textwidth]{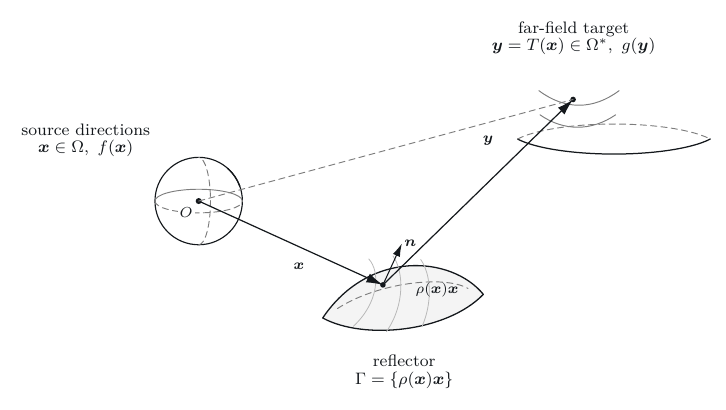}
    \caption{Schematic of the far-field reflector system.}
    \label{fig:far_field_reflector_system}
\end{figure}

In many numerical treatments, when the optical map $T$ is smooth and single-valued, the conservation law (\ref{eq:measure-preserving}) is converted into the Jacobian equation $f(\boldsymbol{x})=g(T(\boldsymbol{x}))|\det\nabla T(\boldsymbol{x})|$.
Together with the reflection law and the geometry of the optical surface, this
relation leads to a Monge--Amp\`ere-type equation for the radial function of $\Gamma$ \cite{wang1996design},
\begin{equation}\label{eq:MA}
\begin{aligned}
\frac{1}{\eta^2}
\det\left(
    -D^2 \rho+\frac{2}{\rho}D\rho\otimes D\rho
    -(\rho-\eta)I
\right)&=\frac{f(\boldsymbol{x})}{g(T(\boldsymbol{x}))},\\
\eta=\left(|\nabla \rho|^2+\rho^2\right) / 2 \rho,&\qquad \text{on }\Omega.
\end{aligned}
\end{equation}
Classical regularity theory for \eqref{eq:MA} is typically established under suitable structural and geometric conditions of $\Omega,\Omega^*$, together with nondegenerate source and target densities. One commonly assumes \cite{TW2009SecondBoundary,CGH2008Reflector,Loeper2009Regularity}
\begin{equation}\label{eq:density-bounds}
    0<\lambda
\leq f(\boldsymbol{x}),\,g(\boldsymbol{y})
\leq \Lambda<\infty,
\qquad
\boldsymbol{x}\in\Omega,\quad
\boldsymbol{y}\in\Omega^*,
\end{equation}
while $\Omega$ and $\Omega^*$ are assumed to be $c$ and $c^*$-convex with respect to each other for the reflector cost $c(\boldsymbol{x},\boldsymbol{y})=-\log(1-\boldsymbol{x}\cdot \boldsymbol{y})$; see \cite{Villani2009} for the relevant definitions and properties.

Existing numerical strategies for inverse reflector and refractor design may be divided into several groups. Ray-mapping methods first construct a map between the source and the target and then recover an optical surface compatible with this map \cite{bosel2017single,desnijder2019ray,feng2016freeform,Fournier2010fast}. These methods are intuitive and computationally efficient, but the constructed map must satisfy an integrability condition in order to correspond to a physically realizable optical surface, which can be difficult to guarantee in general \cite{desnijder2019ray,Fournier2010fast}.

PDE-based methods formulate inverse design as a generalized Monge--Amp\`ere-type equation \eqref{eq:MA}. Numerical approaches include discretizations of the resulting fully nonlinear equation, such as finite difference and finite element methods often combined with Newton-type solvers, as well as variational and least-squares formulations \cite{FD&Newton_2,FD&Newton_3,FD&Newton_4,FD&Newton_5,FD&Newton_6,GJENewton,OT3,LS1,LS2,LS3,LS5}. Such methods provide a direct and systematic connection between the optical surface and the prescribed intensity distribution. They are particularly effective when the optical map is regular and the source and target densities satisfy the regularity conditions in \eqref{eq:density-bounds}. However, these PDE-based methods are generally not directly applicable to high-contrast targets. In this case, the density ratio in \eqref{eq:MA} may become singular and the regularity required by the pointwise Monge--Amp\`ere formulation may break down \cite{TW2009SecondBoundary,CGH2008Reflector,Loeper2009Regularity}.

Recently, a Sobolev gradient method was developed for the far-field reflector problem from the viewpoint of optimal transport \cite{Far_ref}. It optimizes the Kantorovich dual functional using a Sobolev-preconditioned ascent direction. Its strict ascent property, however, is established in the regular positive-density regime \eqref{eq:density-bounds}. Outside this regime, an unconstrained update need not preserve the $c$-concavity, and hence the admissibility, of the reflector potential.

Supporting quadric methods (SQM) and closely related semi-discrete optimal transport (OT) methods represent the optical surface as an envelope of local optical pieces and determine the associated parameters by matching prescribed target point energies \cite{SQM1,SQM2,SQM5}. These methods are well suited to discrete target data and can handle nonuniform and high-contrast irradiance distributions. A major limitation arises from the discrete representation of the target measure. For continuous illumination patterns, the target must be sampled by a large number of points to suppress discretization artifacts, which limits the scalability of these methods for high-resolution targets. A further computational challenge lies in the geometry of the induced partition. In quadratic OT, the partition is represented by Laguerre or power diagrams, for which efficient construction and update algorithms are well established \cite{SemiDiscreteOT1,SemiDiscreteOT2}. For the far-field reflector problem, the corresponding partition is generated by supporting paraboloids on the sphere \cite{ParaboloidPowerDiagram}, leading to a more complex partition construction than in quadratic OT. For general optical systems such as near-field problems \cite{OT1,OT4,SQMNearField}, tractable diagram constructions may no longer be available, and the resulting computations can scale quadratically with the discretization size.

Our method incorporates repeated $c$-transforms into a bidirectional Sobolev iteration based on Kantorovich duality, restoring $c$-concavity after each update and preserving the ascent property even in singular regimes. This framework removes two key restrictions of existing approaches. Unlike SQM, it works directly with continuous measures, providing a more efficient discretization for high-resolution illumination patterns. More importantly, it directly handles nonsmooth reflector surfaces arising from high-contrast targets with zero-intensity regions, without artificial background regularization. This direct formulation also carries over naturally to other optical systems.

The main contributions of this work are as follows. We evaluate the gradient of the dual functional directly in terms of the pushforward measures induced by the dual contact maps, avoiding Hessian and determinant evaluations in singular regimes. We develop a fast discrete $c$-transform on spherical meshes using local geometric reconstruction and localized search, and establish interior exactness and linear complexity for the resulting search strategy. A finite element discretization provides a unified computational framework for the measure-based updates and discrete $c$-transforms. Numerical experiments demonstrate the accuracy, stability, and efficiency of the method for high-contrast far-field reflector design.

The paper is organized as follows. Section \ref{sec:problem} reviews the far-field reflector formulation and introduces the bidirectional framework. Section \ref{sec:fast-ctransform} develops the finite element discrete $c$-transform. Section \ref{sec:bidirectional} presents the pushforward evaluation and Sobolev preconditioned updates. Numerical experiments are presented in Section \ref{sec:numerical-experiments}.

\section{Problem Formulation and Measure-Based Framework}\label{sec:problem}
Throughout the paper, $\Omega\subset\mathbb S_-^2$ and $\Omega^*\subset\mathbb S_+^2$ are open spherical caps compactly contained in opposite hemispheres, so that $1-\boldsymbol x\cdot\boldsymbol y\ge\delta_c>0$ on $\Omega\times\Omega^*$. We allow the source and target densities to have infinite contrast, $\inf_\Omega f=\inf_{\Omega^*}g=0$, without adding an artificial positive background. Their nonzero supports may contain holes, corners, or multiple components. We assume only that, for some $d_0>0$, they remain uniformly separated from the domain boundaries:
\begin{equation}\label{eq:dis_condition}
\operatorname{dist}(\operatorname{supp} f,\partial\Omega)\ge d_0,
\qquad
\operatorname{dist}(\operatorname{supp} g,\partial\Omega^*)\ge d_0.
\end{equation}

\subsection{Reflector Model and Optimal Transport}
With the radial form of the reflector surface given in \eqref{eq:radial-surface}, we introduce $u(\boldsymbol{x}):=\log \rho(\boldsymbol{x})$. It has been shown \cite{OT2,wang2004design} that the reflector design problem can be formulated as an optimal transport problem with cost
\begin{equation}\label{reflector_cost}
    c(\boldsymbol{x},\boldsymbol{y})=-\log(1-\boldsymbol{x}\cdot \boldsymbol{y}).
\end{equation}
The corresponding Kantorovich dual problem is
\begin{equation}\label{eq:dual-OT}
    \sup_{(u,v)\in \mathcal{K}}
\left\{
\int_{\Omega} u(\boldsymbol{x})f(\boldsymbol{x})\,d\boldsymbol{x}
+\int_{\Omega^*} v(\boldsymbol{y})g(\boldsymbol{y})\,d\boldsymbol{y}
\right\},
\end{equation}
where $\mathcal{K}:=\{(u,v):u(\boldsymbol x)+v(\boldsymbol y)\le
c(\boldsymbol x,\boldsymbol y),\ \forall(\boldsymbol x,\boldsymbol y)\in
\Omega\times\Omega^*\}$. In fact, the dual variables $u$ and $v$ are related to each other by letting $v=u^c$ or $u=v^c$, where the superscript $c$ here denotes the c-transform,
\begin{equation}
    \left\{\begin{aligned}
    u^c(\boldsymbol{y})&:=\inf_{\boldsymbol{x}\in\Omega}
\{c(\boldsymbol{x},\boldsymbol{y})-u(\boldsymbol{x})\},\qquad \boldsymbol{y}\in \Omega^*,\\
v^c(\boldsymbol{x})&:= \inf_{\boldsymbol{y}\in\Omega^*}
\{c(\boldsymbol{x},\boldsymbol{y})-v(\boldsymbol{y})\},\qquad \boldsymbol{x}\in \Omega.
    \end{aligned}\right.
\end{equation}
A function $u$ or $v$ is called $c$-concave if there exists some $\tilde{v}$ on $\Omega^*$ or $\tilde{u}$ on $\Omega$ such that $u=\tilde{v}^c$ or $v=\tilde{u}^c$. Therefore, the dual functional \eqref{eq:dual-OT} can be reduced to a single potential functional depending on either $u$ or $v$,
\begin{equation}\label{reduced_dual_functionals}
\mathcal J(u):=\int_\Omega u\,f\,d\boldsymbol x+\int_{\Omega^*}u^c\,g\,d\boldsymbol y,\qquad\mathcal J_*(v):=\int_{\Omega^*}v\,g\,d\boldsymbol y+\int_\Omega v^c\,f\,d\boldsymbol x, 
\end{equation}
where optimizations can be performed over the family of $c$-concave functions \cite{Villani2009}.

\subsection{Reflector Map}
\noindent
For a sufficiently smooth $c$-concave function $u$, the optical map induced by the corresponding radial $\rho = e^u$ can be written as 
\begin{equation}\label{eq:explict_map}
    \boldsymbol{y}=T_u(\boldsymbol{x})=T(\boldsymbol{x},\nabla u(\boldsymbol x)),
\end{equation}
where the map $T: \Omega\times T_{\boldsymbol{x}}\Omega\to \Omega^*$ is given by \cite{wang2004design,Far_ref}
\begin{equation}
     T(\boldsymbol{x},\boldsymbol{p})=\frac{(|\boldsymbol{p}|^2-1)\boldsymbol{x}+2\boldsymbol{p}}{|\boldsymbol{p}|^2+1}.
\end{equation}
It can also be derived from the twist relation $\nabla_{\boldsymbol x}c(\boldsymbol x,T(\boldsymbol x,\boldsymbol p))=\boldsymbol p,\;\boldsymbol p\in T_{\boldsymbol{x}}\Omega$, which arises from differentiating the active contact condition induced by the c-transform. Here, both $\nabla$ and $\nabla_{\boldsymbol x}$ denote the intrinsic gradient on the spherical surface, and $T_{\boldsymbol{x}}\Omega$ is the tangent space of $\Omega$ at $\boldsymbol{x}$. The reverse map $T_{v}:\Omega^*\to\Omega$ has the same form as $T_u$ with the roles of $\boldsymbol{x}$ and $\boldsymbol{y}$ interchanged.

However, in the high-contrast regime, the $c$-concave function $u$ may fail to be differentiable and is in general only Lipschitz continuous on $\Omega$ \cite{Villani2009}. In this case, the contact map $T_u$ is defined through the $c$-contact relation,
\begin{equation}\label{eq:contact}
T_{u}(\boldsymbol{x})\in
    \{\boldsymbol y\in\Omega^*:
    u(\boldsymbol x)+u^c(\boldsymbol y)=c(\boldsymbol x,\boldsymbol y)\},\quad\boldsymbol{x}\in \Omega.
\end{equation}
Similarly, the reverse map $T_{v}:\Omega^*\to\Omega$ induced by the $c$-concave function $v$ on $\Omega^*$ is
\[T_{v}(\boldsymbol y)\in
    \{\boldsymbol x\in\Omega:
    v(\boldsymbol y)+v^c(\boldsymbol x)=c(\boldsymbol x,\boldsymbol y)\},\quad \boldsymbol{y}\in \Omega^*.\]

\subsection{Jacobian Residual}
After defining the map $T_u$ induced by the potential $u$, we can write the first variations of the functionals (\ref{reduced_dual_functionals}) \cite{BackForth,Far_ref}, 
\begin{equation}\label{eq:new_grads}
\left\{
\begin{aligned}
    \mathcal{J}^{\prime}(u)& =f - (T_{u^c})_{\#}g,\\
    \mathcal{J}_{*}^{\prime}(v)& =g - (T_{v^c})_{\#}f,
\end{aligned}
\right.
\end{equation}
which are the pushforward-measure residuals associated with mass conservation on the source and target sides, respectively.
 
If the $c$-concave function $u$ is sufficiently smooth, the pushforward-measure residual in \eqref{eq:new_grads} admits the following Jacobian representation through the change of variables \cite{Far_ref},
\begin{equation}\label{eq:jacobian-residual}
    \mathcal J^{\prime}[u]:=
    f(\boldsymbol x)
    -\frac{g(T_u(\boldsymbol x))}
    {\left|\det D^2_{\boldsymbol x\boldsymbol y}c
    (\boldsymbol x,T_u(\boldsymbol x))\right|}
    \det\!\left(
    D^2_{\boldsymbol x\boldsymbol x}c(\boldsymbol x,T_u(\boldsymbol x))
    -D^2u(\boldsymbol x)\right),
\end{equation}
where the Hessians and determinants are computed with respect to local orthonormal tangent frames. The regularity of the generalized Monge--Amp\`ere equation $\mathcal J^{\prime}[u]=0$ relies on positive density bounds \eqref{eq:density-bounds} and the uniform $c$-concavity of the potential $u$ \cite{ma2005regularity,TW2009SecondBoundary,CGH2008Reflector,loeper2011regularity},
\[D_{x x}^2 c\left(x, T_u(x)\right)-D^2 u(x) \geq \Theta I,\quad \Theta>0.\]

In the Sobolev gradient method \cite{Far_ref}, the functional $\mathcal J$ is optimized through the iteration 
\begin{equation}\label{eq:single-iteration}
    \Delta u^{n+1} = \Delta u^n -  \mathrm dt\, \mathcal{J}^{\prime}(u^n),
\end{equation}
where the choice of the fixed step size $\mathrm dt$ relies on the parameters $\lambda,\Lambda,\Theta$, etc.

In the high-contrast regime considered here, the target density $g$ may vanish on subsets of $\Omega^*$, and $\operatorname{supp}(g)$ may contain holes, corners, or multiple components, meaning that the regularity assumptions for Monge--Amp\`ere equations are no longer satisfied \cite{ma2005regularity,CGH2008Reflector,Loeper2009Regularity,loeper2011regularity}. A direct consequence is that the map may become discontinuous. This loss of regularity prevents the Jacobian residual from being evaluated by formula \eqref{eq:jacobian-residual}, and the original regularity-based argument no longer yields a positive admissible upper bound on the step size in \eqref{eq:single-iteration}. Consequently, once the iteration fails to preserve the $c$-concavity of the potential, the residual in \eqref{eq:new_grads} ceases to be a valid gradient of the reduced dual functional. Thus the ascent structure of \eqref{eq:single-iteration} is lost, and strict ascent or convergence can no longer be expected from the original argument \cite{Far_ref}.

\subsection{Measure-Based Bidirectional Framework}
\label{subsec:measure-framework}
To retain the dual ascent structure without assuming a smooth transport map, we evaluate the residuals in \eqref{eq:new_grads} directly from pushforward measures. Following the back-and-forth idea \cite{BackForth}, we alternate Sobolev updates on the source and target sides with $c$-transforms:
\begin{equation}\label{eq:bidirectional}
\left\{\begin{aligned}
\Delta \tilde{u}^{n+1}&=\Delta u^n-\mathrm{d}t \cdot \mathcal{J}^{\prime}\left(u^n\right),\\
v^{n} &= \left(\tilde{u}^{n+1}\right)^{c},\\
\Delta \tilde{v}^{n+1}&=\Delta v^n-\mathrm{d}t \cdot\mathcal{J}_{*}^{\prime}\left(v^n\right),\\
u^{n+1} &= \left(\tilde{v}^{n+1}\right)^{c},
\end{aligned}\right.
\end{equation}
where $\Delta$ denotes the Laplace--Beltrami operator on the spherical surface. In scheme \eqref{eq:bidirectional}, the bidirectional iteration of $u$ and $v$ accounts for the residuals on both the source and target sides through the forward and reverse maps, avoiding the imbalance of a one-sided update. The $c$-transform, in turn, restores $c$-concavity after each Sobolev step, preserving the validity of the gradient and the ascent structure.

The main computational bottleneck is the repeated discrete $c$-transform. A direct search over all source--target pairs requires $O(N^2)$ operations for comparable mesh sizes. Section~\ref{sec:fast-ctransform} develops a localized search for this step, while Section~\ref{sec:bidirectional} describes the pushforward evaluation and Poisson updates.

\section{Fast Discrete \texorpdfstring{$c$}{c}-Transform}\label{sec:fast-ctransform}

\subsection{Candidate Localization Principle}
We begin by describing the continuous localization principle for \eqref{eq:contact} that underlies the fast search. For the reflector cost \eqref{reflector_cost}, the $c$-transform cannot be reduced to a Legendre transform, in contrast to the Euclidean quadratic-cost setting where FFT-type accelerations are available \cite{BackForth}.

To construct the acceleration strategy, we exploit the connection between the contact relation \eqref{eq:contact} and the explicit transport formula \eqref{eq:explict_map} in the smooth regime. If $u$ is smooth and an interior point $\boldsymbol x$ realizes the minimum $u^c(\boldsymbol y)=c(\boldsymbol x,\boldsymbol y)-u(\boldsymbol x)$, then the first-order optimality condition gives \eqref{eq:explict_map}. Thus, apart from boundary contacts, a target point $\boldsymbol y$ can only be matched with source points whose predicted transport images lie near $\boldsymbol y$.

For nonsmooth potentials, the same localization principle is interpreted through the Clarke generalized subdifferential $\partial u(\boldsymbol x)\subset T_{\boldsymbol x}\Omega$ \cite{Clarke1990}. At differentiability points, this reduces to $\partial u(\boldsymbol x)=\{\nabla u(\boldsymbol x)\}$, while in the convex case it agrees with the usual convex subdifferential. Hence an interior contact is generated either by 
\begin{equation}\label{eq:diff-contact}
\begin{cases}
\boldsymbol y=T(\boldsymbol x,\nabla u(\boldsymbol x)),
& \text{if }u\text{ is differentiable at }\boldsymbol x,\\[0.4em]
\boldsymbol y=T(\boldsymbol x,\boldsymbol g),\quad 
\boldsymbol g\in\partial u(\boldsymbol x),
& \text{if }u\text{ is non-differentiable at }\boldsymbol x.
\end{cases}
\end{equation}
This provides a localized candidate set for the contact search, rather than requiring a global search over all source-target pairs in \eqref{eq:contact}.

The discrete algorithm below approximates this local generalized-gradient set by a one-ring gradient hull in the finite element discretization. Boundary contacts require separate consideration and are discussed in Section \ref{subsec:boundary-handling}.

\subsection{Finite Element Setup and Gradient Reconstruction}\label{sec:fem}
In the finite element setting, we discretize the source and target domains by curved quadratic meshes, denoted by $\mathcal T_H$ and $\mathcal T_H^*$, with nodal sets $\mathcal N$ and $\mathcal N^*$, respectively. We denote the corresponding source and target nodes by ${\boldsymbol{x}_p},\,p\in\mathcal N$ and ${\boldsymbol{y}_q},\,q\in\mathcal N^*$, and write $N:=\#\mathcal N$ and $M:=\#\mathcal N^*$. 

On $\mathcal T_H$, we use the quadratic Lagrange finite element space
\begin{equation}\label{quad_space}
V_H=\left\{v \in C^0:\; v|_K \in P_2(K),\; \forall K \in \mathcal{T}_H\right\}.
\end{equation}
For $u_H\in V_H$, write $u_H(\boldsymbol{x})=\sum_{p\in\mathcal N}u_H(\boldsymbol{x}_p)\phi_p(\boldsymbol{x})$, where $\phi_p$ denotes the quadratic nodal basis function. The target potential $v_H$ is represented analogously on $\mathcal T_H^*$.

To compute the gradient of $u_H$ at each node, we introduce an associated refined linear mesh $\mathcal{T}_h$. Each curved quadratic surface element $K\in\mathcal{T}_H$ is subdivided into four planar triangular faces using its three vertex nodes and three mid-edge nodes. Thus $\mathcal{T}_h$ is obtained directly from the nodal structure of $\mathcal{T}_H$, and its connectivity is precomputed once and reused throughout the algorithm. The linear interpolant of $u_H\in V_H$ on $\mathcal{T}_h$ is defined by $u_h=\mathcal{I}_h u_H\in V_h$, where $V_h=\{v\in C^0:\;v|_k\in P_1(k),\;\forall k\in\mathcal{T}_h\}$.

Since $u_h$ is piecewise linear on $\mathcal{T}_h$, its gradient is constant on each element $k\in\mathcal{T}_h$, 
\[
\boldsymbol{g}_k := \nabla (u_h|_k).
\]
For each node $p\in \mathcal{N}$, we define the one-ring patch
\begin{equation}\label{patch_def}
\omega_p := \{k \in \mathcal{T}_h :\; p \text{ is a vertex of } k\}.
\end{equation}
Since the refined mesh is quasi-uniform in our computations, we construct the nodal gradient $\mathbf{g}_p$ at $\boldsymbol{x}_p$ by averaging the element gradients over the patch $\omega_p$, and then projecting it onto the tangent plane,
\begin{equation}\label{nodal_grad_recon}
\mathbf{g}_p := \frac{1}{\#\omega_p} \sum_{k \in \omega_p} \mathcal{P}_{p}\boldsymbol{g}_k,
\end{equation}
where $\mathcal{P}_{p}\boldsymbol{\xi}:=
\boldsymbol{\xi}-(\boldsymbol{\xi}\cdot \boldsymbol{x}_p)\boldsymbol{x}_p$
is the orthogonal projection onto the tangent space at $\boldsymbol{x}_p$. An area-weighted average can be used in the same way for highly nonuniform meshes.

The target side is treated analogously. The reconstructed gradients are used for node classification, subgradient sampling, and candidate generation in the fast $c$-transform.

\subsection{Discrete \texorpdfstring{$c$}{c}-Transform and Candidate Generation}
For a given finite element potential $u_H\in V_H$, the nodal discrete
$c$-transform on the target mesh is defined by
\begin{equation}\label{disc_ct_u}
\begin{aligned}
u_{H}^{c}(\boldsymbol{y})&:= \sum_{q\in \mathcal{N}^{*}}v_q\,\phi_q^{*}(\boldsymbol{y}),\\
v_q &:= \min_{p\in \mathcal{N}}-\log(1-\boldsymbol{x}_{p}\cdot \boldsymbol{y}_{q}) - u_H(\boldsymbol{x}_p).
\end{aligned}
\end{equation}
Similarly, for a target potential $v_H$ on $\mathcal{T}_{H}^{*}$, the corresponding source-side transform is $v_{H}^{c}(\boldsymbol{x}):=\sum_{p\in \mathcal{N}}u_p\,\phi_p(\boldsymbol{x})$, with $u_p:=\min_{q\in \mathcal{N}^{*}}\{-\log(1-\boldsymbol{x}_{p}\cdot \boldsymbol{y}_{q})-v_H(\boldsymbol{y}_q)\}$.

A direct evaluation of these minima requires $O(MN)$ comparisons. The localization principle above suggests replacing the global search by a localized one guided by the generated transport images $T(\boldsymbol x_p,\boldsymbol g)$, where $\boldsymbol g$ is either the reconstructed nodal gradient or a sampled subgradient direction approximating an element of $\partial u(\boldsymbol x_p)$ in the continuous setting. These generated images are used to build candidate source-target pairs, so that each target node only tests nearby source candidates instead of all nodes in $\mathcal N$.

\subsection{Node Classification and Fast Candidate Search}

Motivated by the characterization above, we develop a fast algorithm for computing the discrete $c$-transform. The main idea is to recover local transport directions from the discrete potential and then restrict the minimization to a small set of relevant candidates.

For a prescribed tolerance $\epsilon>0$, we first classify the source nodal points into a numerically regular set $\mathcal{N}_{0}$ and a numerically singular set $\mathcal{N}_{1}$ by the threshold criterion
\begin{equation}\label{node_classification}
\begin{aligned}
\mathcal{N}_{0}&:=\left\{p\in \mathcal{N}:\; \max_{k\in \omega_{p}}\left\|\mathcal{P}_{p}\boldsymbol{g}_k - \mathbf{g}_p\right\|_2<\epsilon \right\},\\
 \mathcal{N}_{1} &:= \mathcal{N}\,\backslash \,\mathcal{N}_{0}.
\end{aligned}
\end{equation}
This classification is purely numerical and depends on the tolerance $\epsilon$, rather than a rigorous differentiability criterion.

For each regular node $p\in \mathcal{N}_0$, we keep the single recovered gradient $\mathbf{g}_p$ in \eqref{nodal_grad_recon}. For each singular node $p\in \mathcal{N}_{1}$, we approximate the local subdifferential $\partial u(\boldsymbol x_p)$ by the one-ring gradient hull. More precisely, we project the element gradients $\boldsymbol{g}_k$ of $k\in \omega_p$ onto the tangent plane at $\boldsymbol{x}_p$, and cluster nearby element gradients using the same tolerance $\epsilon$. Denoting the resulting clusters of element indices by $\{\mathcal{C}_{m}^{p}\}_{m=1}^{M_p}$, each cluster satisfies
\begin{equation}\label{eq:criterion}
   \left\|\mathcal{P}_{p}\left(\boldsymbol{g}_{k} - \boldsymbol{g}_{k'}\right)\right\|_{2}<\epsilon,
\qquad \forall\, k,k'\in \mathcal{C}_{m}^{p}\subset \omega_p.
\end{equation}
We collect the centroids of each cluster $\mathcal{C}_{m}^{p}$,
\[\widehat{\mathcal{G}}_{p}:=
\left\{
\frac{1}{|\mathcal{C}_{m}^{p}|}
\sum_{k\in \mathcal{C}_{m}^{p}}\mathcal{P}_{p}\boldsymbol{g}_{k}
:\; m=1,\cdots,M_p
\right\},\]
and use them to build a discrete approximation $\mathbf{G}_{p}$ of the subdifferential $\partial u(\boldsymbol x_p)$,
\begin{equation}\label{subgradient_sampling}
\mathbf{G}_{p}:= \left\{\begin{aligned}
&\operatorname{Sample}\left(\operatorname{Convexhull}(\widehat{\mathcal{G}}_{p}),\, \epsilon\right), \quad && p\in \mathcal{N}_1,\\
&\qquad \{\mathbf{g}_p\}, \quad && p\in \mathcal{N}_0.
\end{aligned}\right.
\end{equation}
Here, $\operatorname{Sample}(\cdot,\epsilon)$ denotes barycentric sampling with spacing $\epsilon$, where $\epsilon$ is the same parameter as in \eqref{eq:criterion}. The sampled set $\mathbf G_p$ can provide an $O(h)$-net of the one-ring gradient hull 
\begin{equation}\label{eq:one-hull}
     S_p^h:=\operatorname{co}\{\mathcal P_p\boldsymbol g_k:\ k\in\omega_p\}.
\end{equation}
Namely, let $\epsilon=\gamma h$ for a fixed constant $\gamma$ independent of $h$. For every $\boldsymbol g\in S_p^h$ there exists $\boldsymbol g_{p,j}\in \mathbf G_p$ such that
\begin{equation}\label{gradient_covering}
    |\boldsymbol g-\boldsymbol g_{p,j}|\le C_{\rm samp}h,
\end{equation}
where $C_{\rm samp}$ depends on the sampling constant $\gamma$. This provides the covering estimate required for the localized search.

The resulting set $\mathbf{G}_p$ is used to generate candidate transport images. Let $J_p:=\#\mathbf G_p$. If $p\in \mathcal{N}_0$, $J_p=1$; if $p\in \mathcal{N}_1$, $J_p$ is the number of sampled subgradients. We compute transport points by
\begin{equation}\label{targets}
    \left\{\boldsymbol{y}_{p, j}= T(\boldsymbol{x}_p,\mathbf{g}_{p,j}):\; \mathbf{g}_{p,j}\in \mathbf{G}_p\right\}_{\substack{p=1, \ldots, N \\ j=1, \ldots, J_p}}.
\end{equation}
Ideally, for each $q\in \mathcal{N}^{*}$, we would like to construct a spherical candidate set
\begin{equation}\label{candidate_ideal}
\mathcal{B}_q^{\rm Nbr} :=\{p\in \mathcal{N}: \exists 1\leq j \leq J_p \text{ s.t. } \boldsymbol{y}_{p,j} \in B_r(\boldsymbol{y}_q)\}
\end{equation}
and then replace the global minimization in (\ref{disc_ct_u}) by the local search
\[
v_q := \min_{p\in \mathcal{B}_q^{\rm Nbr} }-\log(1-\boldsymbol{x}_{p}\cdot \boldsymbol{y}_{q}) - u_H(\boldsymbol{x}_p).
\]
The remaining task is therefore to realize this spherical neighborhood search efficiently.

\subsection{Stereographic Projection for Fast Search}\label{subsec:stereo-search}
We now describe the search step for the reflector cost \eqref{reflector_cost} and the spherical domains $\Omega$ and
$\Omega^{*}$. The projection introduced in this
subsection is used for fast candidate search.

The nodes $\{\boldsymbol{y}_q\}_{q\in \mathcal{N}^{*}}$ and the target points $\left\{\boldsymbol{y}_{p, j}\right\}_{p,j}$ in (\ref{targets}) lie on the sphere. To realize the neighborhood search in \eqref{candidate_ideal} efficiently, we map these points to a planar chart by the stereographic projection
\[
\Pi(\boldsymbol y)=
\left(
\frac{y_1}{1+y_3},
\frac{y_2}{1+y_3}
\right),
\qquad
\boldsymbol y=(y_1,y_2,y_3)\in \mathbb S_+^2 .
\]
We denote the projected points by
\[
\boldsymbol Y_{p,j}:=\Pi(\boldsymbol y_{p,j}),
\qquad
\boldsymbol Y_q:=\Pi(\boldsymbol y_q).
\]
Since the target region is compactly contained in this projection chart, $\Pi$ is bi-Lipschitz on the relevant set. Hence spherical neighborhoods of radius $r$ are mapped to planar neighborhoods of comparable size. We therefore choose a planar search radius $R\asymp r$, partition the plane into square cells of width comparable to $R$, and store the generated points $\boldsymbol Y_{p,j}$ by their cell indices. For each target node $\boldsymbol Y_q$, the candidate set is then obtained by inspecting only the nearby cells, instead of comparing $\boldsymbol Y_q$ with all generated transport points.

Since $\boldsymbol Y_{p,j}$ are stored in a uniform square grid with cell width $R$, for each query point $\boldsymbol Y_q$, we only need to inspect a fixed number of neighboring grid cells around the cell containing $\boldsymbol Y_q$, whose union is denoted by $\operatorname{Nbr}_R(\boldsymbol Y_q)$. The practical candidate set is then constructed by the bucket-based rule
\begin{equation}\label{eq:bucket}
    \mathcal{B}_q=\left\{p\in \mathcal{N}: \exists 1\leq j \leq J_p \text{ s.t. } \boldsymbol{Y}_{p,j} \in \operatorname{Nbr}_{R}(\boldsymbol{Y}_q)\right\}.
\end{equation}
In the theoretical discussion below, the radius is chosen at the mesh scale, $R=\kappa h$. Since the target region is fixed away from the singular point of the stereographic projection, this is equivalent to taking a spherical search radius of order $h$. We choose the grid width and the number of inspected neighboring cells so that, for every $q$,
\begin{equation}\label{bucket_covering}
    B(\Pi(\boldsymbol y_q),R)
    \subset \operatorname{Nbr}_R(\boldsymbol Y_q)
    \subset B(\Pi(\boldsymbol y_q),C_{\rm nbr}R),
\end{equation}
where $C_{\rm nbr}$ is independent of $h$ and $q$.
Thus $\mathcal B_q$ consists of the source node indices whose generated projected transport points fall into the inspected cells around $\boldsymbol Y_q$.

The resulting fast $c$-transform is summarized below.
\begin{algorithm}
    \caption{Fast $c$-transform}\label{F}
    \begin{algorithmic}  
    \STATE Given $\{\boldsymbol{y}_q\}_{q\in \mathcal{N}^{*}}$ on $\Omega^{*}$ and $u_{H}(\boldsymbol{x})$ on $\{\boldsymbol{x}_p\}_{p\in \mathcal{N}}\subset \Omega$.
    \STATE Set the parameters: $\epsilon=\gamma h$ and $R=\kappa h$.
    \STATE 1. Compute $\{\mathbf{g}_{p}\}_{p\in \mathcal{N}}$ and $\{\boldsymbol{g}_{k}\}_{k\in \mathcal{T}_{h}}$ of $u_H$.
    \STATE 2. Divide $\mathcal{N}$ into $\mathcal{N}_0$ and $\mathcal{N}_1$.
    \STATE 3. For $p\in \mathcal{N}_1$, cluster nearby $\boldsymbol{g}_k$ and compute $\widehat{\mathcal{G}}_p$ to obtain $\mathbf{G}_p$.
    \STATE 4. Generate the target points $\{\boldsymbol{y}_{p,j}\}_{p,j}$ from $\mathbf{G}_p$.
    \STATE 5. Project $\{\boldsymbol{y}_{q}\}_{q\in \mathcal{N}^{*}}$ and $\{\boldsymbol{y}_{p,j}\}_{p,j}$ onto the plane to obtain $\{\boldsymbol{Y}_q\}_{q\in \mathcal{N}^{*}}$ and $\{\boldsymbol{Y}_{p,j}\}_{p,j}$.%
    \STATE 6. Partition the plane into square cells of width $R$ and store the source-node indices according to the cells containing $\{\boldsymbol{Y}_{p,j}\}_{p,j}$.
    \STATE 7. For $q\in \mathcal{N}^{*}$, construct $\mathcal{B}_q$ from neighboring grid cells around $\boldsymbol{Y}_q$.
    \STATE 8. Compute $v_q=\min_{p\in \mathcal{B}_q}\{-\log(1-\boldsymbol{x}_{p}\cdot \boldsymbol{y}_{q})-u_H(\boldsymbol{x}_p)\}$.
    \end{algorithmic}  
\end{algorithm}

\subsection{High-Order Local Refinement}\label{High-Order Local Refinement}
In order to improve the smoothness of the transformed function, we build a local quadratic correction on top of Algorithm \ref{F}. For a fixed
$\boldsymbol{y}_q$, the fast discrete transform above computes the minimum only over nodal values. 
Equivalently, it can be viewed as minimizing the piecewise linear interpolant of the sampled objective values $-\log(1-\boldsymbol{x}\cdot \boldsymbol{y}_q)-u_H(\boldsymbol{x})$ on the refined linear mesh $\mathcal T_h$, so the minimizer over each simplex is attained at a vertex. This is efficient, but it also explains the piecewise artifacts that appear in the gradient of the transformed potential.

To reduce these artifacts, we introduce a local quadratic correction on the original $P_2$ mesh $\mathcal T_H$. For each element $K\in\mathcal T_H$, let
\begin{equation}\label{quad_form}
    F(\boldsymbol{\lambda};\,\boldsymbol{y}_q,K)
    =
    \sum_{i=1}^{6} f_{q,i}^{K}\hat{\phi}_i(\boldsymbol{\lambda}),
    \qquad
    f_{q,i}^{K}
    =
    -\log(1-\boldsymbol{y}_q\cdot\boldsymbol{x}_{\kappa_i})
    -u_H(\boldsymbol{x}_{\kappa_i}),
\end{equation}
where $\widehat K$ denotes the reference triangle, $\hat{\phi}_i$, $i=1,\ldots,6$, are the quadratic Lagrange basis functions on $\widehat K$, and $\kappa_i=\kappa_i(K)$ is the global index of the $i$-th local node of $K$.

For each target node $\boldsymbol y_q$, choose
\begin{equation}\label{min_index}
    p(q)
    \in
    \operatorname{argmin}_{p\in\mathcal B_q}
    \left\{
    -\log(1-\boldsymbol{x}_p\cdot\boldsymbol{y}_q)
    -u_H(\boldsymbol{x}_p)
    \right\}.
\end{equation}
The quadratic correction is restricted to the elements $\mathcal K_{p(q)}:=\{K\in\mathcal T_H:p(q)\in K\}$ incident to $p(q)$, and the refined transform is defined by
\begin{equation}\label{refined_trans}
    \begin{aligned}
        u_H^c(\boldsymbol y)
        &=\sum_{q\in\mathcal N^*}v_q^{new}\phi_q^*(\boldsymbol y),\\
        v_q^{new}
        &=
        \min_{K\in\mathcal K_{p(q)}}
        \min_{\boldsymbol\lambda\in\widehat K}
        F(\boldsymbol\lambda;\boldsymbol y_q,K).
    \end{aligned}
\end{equation}

We refer to the finite element functions on $\Omega$ and $\Omega^*$ produced by Algorithm \ref{F}, with or without the local refinement \eqref{refined_trans}, as numerically $c$-concave. This means that they are constructed to be compatible, at the discrete level, with the supporting relation $u(\boldsymbol x)+u^c(\boldsymbol y)=c(\boldsymbol x,\boldsymbol y)$ used throughout this paper.

\subsection{Boundary Handling}\label{subsec:boundary-handling}
The discussion of the $c$-transform in the preceding subsections relies on a common interior contact assumption. For each target node $\boldsymbol y_q$, the corresponding contact node $\boldsymbol x_{p(q)}$ selected by the discrete $c$-transform lies in $\operatorname{int}(\Omega)$. If $\boldsymbol x_{p(q)}\in\partial\Omega$, the optimality condition used in the localization argument is no longer valid, and the fast search strategy cannot be directly applied. Boundary contacts therefore require separate consideration through the transport boundary condition. Let
\begin{equation}\label{boundary_nodes}
\mathcal{N}_{b}:=\left\{p\in \mathcal{N}:\boldsymbol{x}_p \in \partial \Omega\right\},
\qquad
\mathcal{N}_{b}^{*}:=\left\{q\in \mathcal{N}^{*}:
\boldsymbol y_q \in \partial \Omega^*\right\},
\end{equation}
be the sets of boundary nodes of $\mathcal{T}_H$ and $\mathcal{T}_H^*$, respectively. With $p(q)$ denoting the minimizing source node selected by the discrete $c$-transform, we impose the following discrete boundary consistency condition:
\begin{equation}\label{boundary_corr}
\begin{aligned}
q\in \mathcal{N}_b^*  &\quad\Rightarrow  \quad&& p(q) \in \mathcal{N}_b, \\
p\in \mathcal{N}_b   &\quad\Rightarrow  \quad&&
\forall q\in \mathcal{N}^{*}\text{ such that } p=p(q),\quad  q\in\mathcal{N}_b^* .
\end{aligned}
\end{equation}
In other words, the boundary correspondence for the discrete $c$-transform is bidirectional. Boundary target nodes are matched only with boundary source nodes, and boundary source nodes do not generate interior target contacts.

The condition in \eqref{boundary_corr} may seem restrictive, but it is the discrete counterpart of the transport boundary condition \cite{FD&Newton_6,Far_ref}
\begin{equation}\label{boundary_map}
T_{u_H}(\partial \Omega) = \partial \Omega^*,
\qquad
T_{u_H^c}(\partial \Omega^*) = \partial \Omega .
\end{equation}
During the iterative process \eqref{eq:bidirectional}, we impose boundary data in the Poisson solves for the updated potentials so that \eqref{boundary_map} is preserved; see Section \ref{subsec:poisson-boundary}.

A direct enforcement of \eqref{boundary_corr} would require an exhaustive comparison over the boundary nodes. Under the support separation condition \eqref{eq:dis_condition}, however, the singular interfaces remain away from the boundary and the transport map remains locally regular near $\partial\Omega$ and $\partial\Omega^*$. The boundary correspondence can therefore be determined from the differentiable transport condition \eqref{eq:diff-contact}, without a separate exhaustive boundary search.

\subsection{Discrete Uniform Exactness and Linear Complexity}
\label{subsec:discrete-uniform-exactness}

The following results are stated directly at the finite element level. Under the assumptions below and the reflector cost \eqref{reflector_cost}, we prove that the fast $c$-transform produced by Algorithm \ref{F} gives the same results as the standard discrete $c$-transform \eqref{disc_ct_u}.

We work under the standing assumptions on $\Omega$ and $\Omega^*$ stated in Section \ref{sec:problem}. Since only the portion of the transport supported by the prescribed measures is relevant to the computation, we work under the following assumptions.

\begin{assumption}\label{ass:disc_exact_main}
The refined source and target meshes $\mathcal{T}_h$ and $\mathcal{T}_h^*$ are shape-regular and quasi-uniform. For each target node $\boldsymbol y_q\in\mathcal N^*$, let
\begin{equation}
p^*(q)\in
\operatorname{argmin}_{p\in\mathcal N}
\left\{
c(\boldsymbol x_p,\boldsymbol y_q)-u_H(\boldsymbol x_p)
\right\}
\end{equation}
be an exact nodal minimizer, and assume that for some $d_*>0$, independent of $h$ and $q$,
\begin{equation}
\operatorname{dist}(\boldsymbol x_{p^*(q)},\partial\Omega)\ge d_*,
\qquad \forall q\in\mathcal N^* .
\end{equation}
\end{assumption}

\begin{proposition}[Interior discrete uniform exactness]\label{thm:disc_exact_main}
Under Assumption \ref{ass:disc_exact_main}, choose $\kappa$ sufficiently large. Then, for every target node $\boldsymbol y_q$ covered by the assumption,
$p^*(q)\in\mathcal B_q$. Consequently,
\[
    \min_{p\in\mathcal B_q}
    \{c(\boldsymbol x_p,\boldsymbol y_q)-u_H(\boldsymbol x_p)\}
    =
    \min_{p\in\mathcal N}
    \{c(\boldsymbol x_p,\boldsymbol y_q)-u_H(\boldsymbol x_p)\}.
\]
\end{proposition}

\begin{proof}
Let $p^*=p^*(q)$ and $\Phi_q(\boldsymbol x):=c(\boldsymbol x,\boldsymbol y_q)-u_h(\boldsymbol x)$.
First, nodal minimality gives a local minimality estimate on $\omega_{p^*}$. For any $k\in\omega_{p^*}$ and any vertex
$\boldsymbol z$ of $k$,
\[
    c(\boldsymbol z,\boldsymbol y_q)-u_h(\boldsymbol z)
    =
    c(\boldsymbol z,\boldsymbol y_q)-u_H(\boldsymbol z)
    \ge
    \Phi_q(\boldsymbol x_{p^*}).
\]
Let $I_kc(\cdot,\boldsymbol y_q)$ be the affine nodal interpolant of $c(\cdot,\boldsymbol y_q)$ on $k$. At each vertex of $k$ one has $I_kc=c$,
so the nodal values of $I_kc-u_h$ are bounded below by
$\Phi_q(\boldsymbol x_{p^*})$. Since an affine function on a simplex attains
its minimum at a vertex,
\[
    I_kc(\boldsymbol x,\boldsymbol y_q)-u_h(\boldsymbol x)
    \ge
    \Phi_q(\boldsymbol x_{p^*}),
    \qquad \forall\boldsymbol x\in k.
\]
The $C^2$ interpolation estimate gives $|c(\boldsymbol x,\boldsymbol y_q)-I_kc(\boldsymbol x,\boldsymbol y_q)|\le C_Ih^2$ uniformly in $q$, and therefore
\begin{equation}\label{eq:main_local_almost_min}
    \Phi_q(\boldsymbol x)
    \ge
    \Phi_q(\boldsymbol x_{p^*})-C_Ih^2,
    \qquad \boldsymbol x\in\omega_{p^*}.
\end{equation}

Let $\boldsymbol a_q=\nabla_{\boldsymbol x}c(\boldsymbol x_{p^*},\boldsymbol y_q)$. For the set \eqref{eq:one-hull}, we claim that
\begin{equation}\label{eq:main_discrete_contact}
d:=\operatorname{dist}(\boldsymbol a_q,S_{p^*}^h)\le C_{\rm con}h.
\end{equation}
If $d>0$, let $\boldsymbol s_0$ be the Euclidean projection of $\boldsymbol a_q$ onto $S_{p^*}^h$ and set $\boldsymbol e=\frac{\boldsymbol s_0-\boldsymbol a_q}{d}$. Then
\begin{equation}\label{eq:proof_in1}
    (\boldsymbol a_q-\boldsymbol s)\cdot\boldsymbol e\le -d, \qquad\forall s\in S_{p^*}^h.
\end{equation}
By the standard geometry of shape-regular quasi-uniform triangulations \cite{BrennerScott2008}, there exists $r_*>0$, independent of $h$ and $p^*$, such that for every unit tangent direction $\boldsymbol e$, some element $k_{\boldsymbol e}\in\omega_{p^*}$ contains, in the local mesh coordinates, the segment $\boldsymbol X(t)=\boldsymbol x_{p^*}+t\boldsymbol e$, $0\le t\le r_*h$. Let $\boldsymbol s_{k_{\boldsymbol e}}:=\mathcal P_{p^*}\boldsymbol g_{k_{\boldsymbol e}}\in S_{p^*}^h$.
Since $u_h$ is affine on the refined linear element $k_{\boldsymbol e}$, along this segment we have
\begin{equation}\label{proof:in2}
    u_h(\boldsymbol X(t))-u_h(\boldsymbol X(0))
=t\boldsymbol s_{k_{\boldsymbol e}}\cdot\boldsymbol e,
\qquad 0\le t\le r_*h,
\end{equation}
On the other hand, Taylor's theorem gives
\begin{equation}\label{proof:in3}
     c(\boldsymbol X(t),\boldsymbol y_q)
    -c(\boldsymbol X(0),\boldsymbol y_q)
    =
    t\,\boldsymbol a_q\cdot\boldsymbol e+O(t^2).
\end{equation}
Combining \eqref{proof:in2} and \eqref{proof:in3} with \eqref{eq:proof_in1}, we obtain
\[
    \Phi_q(\boldsymbol X(t))
    -\Phi_q(\boldsymbol X(0))
    \le -td+Ct^2.
\]
Taking $t=r_*h$ and using \eqref{eq:main_local_almost_min} yields $-C_Ih^2\le -r_*hd+Cr_*^2h^2$, hence \eqref{eq:main_discrete_contact}. By \eqref{eq:main_discrete_contact} and \eqref{gradient_covering}, we can choose $\boldsymbol s_q\in S_{p^*}^h$ and
$\boldsymbol g_{p^*,j}\in\mathbf G_{p^*}$ such that
\[
    |\boldsymbol s_q-\boldsymbol a_q|\le C_{\rm con}h,\qquad
    |\boldsymbol g_{p^*,j}-\boldsymbol s_q|\le C_{\rm samp}h.
\]
Since $T(\boldsymbol x_{p^*},\boldsymbol a_q)=\boldsymbol y_q$ and
$\Pi\circ T$ is Lipschitz,
\[
    |\boldsymbol Y_{p^*,j}-\Pi(\boldsymbol y_q)|
    \le C_T(C_{\rm con}+C_{\rm samp})h.
\]
Taking $\kappa\ge C_T(C_{\rm con}+C_{\rm samp})$ gives
$\boldsymbol Y_{p^*,j}\in B(\Pi(\boldsymbol y_q),R)\subset
\operatorname{Nbr}_R(\boldsymbol Y_q)$, and therefore $p^*(q)\in\mathcal B_q$. Since $p^*(q)$ is a global nodal minimizer, the restricted and full minima agree.

\end{proof}

In general, high-contrast illumination targets consist of finitely many regions or inclusions, rather than an infinitely oscillatory or fractal interface
structure. Therefore, the singular nodes are expected to concentrate along finitely many curves and near isolated corners or junctions. This motivates the following assumption.

\begin{assumption}\label{ass:main_linear_complexity}
Let $N:=\#\mathcal N\simeq h^{-2}$, and $\mathcal N=\mathcal N_0\cup\mathcal N_1$ be the numerical classification defined in \eqref{node_classification}. The singular set is further decomposed according to the affine dimension of the one-ring gradient hull \eqref{eq:one-hull},
\begin{equation}\label{eq:decomp}
\mathcal N_1=\mathcal N_1^{(1)}\cup\mathcal N_1^{(2)},\qquad \mathcal N_1^{(d)}:=\left\{p\in\mathcal N_1:\dim\operatorname{aff}(S_p^h)=d
\right\},\quad d=1,2.
\end{equation}
Here $\operatorname{aff}(S_p^h)$ denotes the affine span of $S_p^h$. The following bounds are assumed:
\begin{equation}
    \#\mathcal N_1^{(1)}\le Ch^{-1},\qquad
\#\mathcal N_1^{(2)}\le C,
\end{equation}
where $C$ is a constant independent of $h$.
\end{assumption}

Recall the notation $J_p:=\#\mathbf G_p$, where $\mathbf G_p$ is the sampled set constructed in \eqref{subgradient_sampling}, and let $\mathcal B_q$ be the local candidate set defined in \eqref{eq:bucket}. The resulting complexity estimate is as follows.

\begin{proposition}[Linear algorithmic complexity]\label{prop:main_linear_complexity}
Under Assumptions \ref{ass:disc_exact_main} and
\ref{ass:main_linear_complexity}, and with the sampling spacing $\epsilon=\gamma h$ for a fixed constant $\gamma>0$ independent of $h$, the candidate generation and the discrete search step satisfy the linear estimates
\[
\sum_{p\in\mathcal N}J_p=O(N),
\qquad
\sum_{q\in\mathcal N^*}\#\mathcal B_q=O(N).
\]
Consequently, the fast $c$-transform algorithm takes $O(N)$ operations.
\end{proposition}

\begin{proof}
For regular nodes $p\in\mathcal N_0$, only one recovered gradient \eqref{nodal_grad_recon} is used, so $J_p=1$. For singular nodes, the sampled set $\mathbf G_p$ is obtained with spacing $O(h)$ on the gradient hull $S_p^h$. Since the relevant gradients remain in a bounded set, a one-dimensional hull contributes $O(h^{-1})$ samples and a two-dimensional hull contributes $O(h^{-2})$ samples. Therefore,
\[
\sum_{p\in\mathcal N}J_p
=\sum_{p\in\mathcal N_0}J_p+\sum_{p\in\mathcal N_1^{(1)}}J_p+\sum_{p\in\mathcal N_1^{(2)}}J_p  
\le\#\mathcal N_0+Ch^{-1}\#\mathcal N_1^{(1)}
+Ch^{-2}\#\mathcal N_1^{(2)}.
\]
Using Assumption \ref{ass:main_linear_complexity}, we obtain
\begin{equation}\label{eq:es1}
\sum_{p\in\mathcal N}J_p
\le C\#\mathcal N + Ch^{-2}+Ch^{-2}\le CN.
\end{equation}

By the definition of $\mathcal B_q$ and the outer bucket inclusion \eqref{bucket_covering},
\[
\#\mathcal B_q
\le\#\left\{(p,j):|\boldsymbol Y_{p,j}-\boldsymbol Y_q|\le C_{\rm nbr}R
\right\}.
\]
Summing over $q$ and exchanging the order of summation gives
\begin{equation}\label{eq:es2}
    \sum_{q\in\mathcal N^*}\#\mathcal B_q
\le\sum_{p\in\mathcal N}\sum_{j=1}^{J_p}
\#\left\{q\in\mathcal N^*:|\boldsymbol Y_{p,j}-\boldsymbol Y_q|\le C_{\rm nbr}R
\right\}.
\end{equation}

By the quasi-uniformity of the target mesh, the target nodes are separated on the sphere, 
$d_{\mathbb S^2}(\boldsymbol y_q,\boldsymbol y_{q'})\ge C_{\text{Sep}}h$ for any $q\neq q'$.
Since the projection $\Pi$ is bi-Lipschitz on the compact target region, there exists $C_\Pi>0$ such that $|\Pi(\boldsymbol y_q)-\Pi(\boldsymbol y_{q'})|\ge C_\Pi d_{\mathbb S^2}(\boldsymbol y_q,\boldsymbol y_{q'})$. Therefore, the projected nodes satisfy
\begin{equation}\label{eq:Sep}
    |\boldsymbol Y_q-\boldsymbol Y_{q'}|\ge m_* h,\qquad m_*:=C_{\text{Sep}}C_{\Pi}.
\end{equation}

For fixed $(p,j)$, set $\mathcal Q_{p,j}:=\{q\in\mathcal N^*:|\boldsymbol Y_{p,j}-\boldsymbol Y_q|\le C_{\rm nbr}R\}$.
Since $R=\kappa h$, all nodes in $\mathcal Q_{p,j}$ lie in
$B(\boldsymbol Y_{p,j},C_{\rm nbr}\kappa h)$. By the separation estimate
\eqref{eq:Sep}, the balls
$B(\boldsymbol Y_q,m_*h/2)$, $q\in\mathcal Q_{p,j}$, are pairwise disjoint
and contained in
$B(\boldsymbol Y_{p,j},C_{\rm nbr}\kappa h+m_*h/2)$.
Comparing their areas gives
\[
\#\mathcal Q_{p,j}
\le
\left(1+\frac{2C_{\rm nbr}\kappa}{m_*}\right)^2
=:C_\kappa,
\]
where $C_\kappa$ is independent of $h$, $p$, and $j$.

Returning to \eqref{eq:es2} and combining it with \eqref{eq:es1},
\[
\sum_{q\in\mathcal N^*}\#\mathcal B_q
\le\sum_{p\in\mathcal N}\sum_{j=1}^{J_p} C_\kappa
=C_\kappa\sum_{p\in\mathcal N}J_p=O(N).
\]
\end{proof}

\section{Pushforward Evaluation and Poisson Updates}\label{sec:bidirectional}
This section describes the implementation of the two remaining components of the bidirectional framework \eqref{eq:bidirectional}, namely the evaluation of pushforward-measure residuals and the finite element Poisson update.

\subsection{Pushforward Measure Evaluation}\label{subsec:pushforward}
We evaluate the pushforward measures in (\ref{eq:new_grads}) through a common elementwise framework. For a measurable set $U\subset\Omega$, the pushforward of the target measure by the reverse map $T_{u^c}:\Omega^*\to\Omega$ is defined by
\[
(T_{u^c})_\#g (U)
:=\int_{T_{u^c}^{-1}(U)}g(\boldsymbol y)\,\mathrm dS(\boldsymbol y).
\]
When the reconstructed maps $T_u$ and $T_{u^c}$ are consistent, the set $T_{u^c}^{-1}(U)\subset\Omega^*$ can be evaluated through $T_u(U)$. Thus, the average density assigned to $U$ can be approximated by
\begin{equation}\label{pushforward_cell}
\overline g_U:=\frac{1}{|U|}(T_{u^c})_{\#}g(U)
\approx\frac{1}{|U|}\int_{T_u(U)}g(\boldsymbol y)\,\mathrm dS(\boldsymbol y),
\end{equation}
where $|U|$ denotes the surface area of $U$. In the finite element discretization, $U$ is taken to be an element $k\in\mathcal T_h$.  

For the target density $g$ and a $c$-concave function $u_H\in V_H$, we first interpolate $u_H$ to the refined linear space to get $u_h = \mathcal{I}_h(u_H)$ and then compute the mass of $g$ over $T_{u_h}(k)$. Since this image is not a standard element of the target mesh, a direct integration on the sphere is inconvenient, especially when the target illumination is given in finite element form.

To efficiently query the target values at integration points, we introduce a central projection used only for mass evaluation. Let
$\widetilde P_d:=\{(Y_1,Y_2,d):(Y_1,Y_2)\in\mathbb R^2\}\subset\mathbb R^3$ be the plane at height $d>0$, and identify it with the coordinate plane $P_d\simeq\mathbb R^2$ through $\boldsymbol Y=(Y_1,Y_2)$. We define
\[\begin{aligned}
\pi_{d}:\qquad &\mathbb{S}_{+}^{2}&\longrightarrow &\qquad P_d\,,\\
&\boldsymbol{y}=(y_1, y_2, y_3)& \mapsto & \qquad \pi_{d}(\boldsymbol{y})= \left(d\frac{y_1}{y_3},\, d\frac{y_2}{y_3}\right).
\end{aligned}\]
This projection is used only for the mass integral and should be distinguished from the stereographic projection introduced earlier for fast candidate search.
Under this projection, the spherical target density is represented on the
far-field plane by the measure identity
\[
\int_A g(\boldsymbol y)\,\mathrm dS(\boldsymbol y)
=\int_{\pi_d(A)}\tilde g(\boldsymbol Y)\,\mathrm d\boldsymbol Y, \qquad A\subset\Omega^*.
\]
Here $\tilde g(\boldsymbol Y)=g(\pi_d^{-1}(\boldsymbol Y))\frac{d}{\left(|\boldsymbol Y|^2+d^2\right)^{3/2}}$. In the implementation, we store and query the equivalent planar density $\tilde g$. The map $\pi_d$ allows the spherical integral over $T_{u_h}(k)\subset\mathbb S_+^2$ to be written as a planar integral over $\pi_d(T_{u_h}(k))\subset P_d$,
\begin{equation}\label{e_g}
    \int_{T_{u_h}(k)}g(\boldsymbol{y})\,\mathrm{d} S(\boldsymbol{y})=\int_{\pi_d(T_{u_h}(k))}\tilde{g}(\boldsymbol{Y})\,\mathrm{d}\boldsymbol{Y}.
\end{equation}
At each quadrature point of the planar integral, the value of $\tilde g$ can then be queried efficiently on a regular square grid, as in Section \ref{subsec:stereo-search}.

To evaluate the mass of $\tilde{g}$ over $\pi_d(T_{u_h}(k))$, we use a nodal reconstruction of this image. For the refined linear function $u_h=\mathcal I_h u_H$, we compute the nodal gradients $\{\mathbf g_p\}_{p\in\mathcal N}$ by \eqref{nodal_grad_recon}. Each source node is mapped to a target point and then projected to $P_d$,
\[
\boldsymbol Y_{p,0}:=\pi_d(T(\boldsymbol x_p,\mathbf g_p)),
\qquad p\in\mathcal N.
\]
For each $k\in\mathcal T_h$, write $k=\operatorname{Triangle}(\boldsymbol x_{\kappa_1},\boldsymbol x_{\kappa_2},\boldsymbol x_{\kappa_3})$, where $\kappa_i=\kappa_i(k)$, $i=1,2,3$, are the corresponding global node indices. We approximate the projected image of $k$ by
\begin{equation}\label{tri_approx}
    \pi_d(T_{u_h}(k))\approx \; \widetilde{K}_k
    :=\operatorname{Triangle}
    (\boldsymbol Y_{\kappa_1,0},
     \boldsymbol Y_{\kappa_2,0},
     \boldsymbol Y_{\kappa_3,0})\subset P_d .
\end{equation}
Corresponding to \eqref{pushforward_cell}, the elementwise pushforward density on $k$ is then approximated by 
\begin{equation}\label{eq:push_approx}
    \left.(T_{u_h^c})_{\#}g\,\right|_{k}
    \approx\;
    g_k^h
    :=
    \frac{1}{|k|}
    \int_{\widetilde{K}_k}\tilde g(\boldsymbol Y)\,\mathrm d\boldsymbol Y,
\end{equation}
where $|k|$ denotes the surface area of the source element $k$. The integral over $\widetilde{K}_k$ is computed by Monte Carlo integration, with sample points generated in barycentric coordinates.

\begin{remark}[Nodal reconstruction of element images]
The triangle $\widetilde{K}_k$ in \eqref{tri_approx} approximates the projected image $\pi_d(T_{u_h}(k))$ from the transported nodal points. When the transport map is locally smooth, this is the standard affine approximation of the element image.

At non-differentiable points, the transport relation may be multivalued, so a single triangle cannot represent all possible image branches. However, the present reconstruction is used only to evaluate the transported mass against $g$. At the optimal OT matching, branches associated with zero-density target regions carry no target mass. Motivated by this structure, the nodal triangle is used as a practical approximation for elementwise mass evaluation during the iteration.
\end{remark}

To use this elementwise pushforward density in the finite element Poisson update, we recover a continuous nodal finite element function $f_{u_h}\in V_H$ from this elementwise pushforward density \eqref{eq:push_approx} by patch averaging,
\begin{equation}\label{pushforward_nodal}
f_{u_h}(\boldsymbol{x}):=\sum_{p\in\mathcal{N}} f_{u_h}(\boldsymbol{x}_p)\,\phi_p(\boldsymbol{x}),
\qquad
f_{u_h}(\boldsymbol{x}_p):=\frac{1}{\#\omega_p}
\sum_{k\in \omega_p} g_k^h .
\end{equation}
This defines a finite element approximation of $\left(T_{u^c}\right)_{\#}g$ in $V_H$.

The pushforward $(T_{v^c})_{\#}f$ is evaluated within the same elementwise framework. For the uniform source used in the numerical experiments below, $f$ is constant on its support, so the mass associated with a reconstructed image element reduces to the constant value of $f$ multiplied by the area of the portion lying inside $\operatorname{supp} f$. Thus, the two dual updates share the same pushforward mechanism, while the source-side integral admits a simpler geometric evaluation.

\subsection{Poisson Update and Boundary Data}\label{subsec:poisson-boundary}
Suppose that $u_H^n$ has been obtained at the $n$-th step, and set $u_h^n:=\mathcal I_h(u_H^n)$. Following Section \ref{subsec:pushforward}, we compute the finite element function $f_{u_h^n}\in V_H$ defined in \eqref{pushforward_nodal}. In this subsection, we denote it by $f_H^n:=f_{u_h^n}\in V_H$, with $f_H^n\approx(T_{(u_h^n)^c})_\#g$.
Here $f_H:=\mathcal I_H f\in V_H$ denotes the prescribed source density. We rescale $f_H^n$ so that it has the same total mass as $f_H$ by $r^n:=\bigl(\int_{\Omega}f_H(\boldsymbol x)\,\mathrm dS(\boldsymbol x)\bigr)\big/\bigl(\int_{\Omega}f_H^n(\boldsymbol x)\,\mathrm dS(\boldsymbol x)\bigr)$.
This normalization ensures that $r^n f_H^n$ and $f_H$ have the same total mass in the Poisson update.

For the update of $u_H^n$, we discretize the Poisson equation as $\Delta \tilde u_H^{n+1}=\Delta u_H^n-\mathrm{d}t\,\bigl(f_H-r^n f_H^n\bigr)$.
After integration by parts, the finite element weak formulation reads
\begin{equation}\label{poisson_weak_update}
\begin{aligned}
\left\langle\nabla \tau_H,\, \nabla \tilde u_H^{n+1}\right\rangle_{\Omega}
&=\left\langle\nabla \tau_H,\, \nabla u_H^n\right\rangle_{\Omega}
+\mathrm{d}t\,\left\langle \tau_H,\, f_H-r^n f_H^n\right\rangle_{\Omega} \\
&\quad+\left\langle \tau_H,\,(l_H^{n+1}-l_H^n)\cdot\boldsymbol n\right\rangle_{\partial \Omega},
\qquad \forall \tau_H\in V_H .
\end{aligned}
\end{equation}
Here $\boldsymbol n$ denotes the outward unit co-normal to $\partial\Omega$ within the spherical surface. The vector field $l_H^{n+1}$ provides the
prescribed boundary gradient data, whose normal component enters the Neumann condition, $\partial_{\boldsymbol n}\tilde u_H^{n+1}
=l_H^{n+1}\cdot \boldsymbol n$.

Specifically, the boundary data $l_H^{n+1}$ is constructed from the current boundary image in order to impose the boundary transport condition
\eqref{boundary_map} in the next update \cite{Far_ref}. To construct $l_H^{n+1}$, each current boundary image $T_{u_H^n}(\boldsymbol x_p)$ is projected onto the
target boundary $\partial\Omega^*$. The boundary gradient is then recovered from the contact relation as follows,
\begin{equation}\label{boundary_grad_data}
\begin{aligned}
    l_H^{n+1}(\boldsymbol{x})
    &:=\sum_{p\in \mathcal{N}_b}l_H^{n+1}(\boldsymbol{x}_p)\,\phi_p(\boldsymbol{x}),\\
    l_H^{n+1}(\boldsymbol{x}_p)
    &:=\nabla_{\boldsymbol{x}} c\left(
    \boldsymbol{x}_p,\,\operatorname{Proj}_{\partial \Omega^*}\left(T_{u_H^n}(\boldsymbol{x}_p)\right)
    \right).
\end{aligned}
\end{equation}
Here $\mathcal N_b$ is the set of boundary nodes defined in \eqref{boundary_nodes}. 

The same weak formulation is used for the two updates in the bidirectional iteration, with the residual and normal data replaced by the corresponding quantities on the current domain. 

The complete numerical implementation is summarized in Algorithm \ref{alg:bidirectional-solver}.

\begin{algorithm}[H]
    \caption{Bidirectional reflector iteration with measure residuals}
    \label{alg:bidirectional-solver}
    \begin{algorithmic}
    \STATE \textbf{Input:} $f_H,g_H,u_H^0,\mathrm dt,\epsilon,\kappa,
    \mathrm{tol},N_{\max}$.
    \STATE Precompute refined meshes, patches, boundary nodes, and projection grids.
    \STATE Initialize the boundary data $l_{u,H}^{0}$ and $l_{v,H}^{0}$.
    \FOR{$n=0,\ldots,N_{\max}-1$}
        \STATE $f_H^n\leftarrow \operatorname{PushForward}_{\Omega}(u_H^n,g_H)$,
        \quad
        $r_u^n\leftarrow
        {\int_{\Omega} f_H\,\mathrm dS}/
        {\int_{\Omega} f_H^n\,\mathrm dS}$.
        
        \STATE $l_{u,H}^{n+1}\leftarrow
        \operatorname{Boundary}_{\Omega}(u_H^n)$.
        
        \STATE $\tilde u_H^{n+1}\leftarrow
        \operatorname{Poisson}_{\Omega}
        (u_H^n,f_H-r_u^n f_H^n,
        l_{u,H}^{n+1}-l_{u,H}^{n})$.
        
        \STATE $v_H^n\leftarrow
        $c$\operatorname{-Transform}_{\Omega\to\Omega^*}
        (\tilde u_H^{n+1};\epsilon,\kappa)$.
        
        \STATE $g_H^n\leftarrow
        \operatorname{PushForward}_{\Omega^*}(v_H^n,f_H)$,
        \quad
        $r_v^n\leftarrow
        {\int_{\Omega^*} g_H\,\mathrm dS}/
        {\int_{\Omega^*} g_H^n\,\mathrm dS}$.
        
        \STATE $l_{v,H}^{n+1}\leftarrow
        \operatorname{Boundary}_{\Omega^*}(v_H^n)$.
        
        \STATE $\tilde v_H^{n+1}\leftarrow
        \operatorname{Poisson}_{\Omega^*}
        (v_H^n,g_H-r_v^n g_H^n,
        l_{v,H}^{n+1}-l_{v,H}^{n})$.
        
        \STATE $u_H^{n+1}\leftarrow
        $c$\operatorname{-Transform}_{\Omega^*\to\Omega}
        (\tilde v_H^{n+1};\epsilon,\kappa)$.
        
        \STATE $E^n\leftarrow
        \max\left\{
        \|f_H-r_u^n f_H^n\|_{L^1(\Omega)},
        \|g_H-r_v^n g_H^n\|_{L^1(\Omega^*)}
        \right\}$.
        
        \IF{$E^n<\mathrm{tol}$}
            \STATE \textbf{break}
        \ENDIF
    \ENDFOR
    \STATE \textbf{Output:} $u_H^{n+1}$, $v_H^n$, and
    $\Gamma_H^{n+1}=\{e^{u_H^{n+1}(\boldsymbol x)}
    \boldsymbol x:\boldsymbol x\in\Omega\}$.
    \end{algorithmic}
\end{algorithm}

\section{Numerical Examples}\label{sec:numerical-experiments}
\subsection{Experimental Setup}
Define the spherical caps $\mathbb{S}_{+\theta}:=\{\boldsymbol{x}\in\mathbb{S}^2:\boldsymbol{x}\cdot\boldsymbol{e}_z\ge\cos(\theta)\}$ and $\mathbb{S}_{-\theta}:=\{\boldsymbol{x}\in\mathbb{S}^2:\boldsymbol{x}\cdot\boldsymbol{e}_z\le-\cos(\theta)\}$ for $\theta\in(0,\pi/2)$, where $\boldsymbol{e}_z=(0,0,1)^T$. We take $\Omega=\mathbb{S}_{-\frac{3\pi}{8}}$ and $\Omega^*=\mathbb{S}_{+\frac{3\pi}{8}}$. The source density is uniform, $f(\boldsymbol{x})\equiv1$ on $\operatorname{supp}f=\mathbb{S}_{-\frac{\pi}{4}}$. For each target pattern, $g$ is normalized so that $\int_{\Omega}f\,dS=\int_{\Omega^*}g\,dS$.

Unless otherwise specified, all experiments use mesh level $100$ and $\mathrm{d}t=0.01$. Sections 5.5 and 5.6 vary the mesh density and time step, respectively. Here the mesh level denotes the number of subdivisions in the radial variable of the cap parametrization, and levels $40,\,80,\,100$ correspond to $4921$, $19441$, and $30301$ nodes, respectively. The iterative process is stopped when either a prescribed maximum of $80$ iterations is reached or the monitored discrete loss no longer decreases. The discrepancy between the prescribed and simulated irradiance distributions is assessed using the $H^{-1}$ and $L^2$ norms of the irradiance residual.

The following tests use zero-background target distributions, i.e., no positive lower bound $g\ge\epsilon>0$ is imposed on $\Omega^*$.

\subsection{Reconstruction of Binary Patterns}\label{ex:Binary}

In the first numerical test, we consider a binary target $g$ representing the letter ``A''. This target poses a challenge due to the sharp discontinuities between the illuminated region and the dark part of $\Omega^*$.

Fig.~\ref{fig:letterA} illustrates the numerical results.
Fig.~\ref{fig:letterA}(a) shows the prescribed target distribution $g$, while Fig.~\ref{fig:letterA}(b) displays the simulated irradiance $g_H$ obtained via Monte Carlo ray tracing. The results indicate that $g_H$ matches the target $g$ with high fidelity, effectively preserving the sharp edges and the topological features of the letter.

\begin{figure}[H]
    \centering
    \begin{subfigure}[t]{0.20\textwidth}
        \centering
        \includegraphics[width=\linewidth]{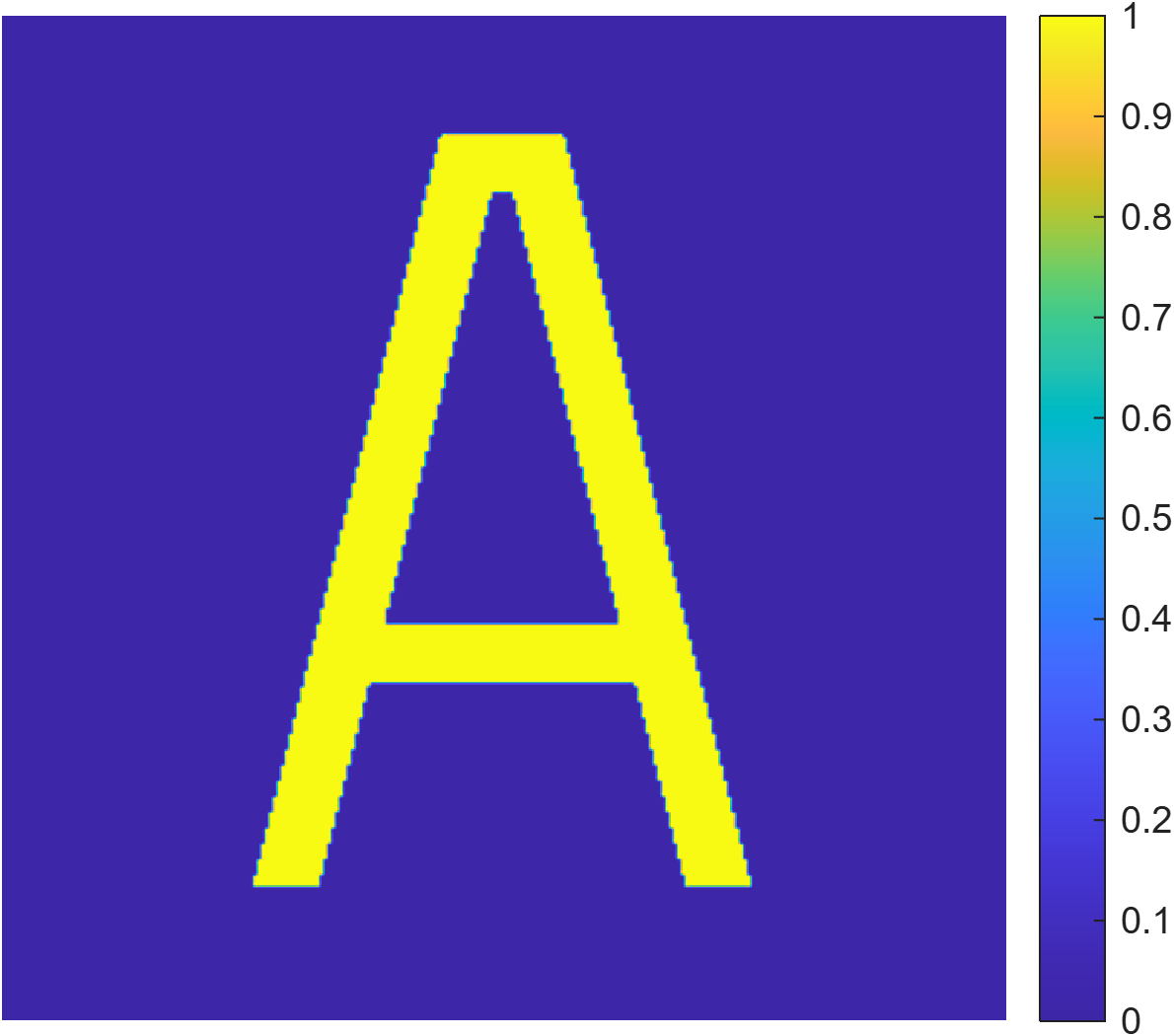}
        \caption{Target $g$}
    \end{subfigure}\hfill
    \begin{subfigure}[t]{0.20\textwidth}
        \centering
        \includegraphics[width=\linewidth]{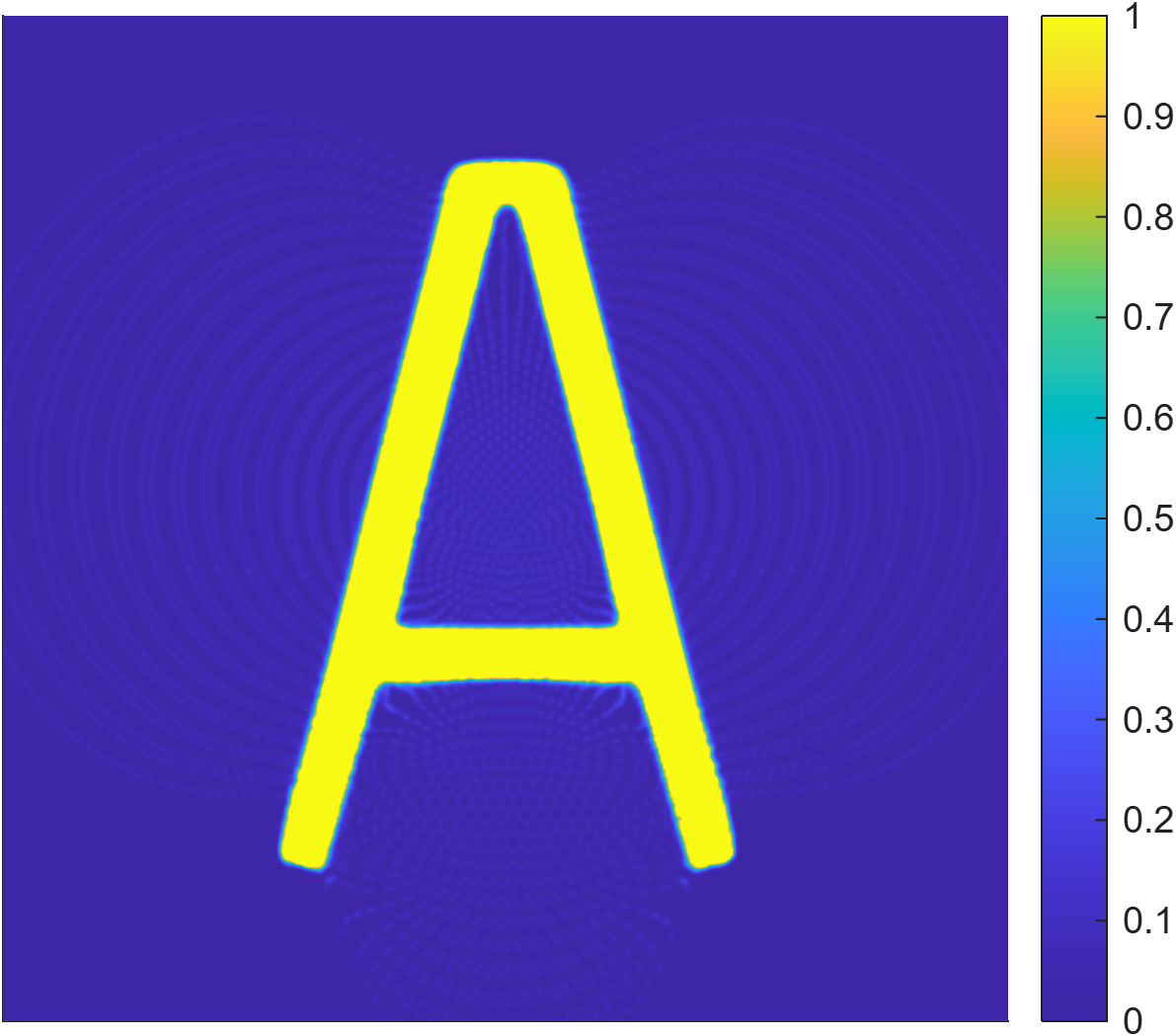}
        \caption{Ray tracing $g_H$}
    \end{subfigure}\hfill
    \begin{subfigure}[t]{0.20\textwidth}
        \centering
        \includegraphics[width=\linewidth]{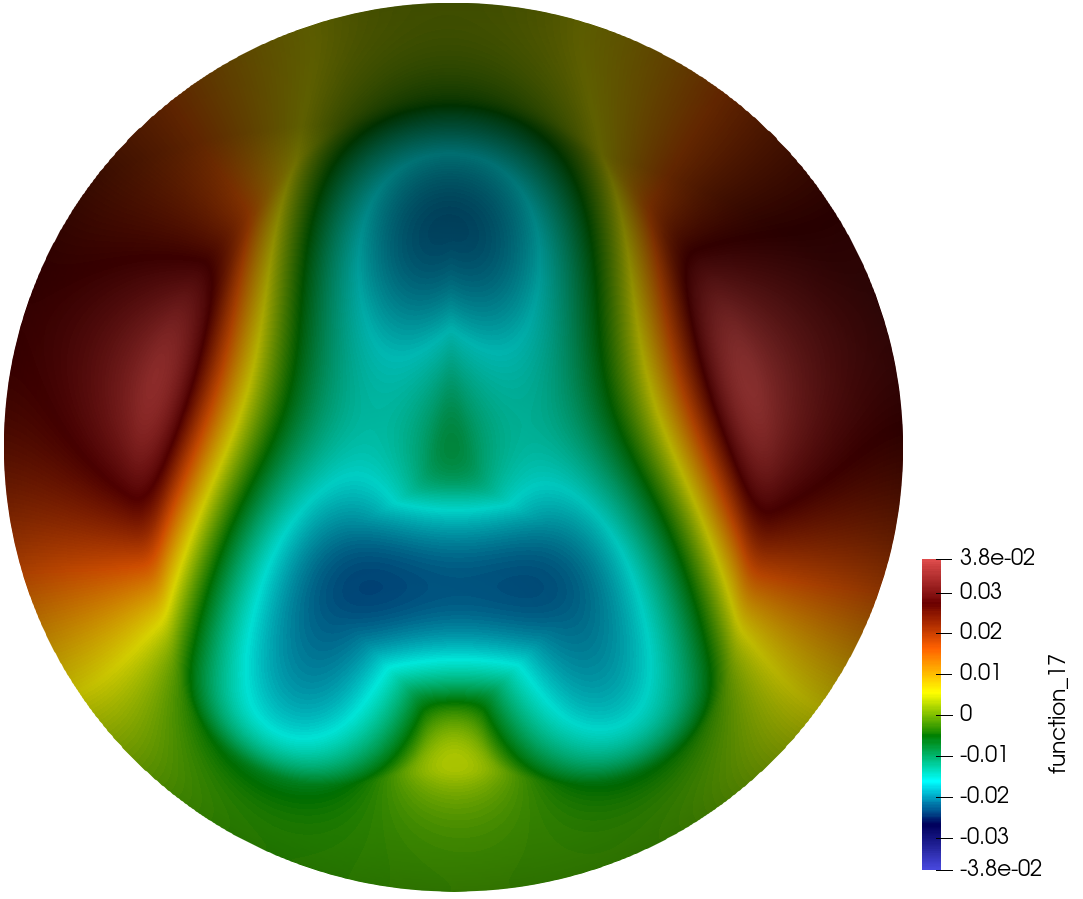}
        \caption{Potential $u_H$}
    \end{subfigure}\hfill
    \begin{subfigure}[t]{0.20\textwidth}
        \centering
        \includegraphics[width=\linewidth]{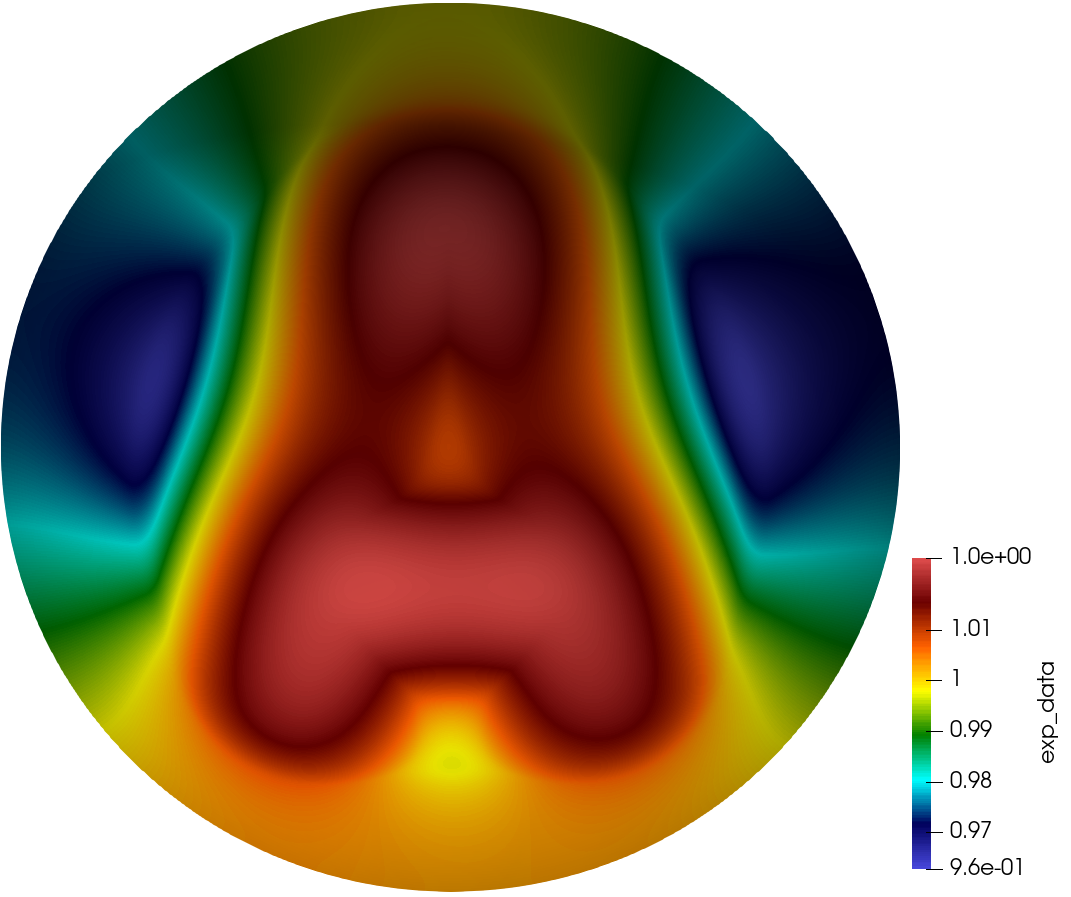}
        \caption{Radial function $\rho_H$}
    \end{subfigure}
    \captionsetup{width=0.9\textwidth}
    \caption{Results for the letter ``A'' target: prescribed target intensity,
    ray-tracing simulation, computed potential, and radial distance function
    $\rho_H=e^{u_H}$.}
    
    \label{fig:letterA}
\end{figure}

Fig.~\ref{fig:letterA}(c) shows the computed potential function $u_H$, and Fig.~\ref{fig:letterA}(d) shows the corresponding reflector surface represented by the radial function $\rho_H=e^{u_H}$. The reconstructed radial function
appears visually regular at the displayed resolution and does not exhibit visible high-frequency oscillations. This behavior is consistent with the expected Lipschitz-type regularity and indicates good numerical stability for this binary target.

\raggedbottom
\subsection{Convergence Analysis and Complex Patterns}

We next evaluate the performance and convergence of the proposed method using four representative targets: the letter ``A'', the text ``ZJU'' with multiple disconnected components, a binary ``Eagle'',
and a multi-level grayscale ``Eagle''.

Fig.~\ref{fig:convergence} displays the targets, the corresponding ray-tracing results, and the convergence curves. The ray-tracing results highlight two additional challenging scenarios:
\begin{itemize}
    \item \textbf{Disconnected components}: For the ``ZJU'' target, the simulation captures multiple disconnected illumination regions without visible topological distortion.
    \item \textbf{Continuous density variations}: For the grayscale ``Eagle'', the simulation reproduces varying intensity levels with small visible error, indicating that the measure-based iteration is not restricted to binary support constraints and can also handle spatially varying target intensities.
\end{itemize}

The convergence curves in Fig.~\ref{fig:convergence} show the $H^{-1}$ and $L^2$ residuals versus iteration. In all four cases, the residuals decay stably, with only small local fluctuations in some $L^2$ traces. The comparable behavior for ``ZJU'' and the grayscale ``Eagle'' indicates similar convergence characteristics for these complex targets.
\begin{figure}[H]
    \centering
    \begin{subfigure}[t]{0.20\textwidth}
        \centering
        \includegraphics[width=\linewidth]{target_A.png}
        \vspace{0.15em}
        \includegraphics[width=\linewidth]{simulation_A.png}
        \vspace{0.15em}
        \includegraphics[width=\linewidth]{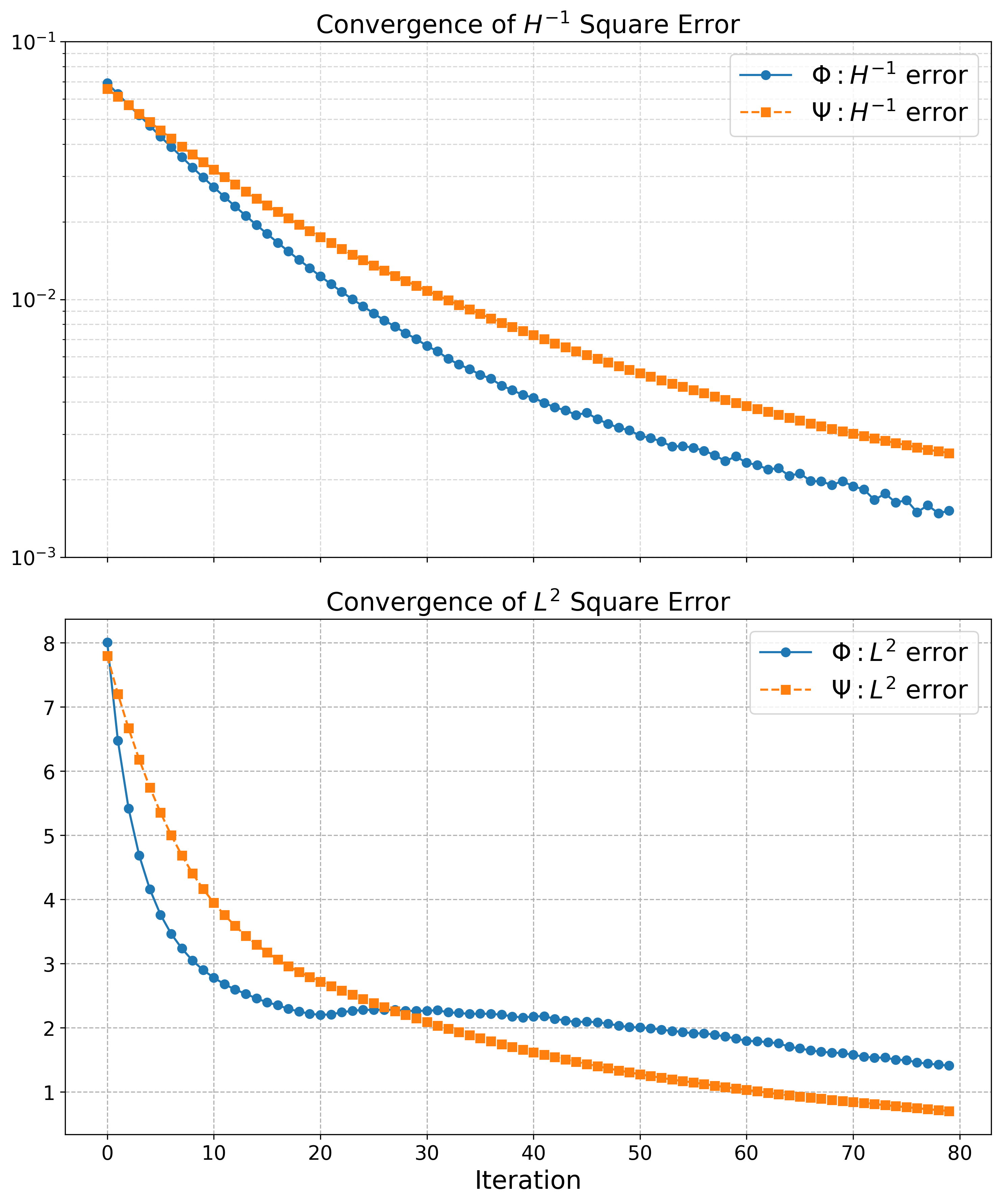}
        \caption{Letter ``A''}
    \end{subfigure}\hfill
    \begin{subfigure}[t]{0.20\textwidth}
        \centering
        \includegraphics[width=\linewidth]{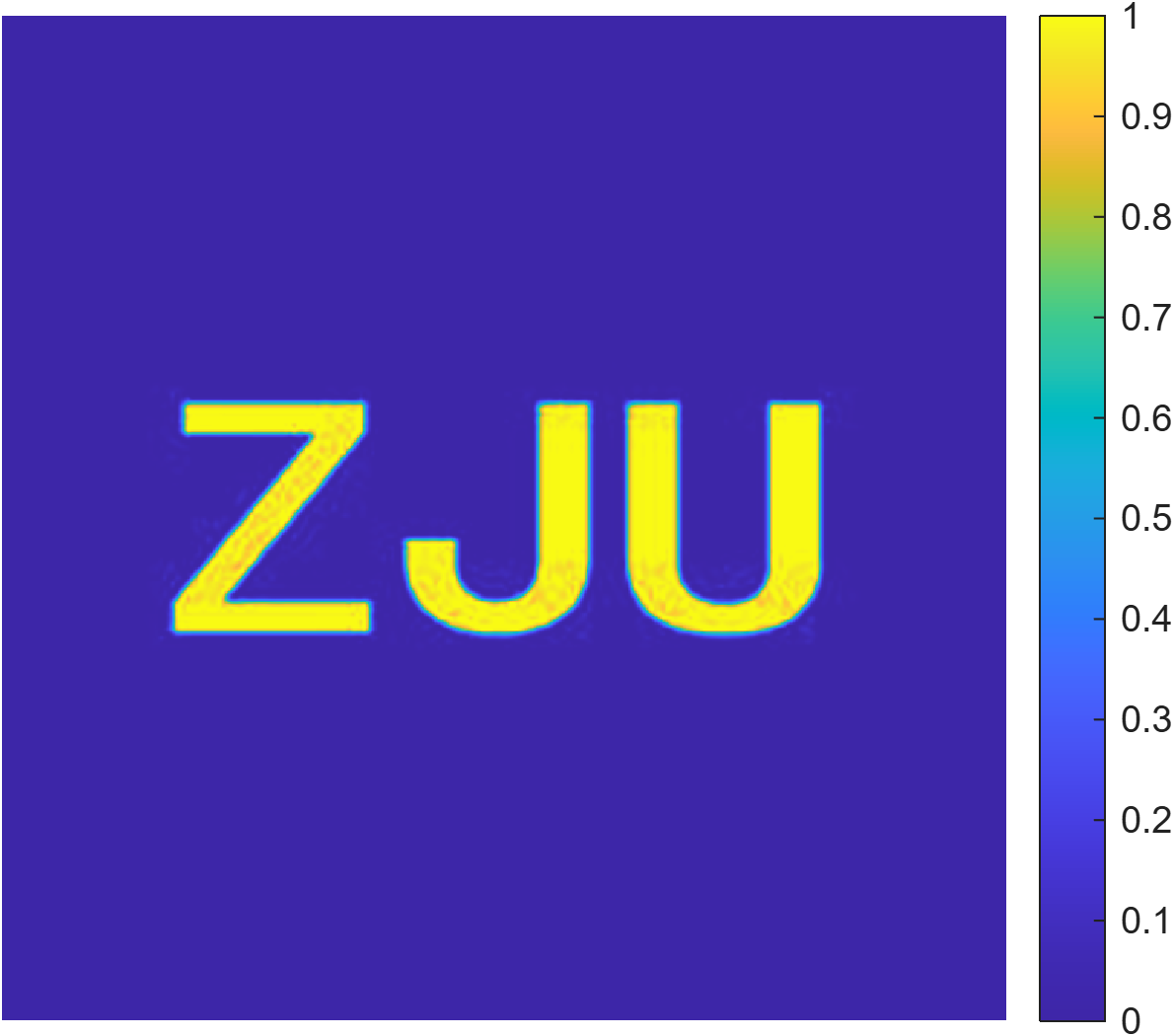}
        \vspace{0.15em}
        \includegraphics[width=\linewidth]{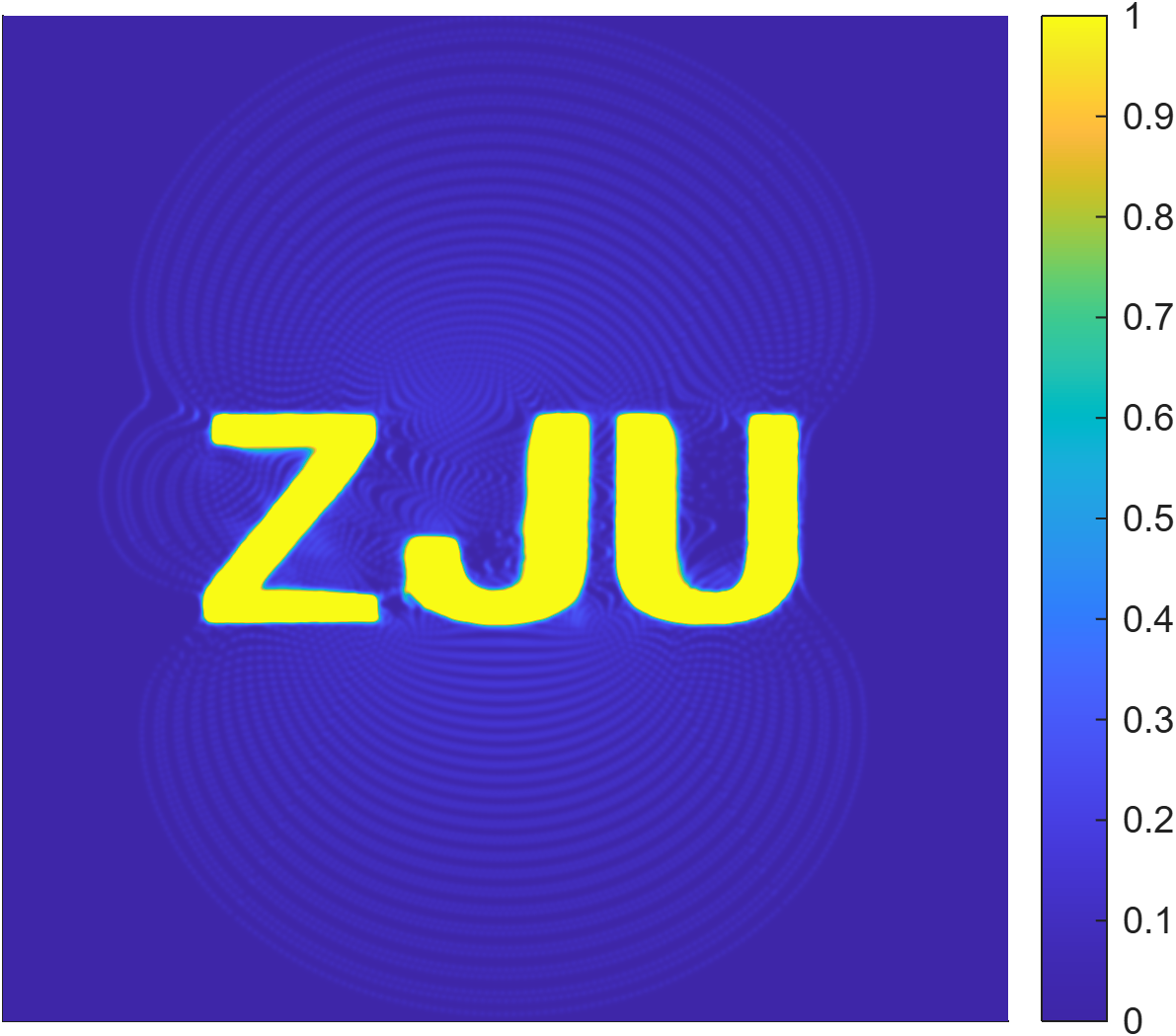}
        \vspace{0.15em}
        \includegraphics[width=\linewidth]{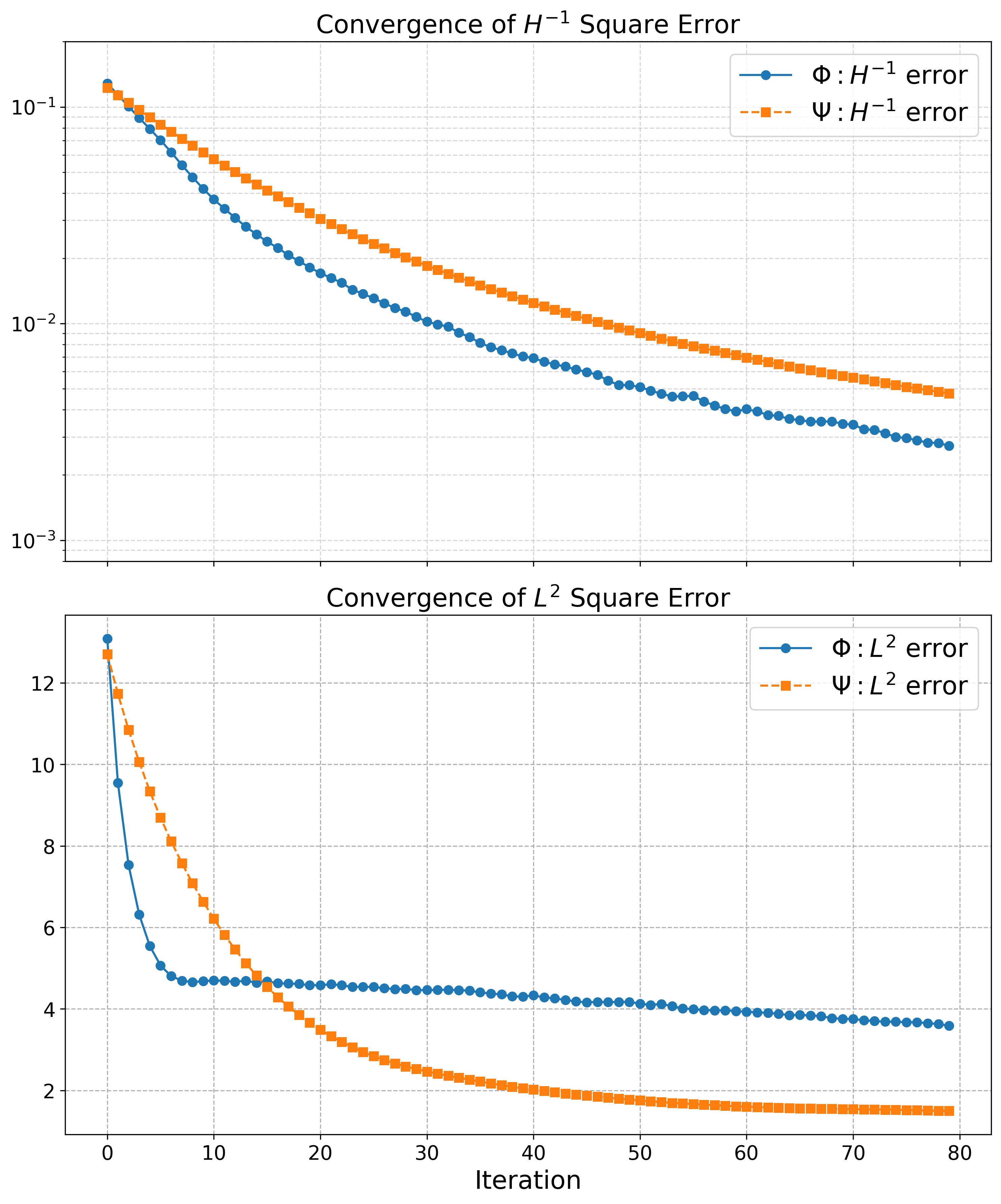}
        \caption{Text ``ZJU''}
    \end{subfigure}\hfill
    \begin{subfigure}[t]{0.20\textwidth}
        \centering
        \includegraphics[width=\linewidth]{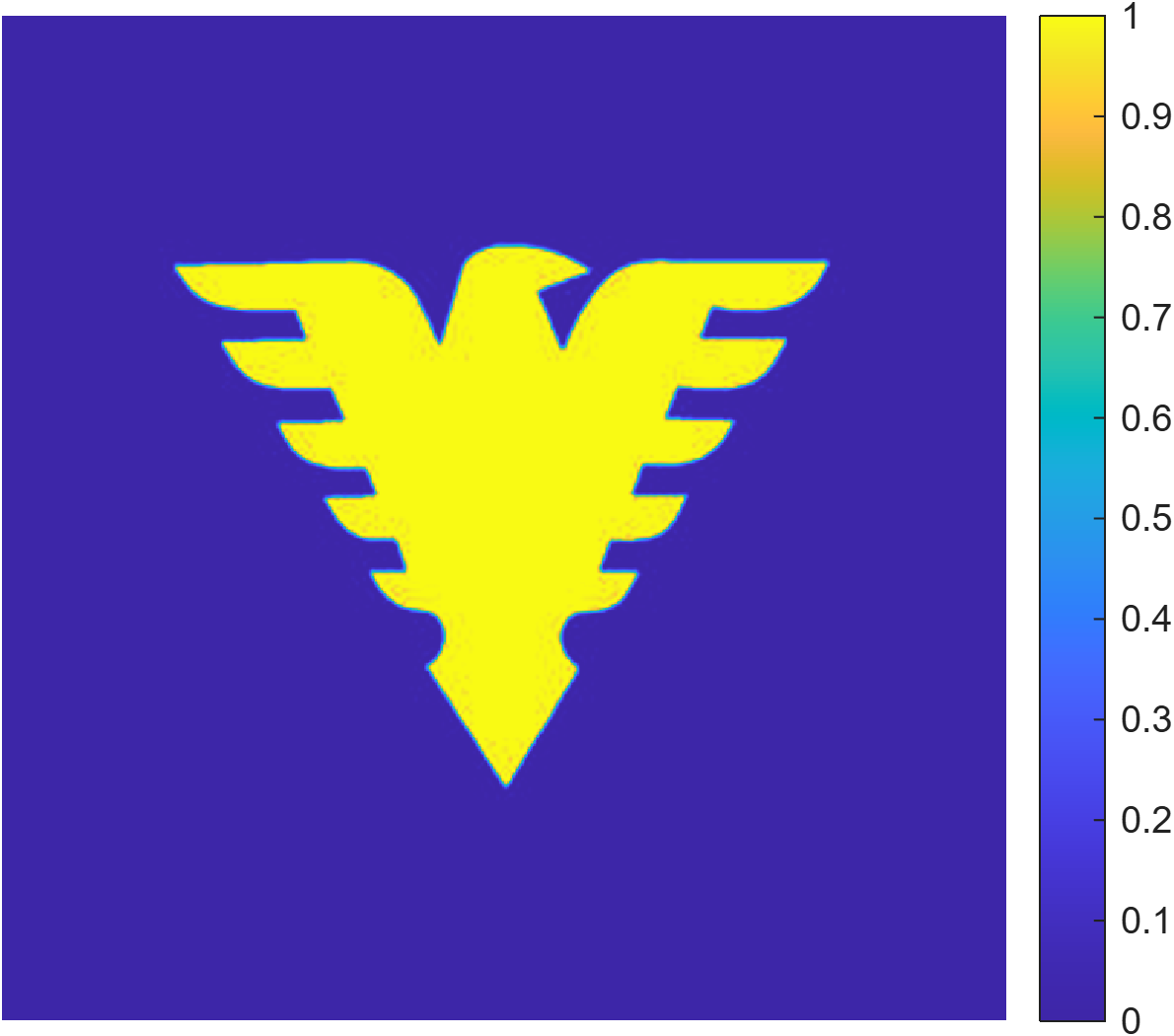}
        \vspace{0.15em}
        \includegraphics[width=\linewidth]{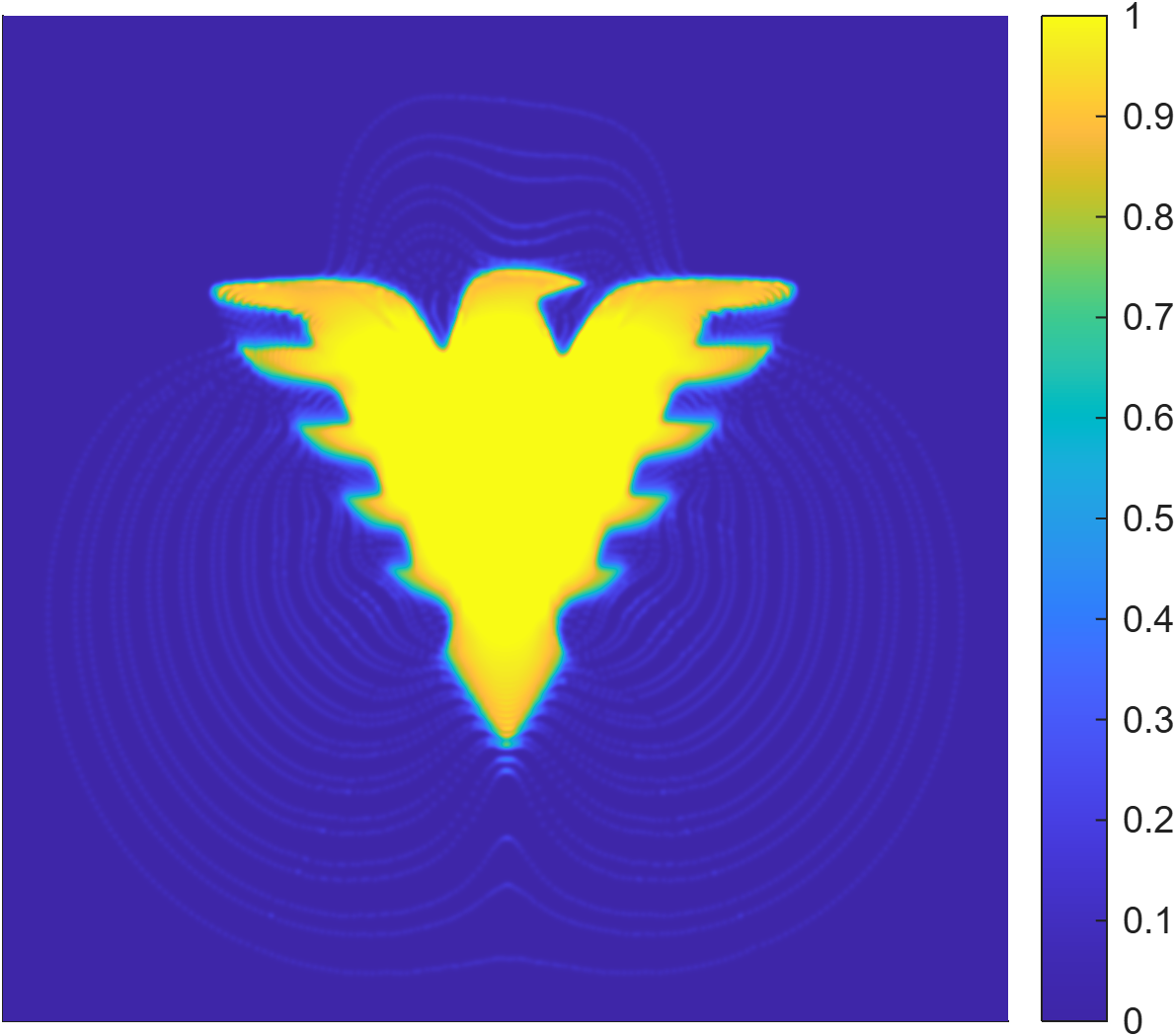}
        \vspace{0.15em}
        \includegraphics[width=\linewidth]{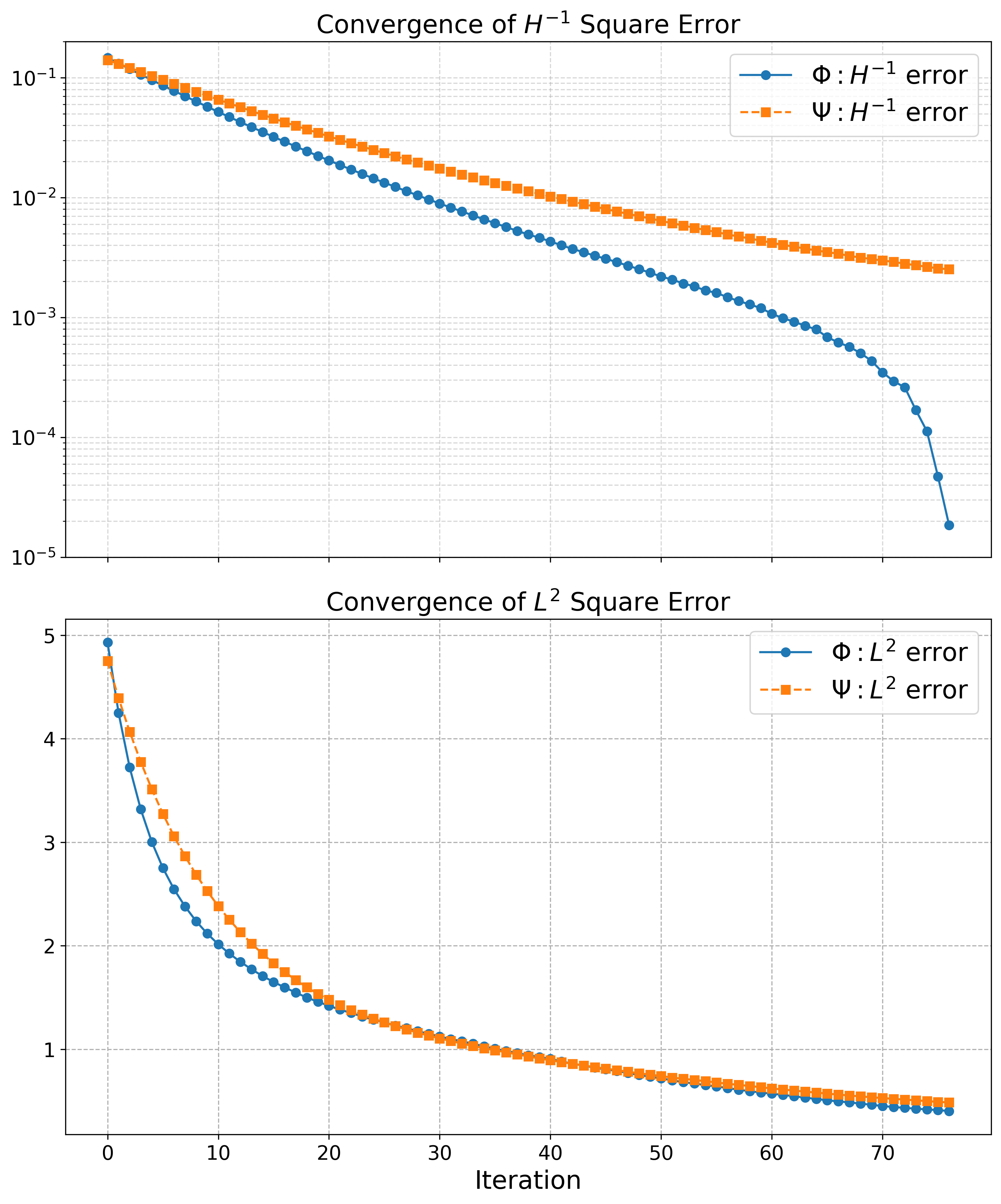}
        \caption{``Eagle'' pattern}
    \end{subfigure}\hfill
    \begin{subfigure}[t]{0.20\textwidth}
        \centering
        \includegraphics[width=\linewidth]{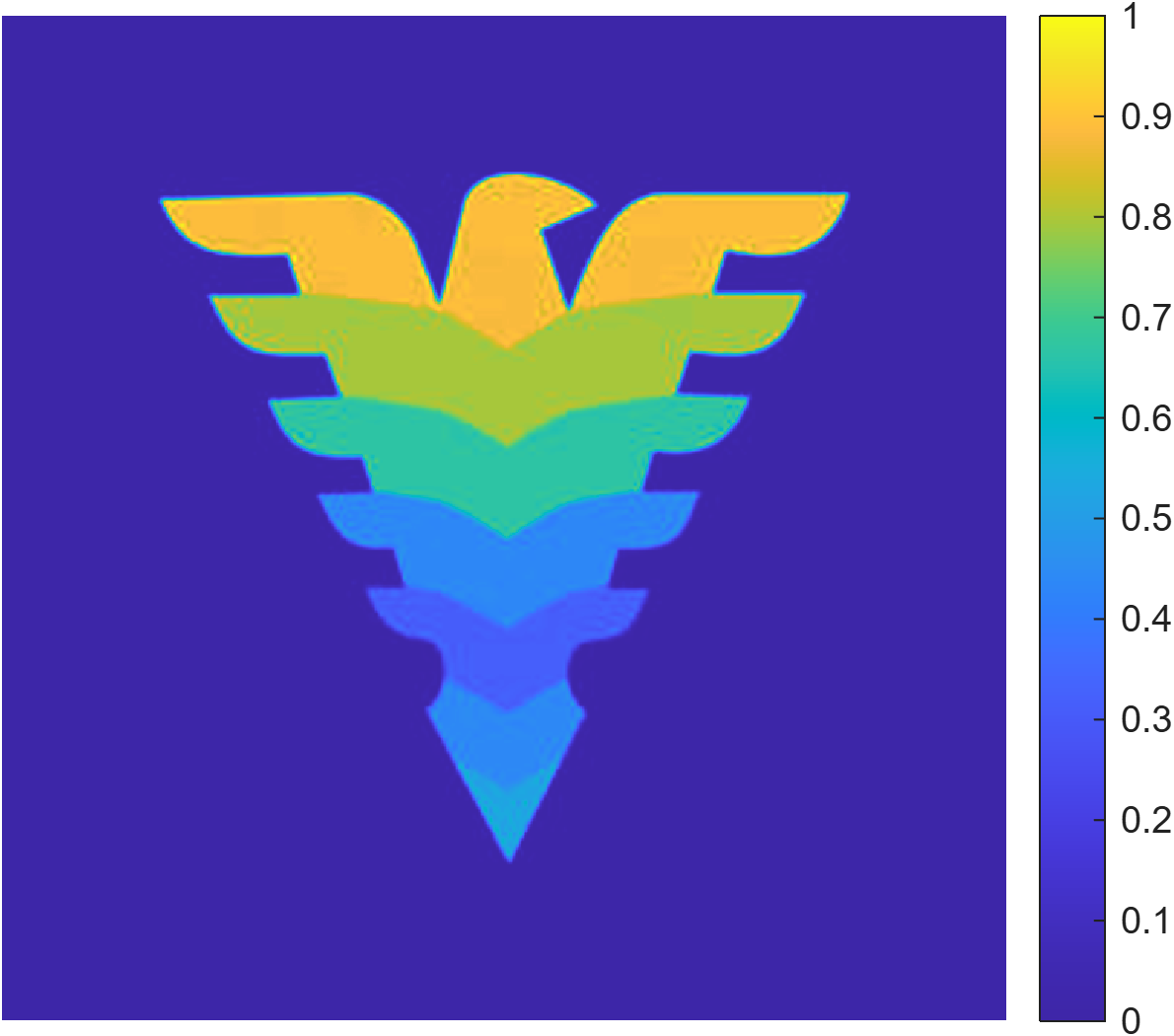}
        \vspace{0.15em}
        \includegraphics[width=\linewidth]{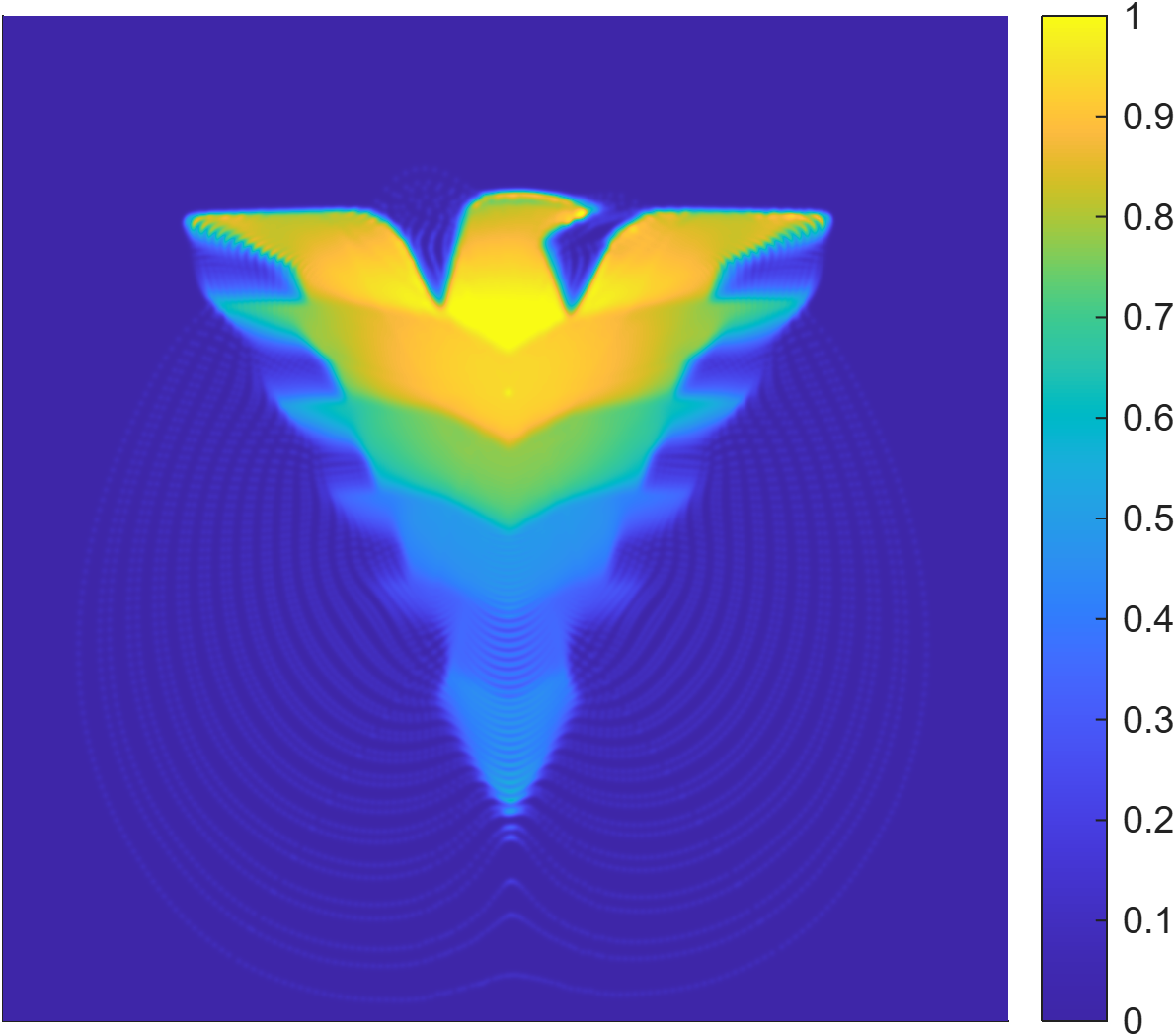}
        \vspace{0.15em}
        \includegraphics[width=\linewidth]{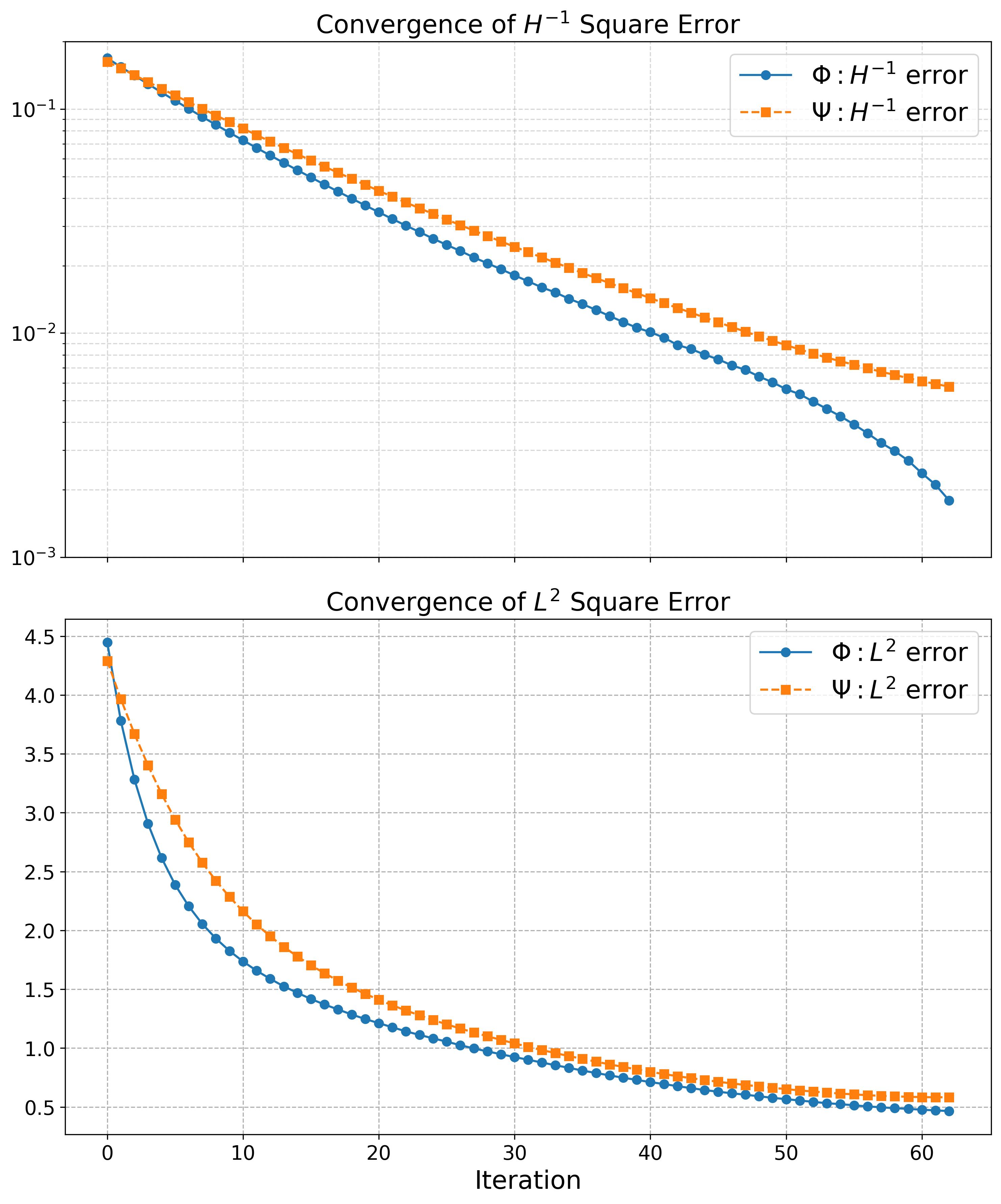}
        \caption{Grayscale ``Eagle''}
    \end{subfigure}
    \caption{Results on four different targets. In each column, the top image is the prescribed target intensity, the middle image is the ray-tracing simulation result, and the bottom image is the convergence curve.}
    \label{fig:convergence}
\end{figure}

\subsection{Comparison with Sobolev Gradient Method}

To further assess the proposed method, we compare it with the Sobolev gradient method~\cite{Far_ref} using the ``ZJU'' and binary ``Eagle'' targets introduced in the previous subsection. In both cases, the prescribed density has compact support in the computational target domain $\Omega^*$ and vanishes on the dark part of $\Omega^*$.

Fig.~\ref{fig:comparison} presents a visual comparison. The Sobolev gradient method recovers the main contours of the targets, but it produces visible artifacts and numerical diffusion near the interfaces. In contrast, the proposed algorithm gives sharper boundaries and lower spurious irradiance in the dark part of $\Omega^*$. For the intricate ``Eagle'' pattern, the ray tracing validation shows that our method better preserves fine structures, including the feather details, while these structures are blurred in the result of the Sobolev gradient method. This demonstrates the advantage of the proposed algorithm for target distributions with zero background, without artificial background regularization.

\begin{figure}[H]
    \centering
    \begin{subfigure}[t]{0.42\textwidth}
        \centering
        \includegraphics[width=0.49\linewidth]{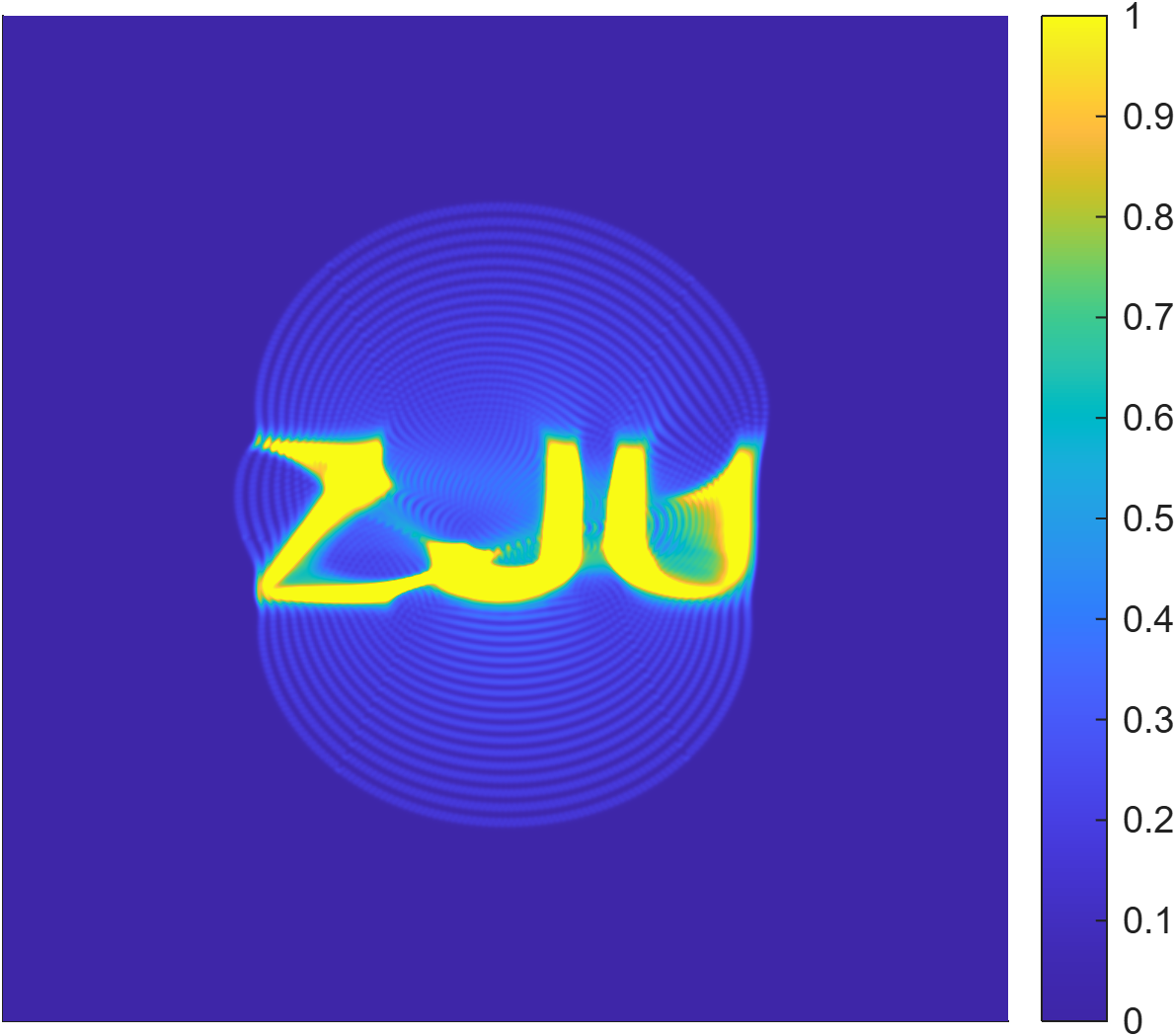}\hfill
        \includegraphics[width=0.49\linewidth]{simulation_ZJU.png}
        \caption{Text ``ZJU''}
    \end{subfigure}\hfill
    \begin{subfigure}[t]{0.42\textwidth}
        \centering
        \includegraphics[width=0.49\linewidth]{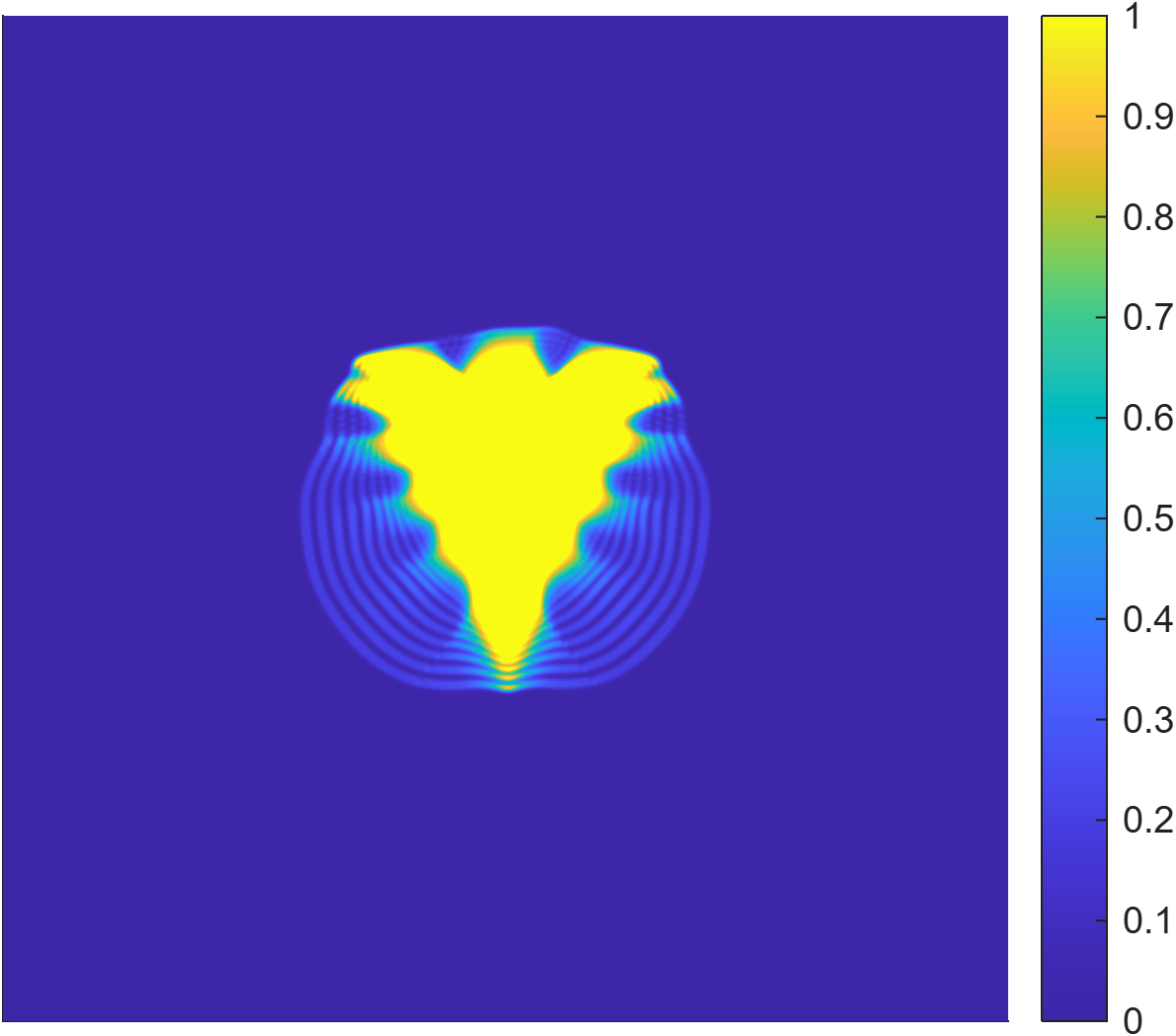}\hfill
        \includegraphics[width=0.49\linewidth]{simulation_eagle.png}
        \caption{``Eagle'' pattern}
    \end{subfigure}
    \caption{Comparison between the Sobolev gradient method and the proposed algorithm. In each pair, the left image is the Sobolev gradient result and the right image is the result of the proposed method.}
    \label{fig:comparison}
\end{figure}

\subsection{Influence of Mesh Density}

To assess the sensitivity of the proposed method to spatial discretization, we examine the influence of the finite element mesh density using the binary ``Eagle'' target, whose fine structures provide a sensitive test of spatial resolution.

The time-step parameter and stopping criterion are fixed, while the mesh is refined through three levels, namely level 40, level 80, and level 100. For each mesh level, we record the ray-tracing irradiance and use the same $H^{-1}$ and $L^2$ residual norms as convergence indicators.

\begin{figure}[H]
    \centering
    \begin{subfigure}[t]{0.20\textwidth}
        \centering
        \includegraphics[width=\linewidth]{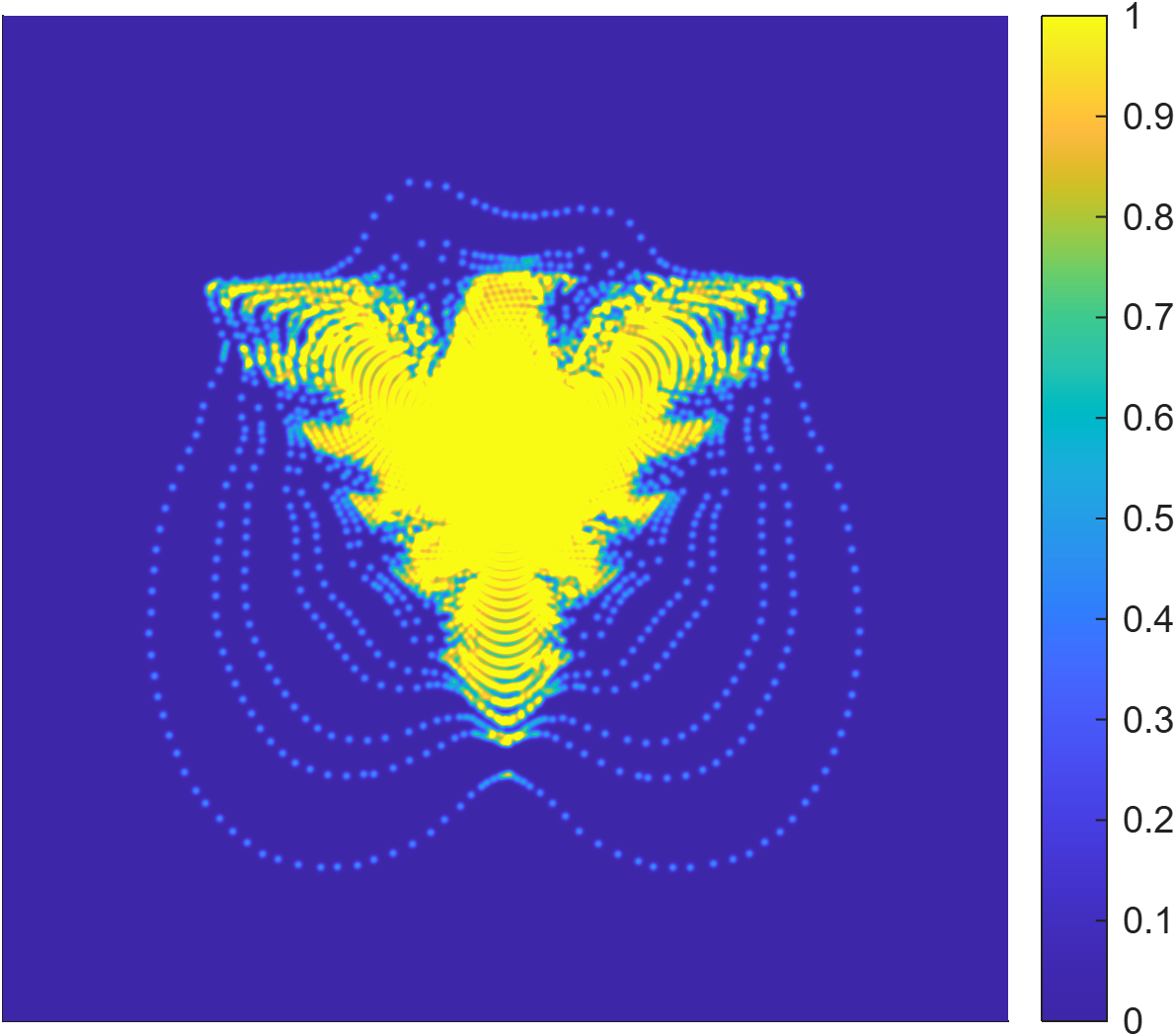}
        \vspace{0.35em}
        \includegraphics[width=\linewidth]{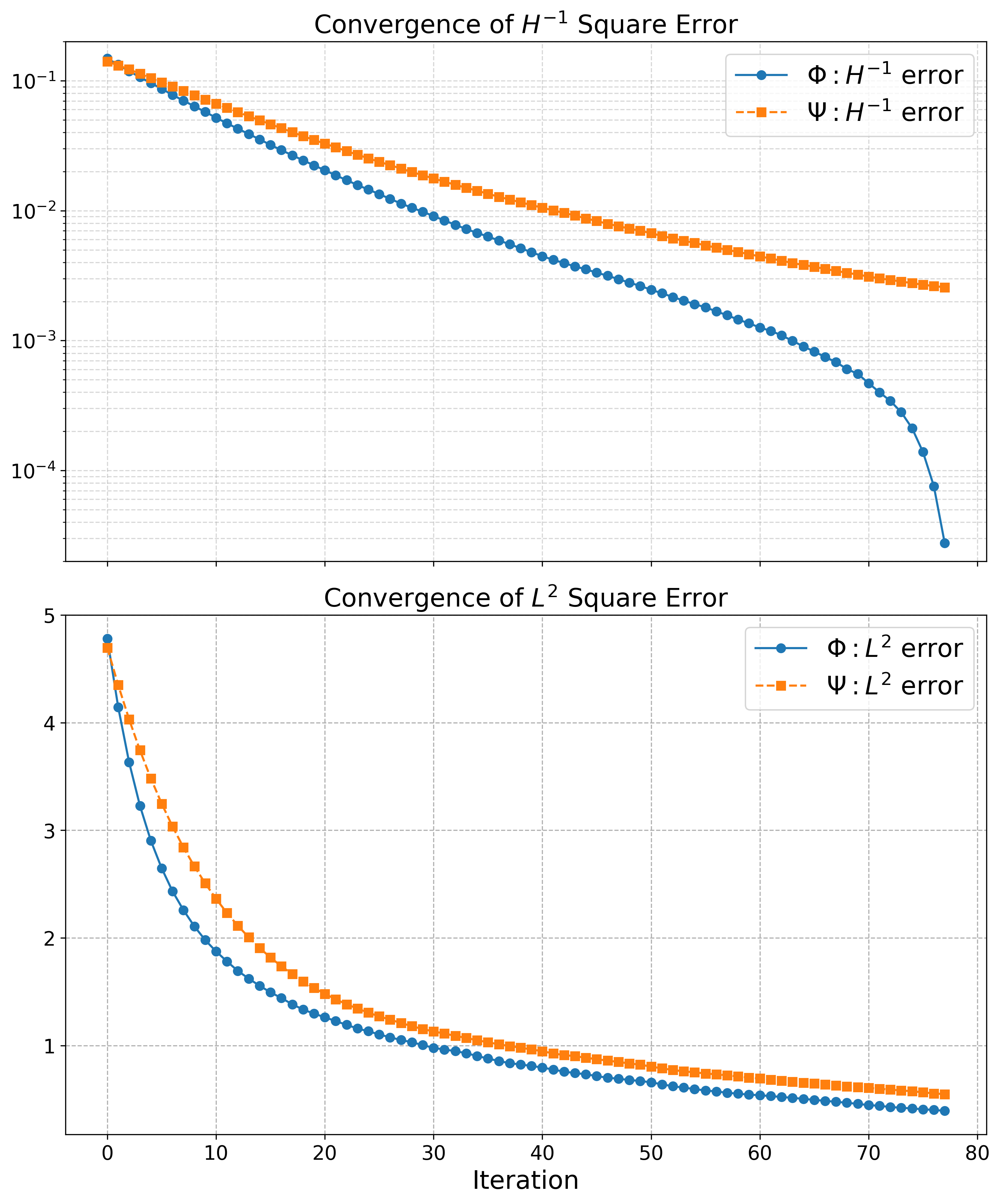}
        \caption{Level 40}
    \end{subfigure}\hfill
    \begin{subfigure}[t]{0.20\textwidth}
        \centering
        \includegraphics[width=\linewidth]{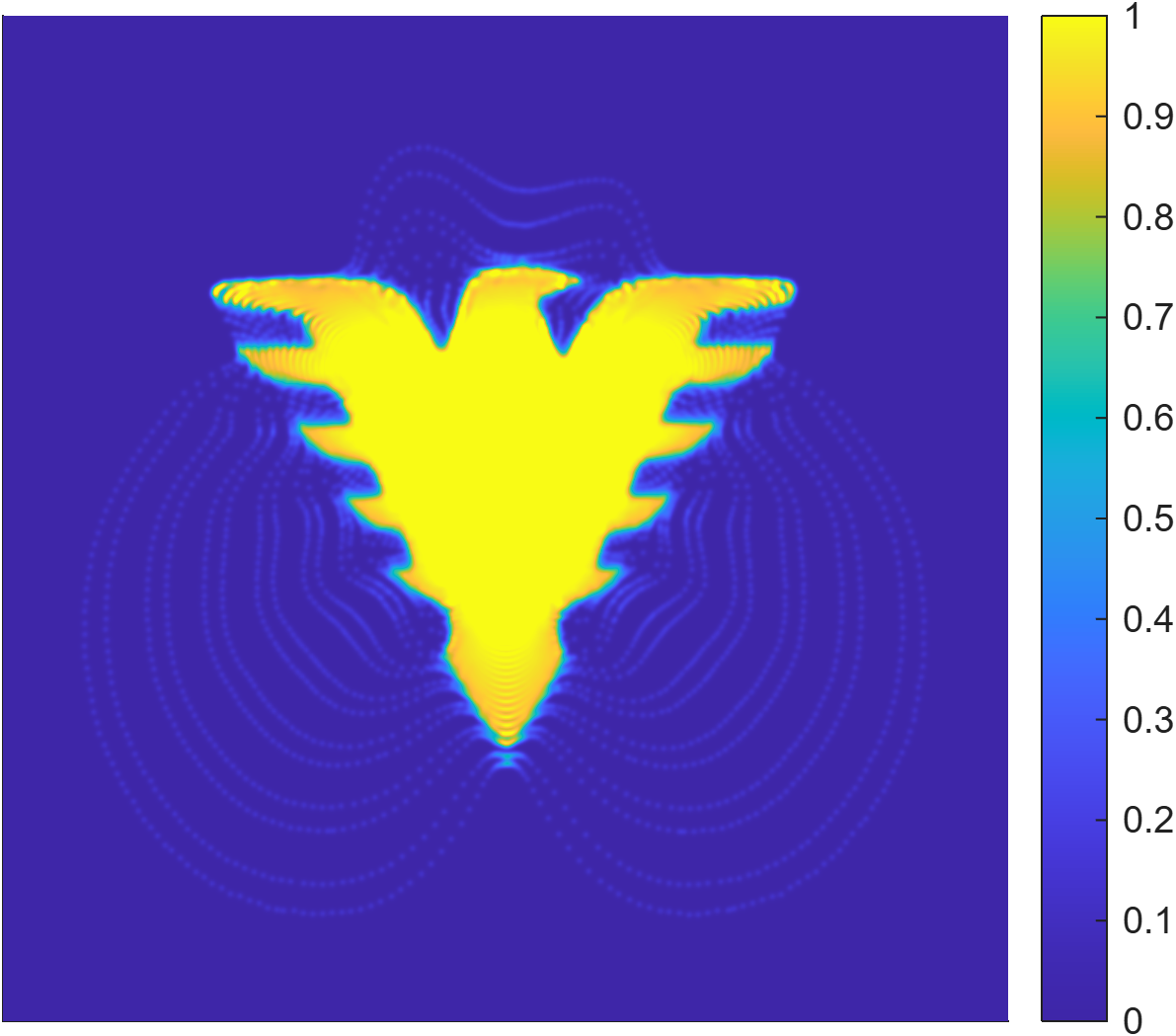}
        \vspace{0.35em}
        \includegraphics[width=\linewidth]{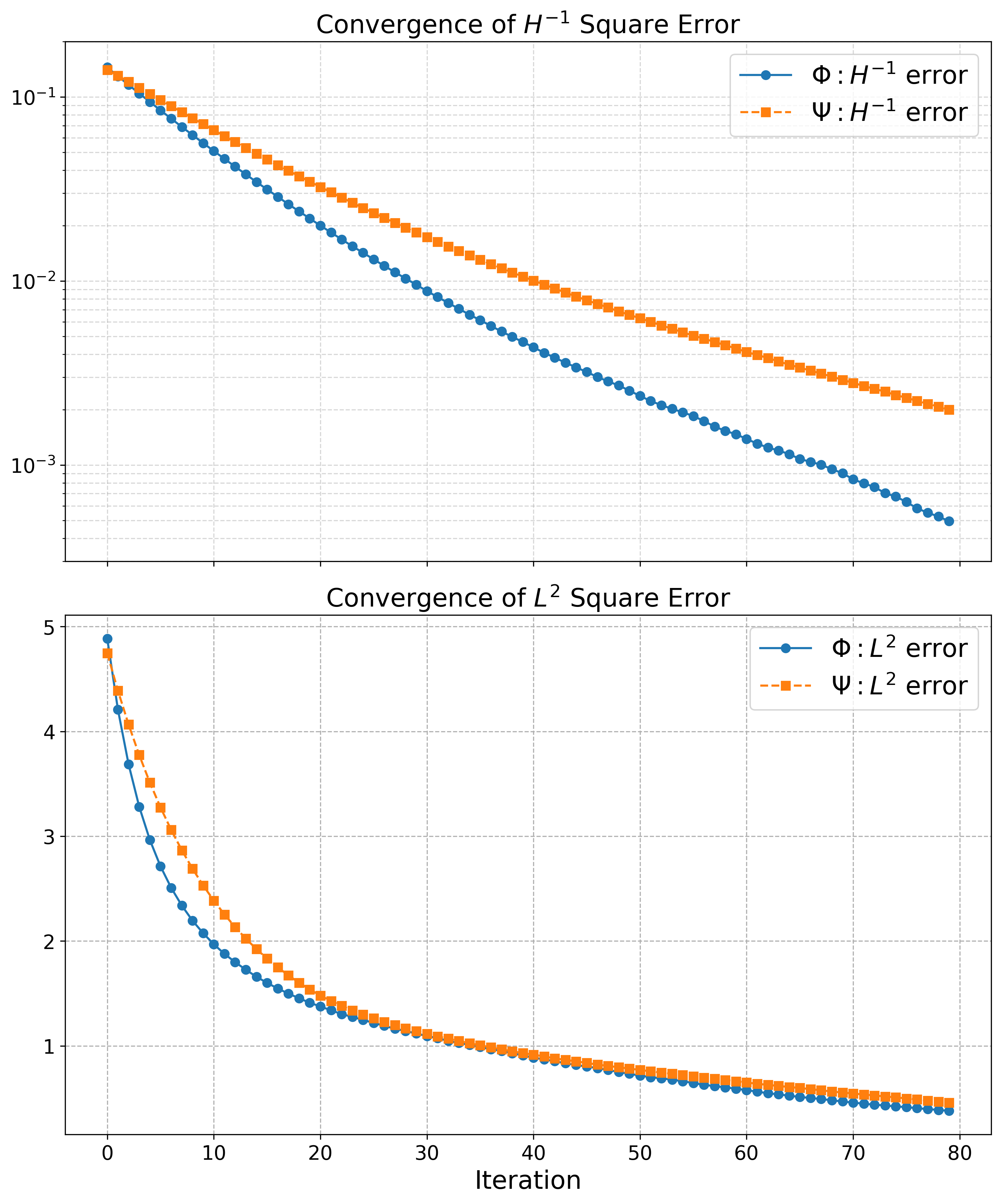}
        \caption{Level 80}
    \end{subfigure}\hfill
    \begin{subfigure}[t]{0.20\textwidth}
        \centering
        \includegraphics[width=\linewidth]{simulation_eagle.png}
        \vspace{0.35em}
        \includegraphics[width=\linewidth]{Curve_qsy_level100.png}
        \caption{Level 100}
    \end{subfigure}
    \caption{Influence of mesh density on the proposed algorithm for the binary ``Eagle'' target. The three columns correspond to mesh levels 40, 80, and 100. The upper row shows the ray-tracing irradiance, and the lower row shows the corresponding $H^{-1}$ and $L^2$ residual histories.}
    \label{fig:mesh-density}
\end{figure}

As shown in Fig.~\ref{fig:mesh-density}, increasing the mesh level consistently improves the spatial resolution of the reconstruction. At level $40$, the overall geometry of the ``Eagle'' target is recovered, but the thin feather structures are partially blurred and the local irradiance is less uniform. Refinement to levels $80$ and $100$ progressively resolves these fine structures, sharpens the interfaces, and suppresses spurious irradiance in the dark region.

The residual curves show the same trend. Finer meshes reach lower final $H^{-1}$ and $L^2$ residuals and produce more faithful ray-tracing results, confirming the importance of sufficient spatial resolution for targets with sharp interfaces and fine structures.

\subsection{Influence of the Artificial Time Step}

We next study the role of the artificial time step parameter $\mathrm{d}t$ in the bidirectional iteration. The mesh level, target pattern, ray-tracing validation, and stopping criterion are fixed, and only $\mathrm{d}t$ is varied.

We use the letter ``A'' target as a representative binary pattern and compare four step sizes: $\mathrm{d}t=0.01,0.04,0.10$, and $0.15$. The case
$\mathrm{d}t=0.01$ is the same as the letter ``A'' result reported in Fig.~\ref{fig:convergence}. For each step size, we record the ray-tracing irradiance and the corresponding loss curves measured by the $H^{-1}$ and $L^2$ residual norms.

\begin{figure}[H]
    \centering
    \begin{subfigure}[t]{0.20\textwidth}
        \centering
        \includegraphics[width=\linewidth]{simulation_A.png}
        \vspace{0.35em}
        \includegraphics[width=\linewidth]{Curve_A_level100.png}
        \caption{$\mathrm{d}t=0.01$}
    \end{subfigure}\hfill
    \begin{subfigure}[t]{0.20\textwidth}
        \centering
        \includegraphics[width=\linewidth]{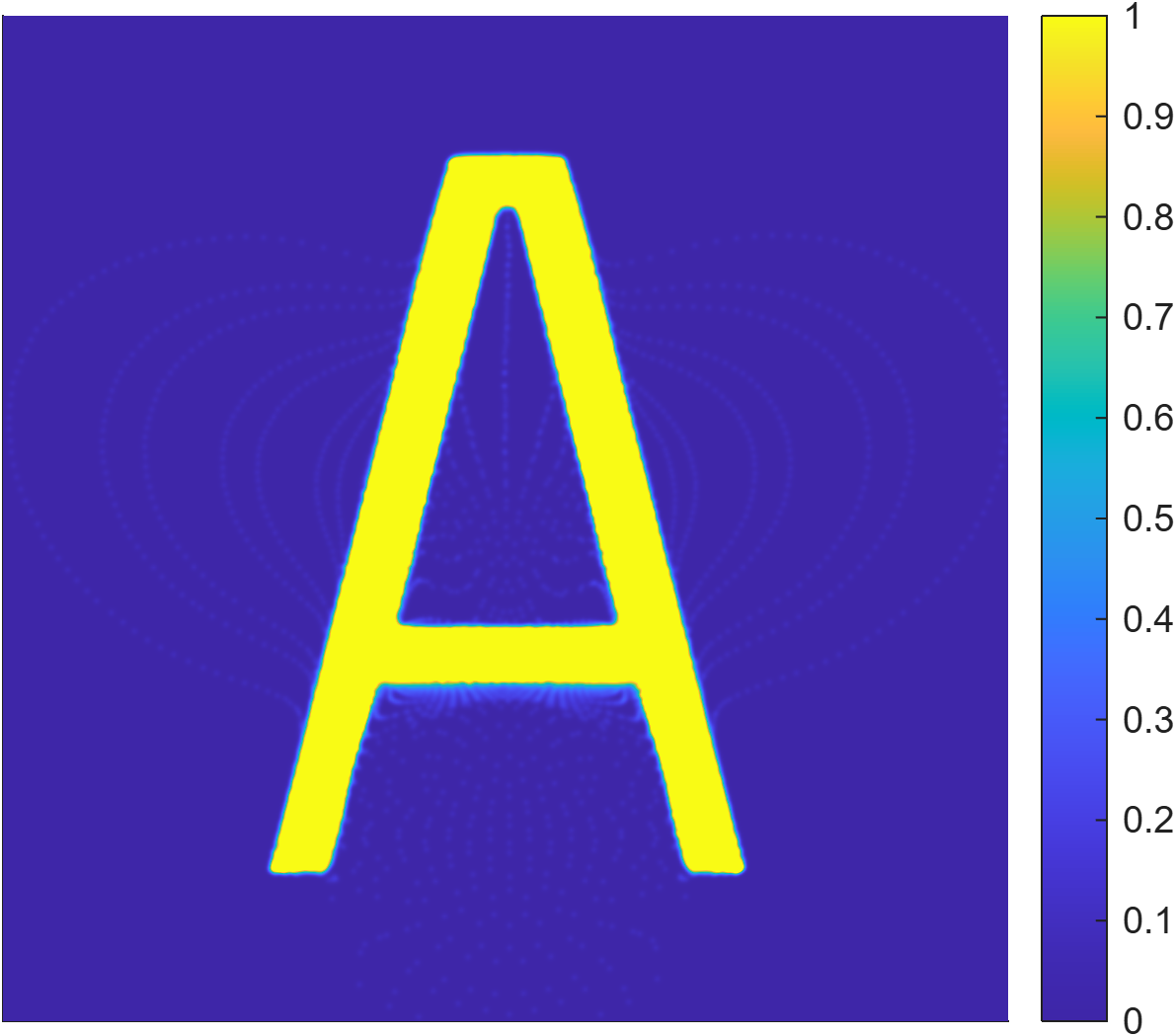}
        \vspace{0.35em}
        \includegraphics[width=\linewidth]{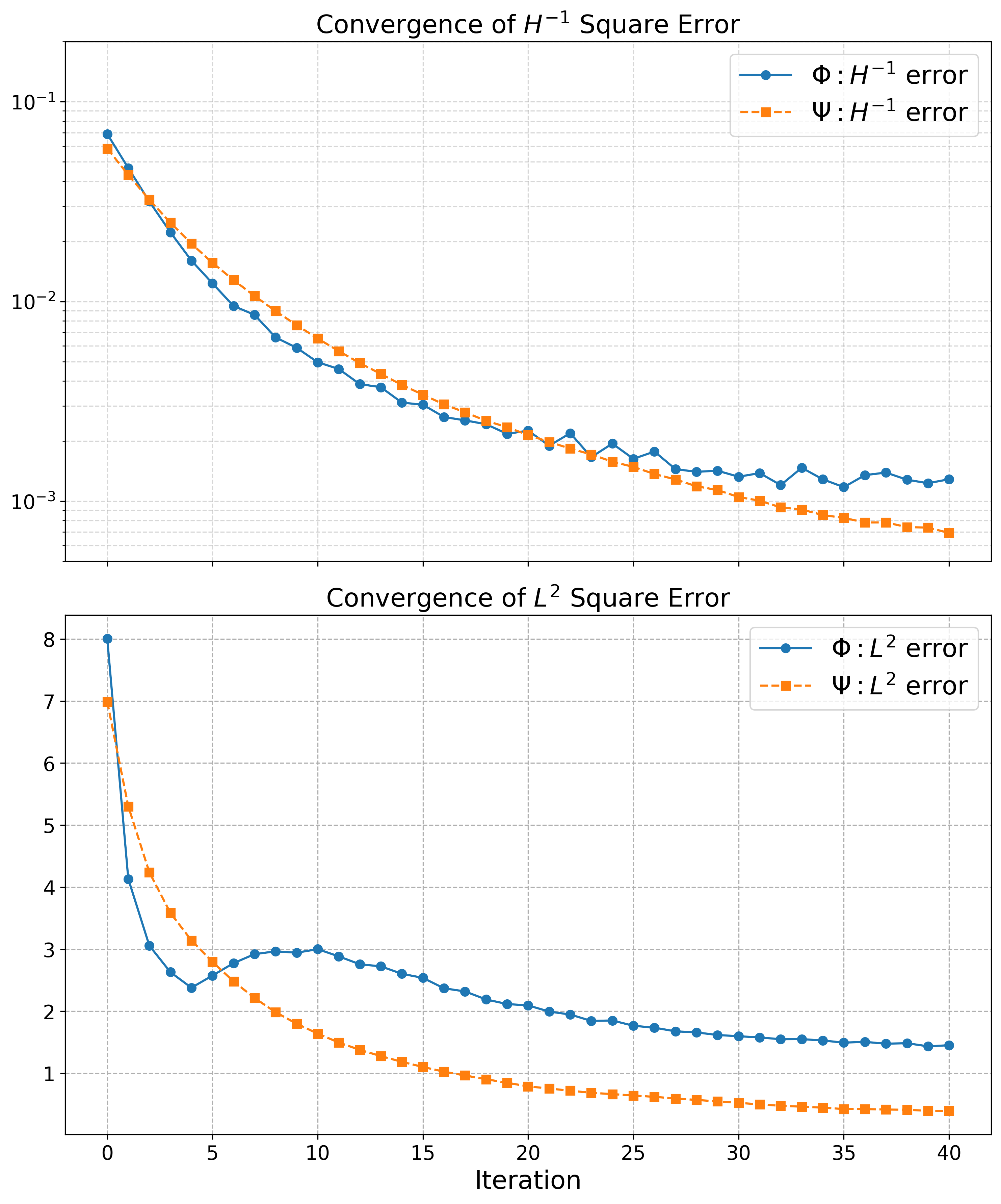}
        \caption{$\mathrm{d}t=0.04$}
    \end{subfigure}\hfill
    \begin{subfigure}[t]{0.20\textwidth}
        \centering
        \includegraphics[width=\linewidth]{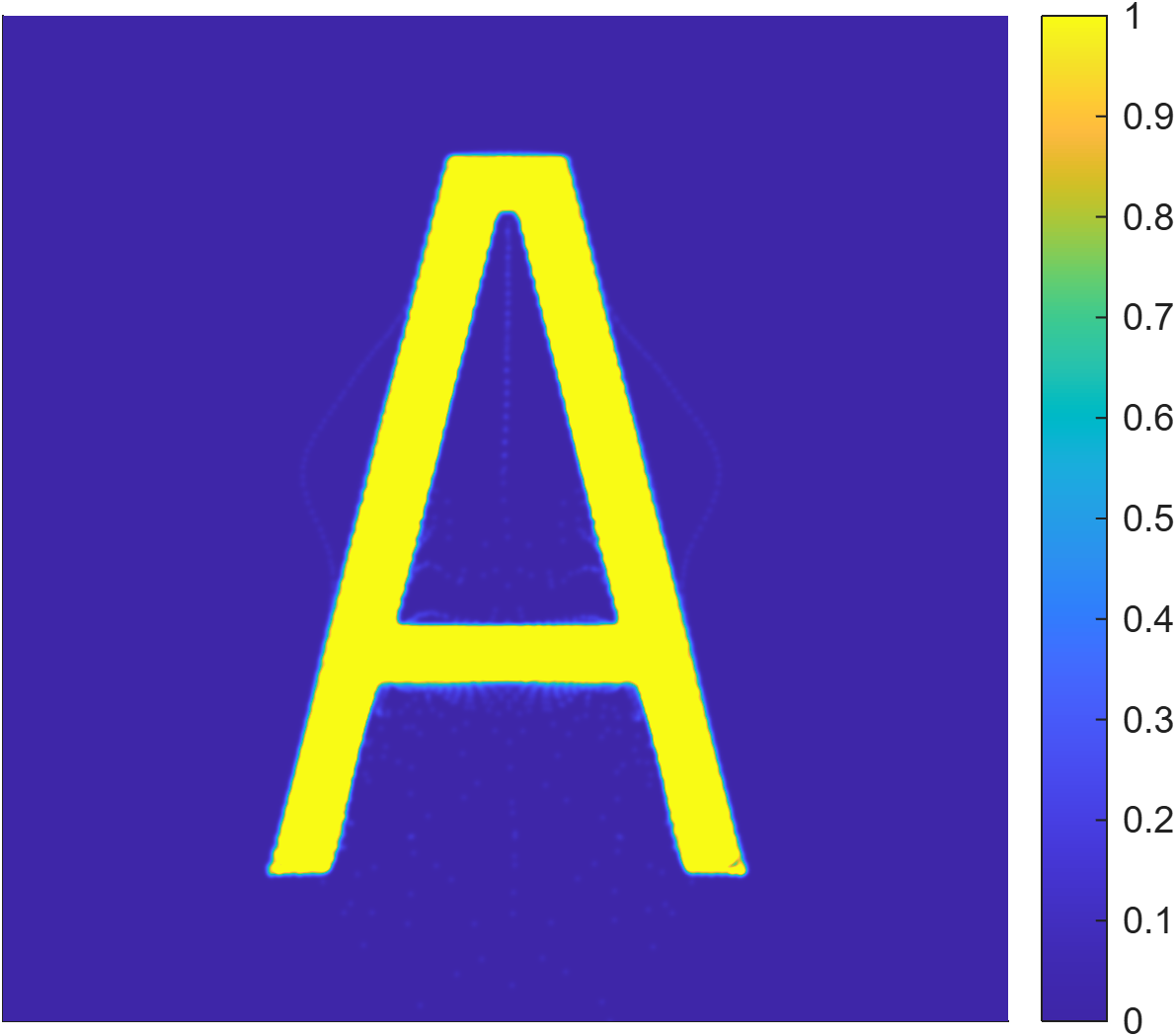}
        \vspace{0.35em}
        \includegraphics[width=\linewidth]{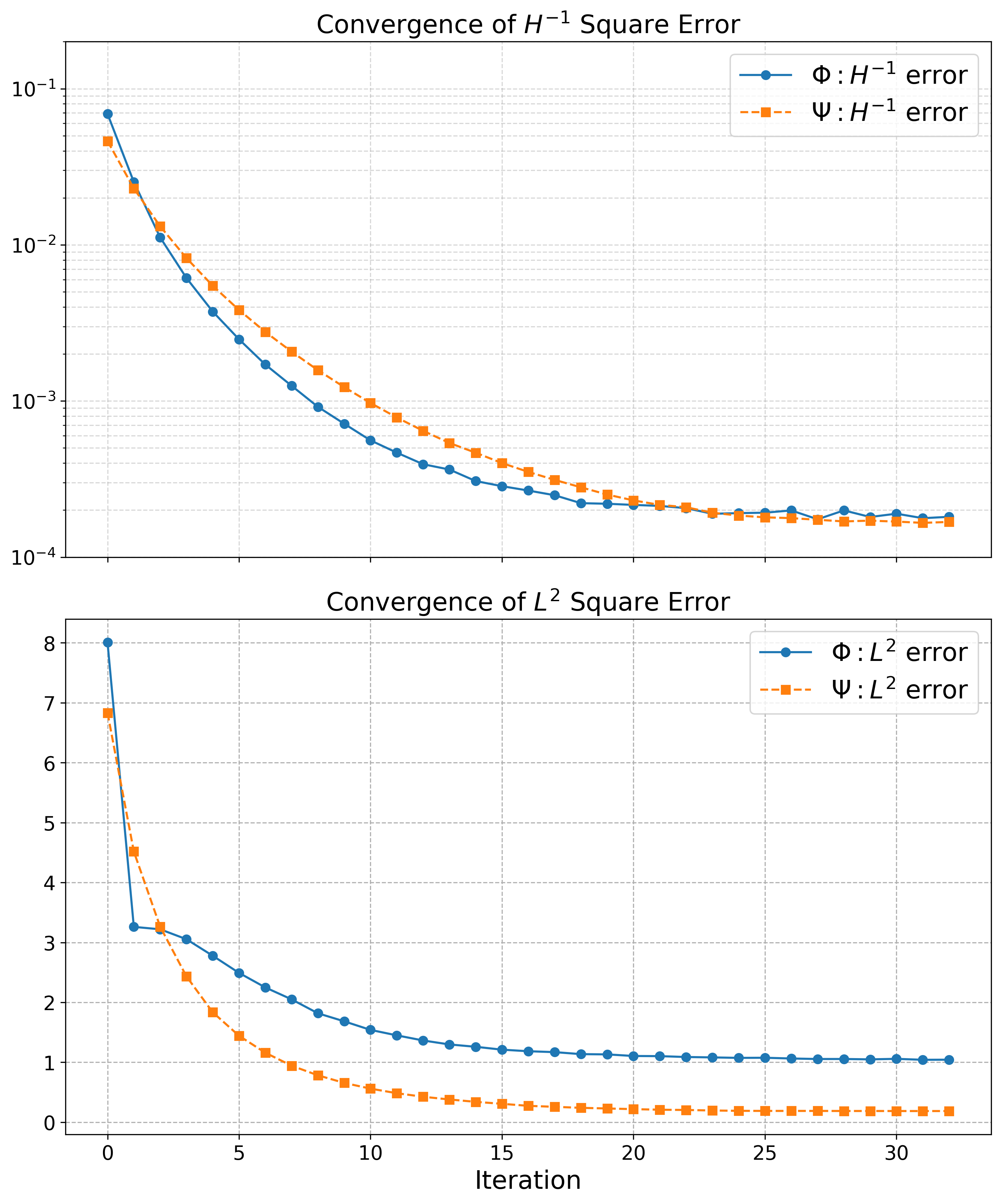}
        \caption{$\mathrm{d}t=0.10$}
    \end{subfigure}\hfill
    \begin{subfigure}[t]{0.20\textwidth}
        \centering
        \includegraphics[width=\linewidth]{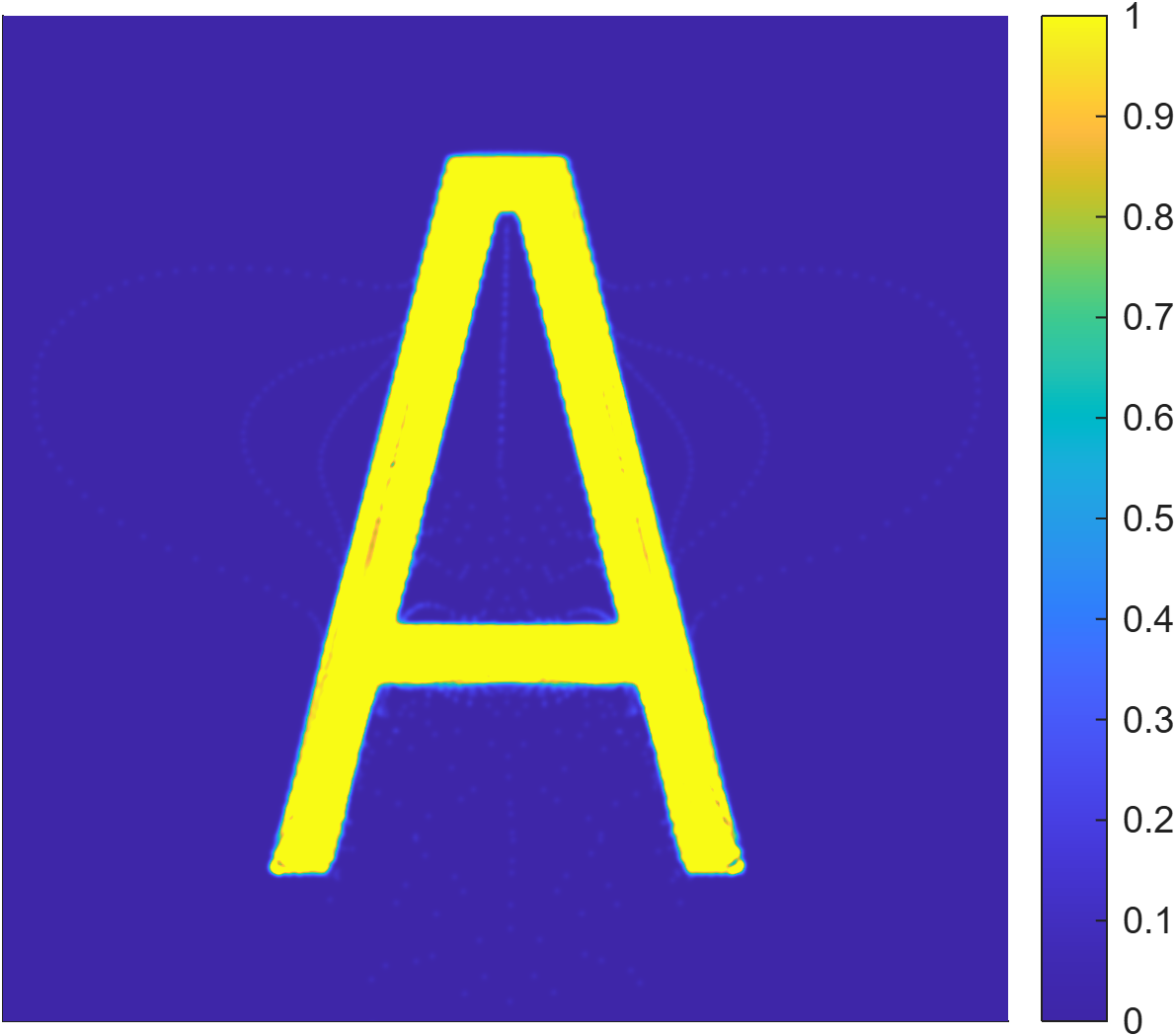}
        \vspace{0.35em}
        \includegraphics[width=\linewidth]{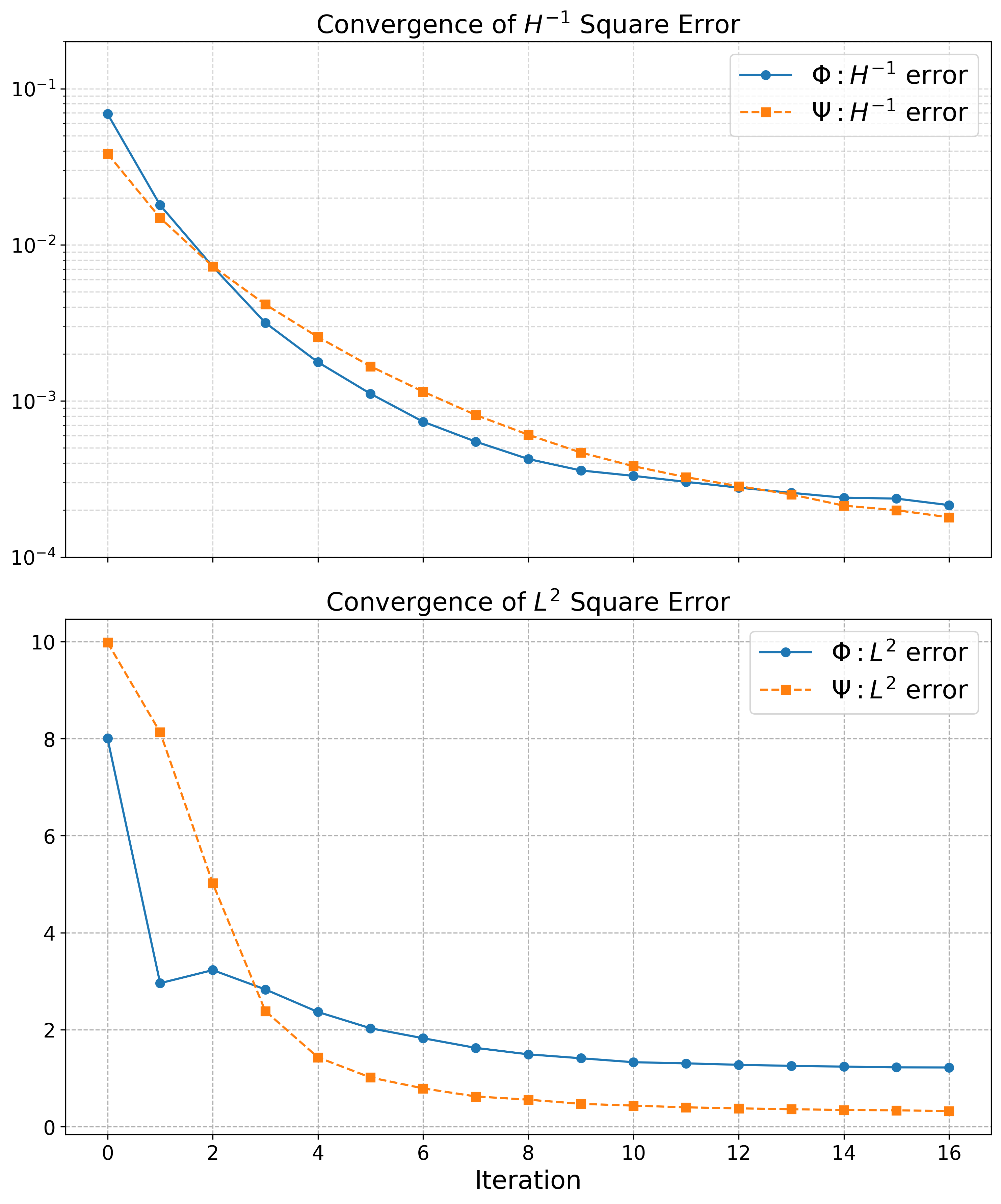}
        \caption{$\mathrm{d}t=0.15$}
    \end{subfigure}
    \caption{Influence of the artificial time step $\mathrm{d}t$ on the letter ``A'' target. The four columns correspond to $\mathrm{d}t=0.01,0.04,0.10$, and $0.15$, respectively. The first row shows the ray-tracing irradiance obtained by the algorithm under different step
    sizes, and the second row shows the corresponding loss histories measured
    by the $H^{-1}$ and $L^2$ residual norms.}
    \label{fig:step-size}
\end{figure}

As shown in Fig.~\ref{fig:step-size}, increasing the time step accelerates the initial decrease of the loss. In the letter ``A'' test, $\mathrm{d}t=0.01$ gives the largest final residual among the tested values, whereas $\mathrm{d}t=0.10$ reaches the lowest final loss. Larger time steps also tend to produce sharper boundaries and reduce some of the background artifacts.

This improvement, however, comes with a stability trade-off. As $\mathrm{d}t$ increases, the reconstruction may develop more irregular artifacts and less uniform irradiance inside the illuminated region. For this example, $\mathrm{d}t=0.04$--$0.10$ provides a favorable balance between convergence speed and reconstruction quality. This range is problem dependent: the appropriate step size depends on the target geometry and mesh resolution, and excessively large steps may lead to instability or divergence for more complex targets.

\subsection{Influence of Target Singularity}

We further examine its robustness with respect to different geometric singularities in zero-background target distributions of this kind. Four representative target types are considered: an annular support with an interior hole, several thin illuminated bars, multiple disconnected target components, and a target located very close to the boundary of $\Omega^*$. These examples are designed to test whether the method can preserve holes,
resolve narrow supports, separate disconnected illuminated regions, and control the transport near the target boundary within the same computational framework.

\begin{figure}[!htbp]
    \centering
    \begin{subfigure}[t]{0.20\textwidth}
        \centering
        \includegraphics[width=\linewidth]{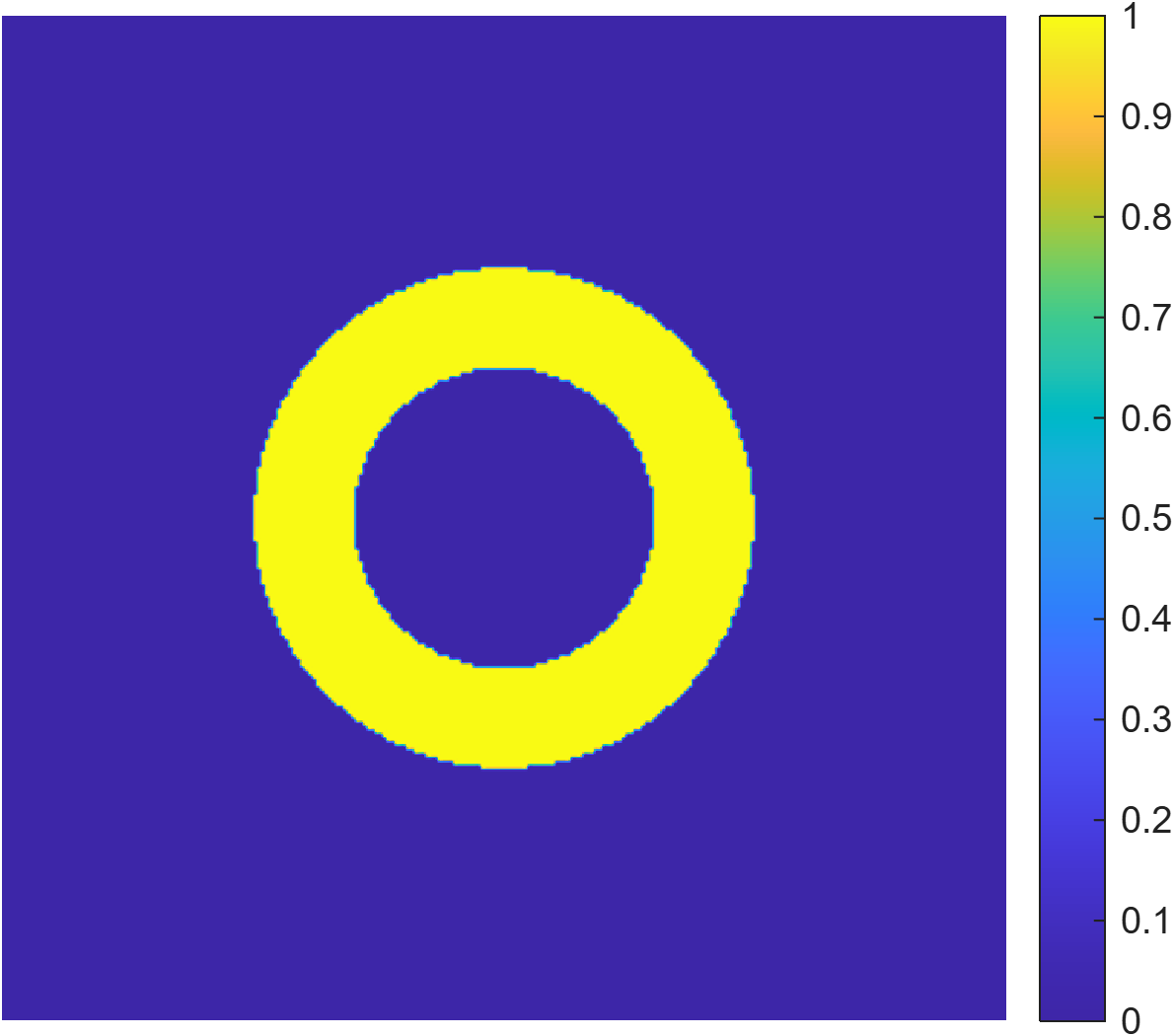}
        \vspace{0.35em}
        \includegraphics[width=\linewidth]{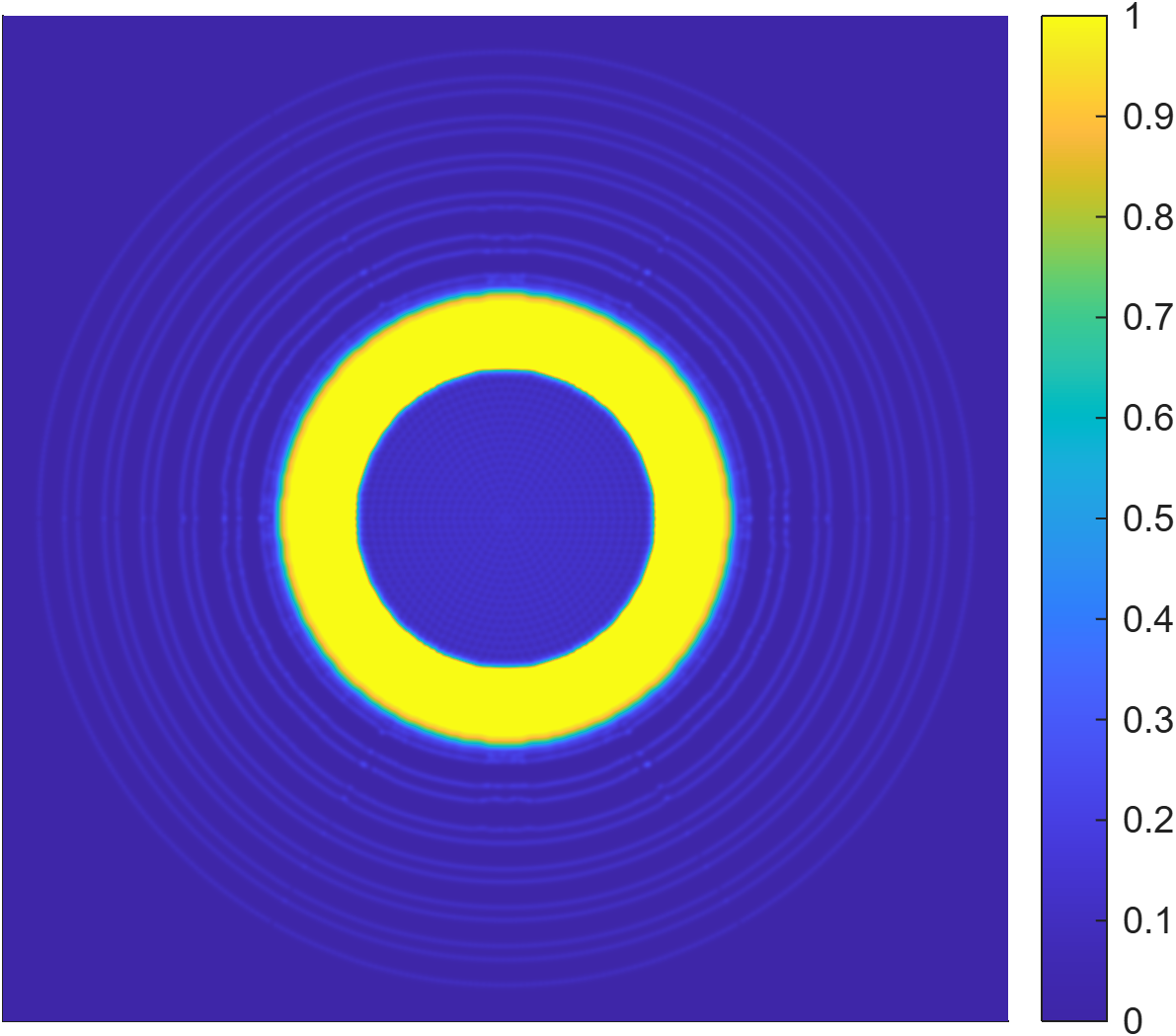}
        \vspace{0.35em}
        \includegraphics[width=\linewidth]{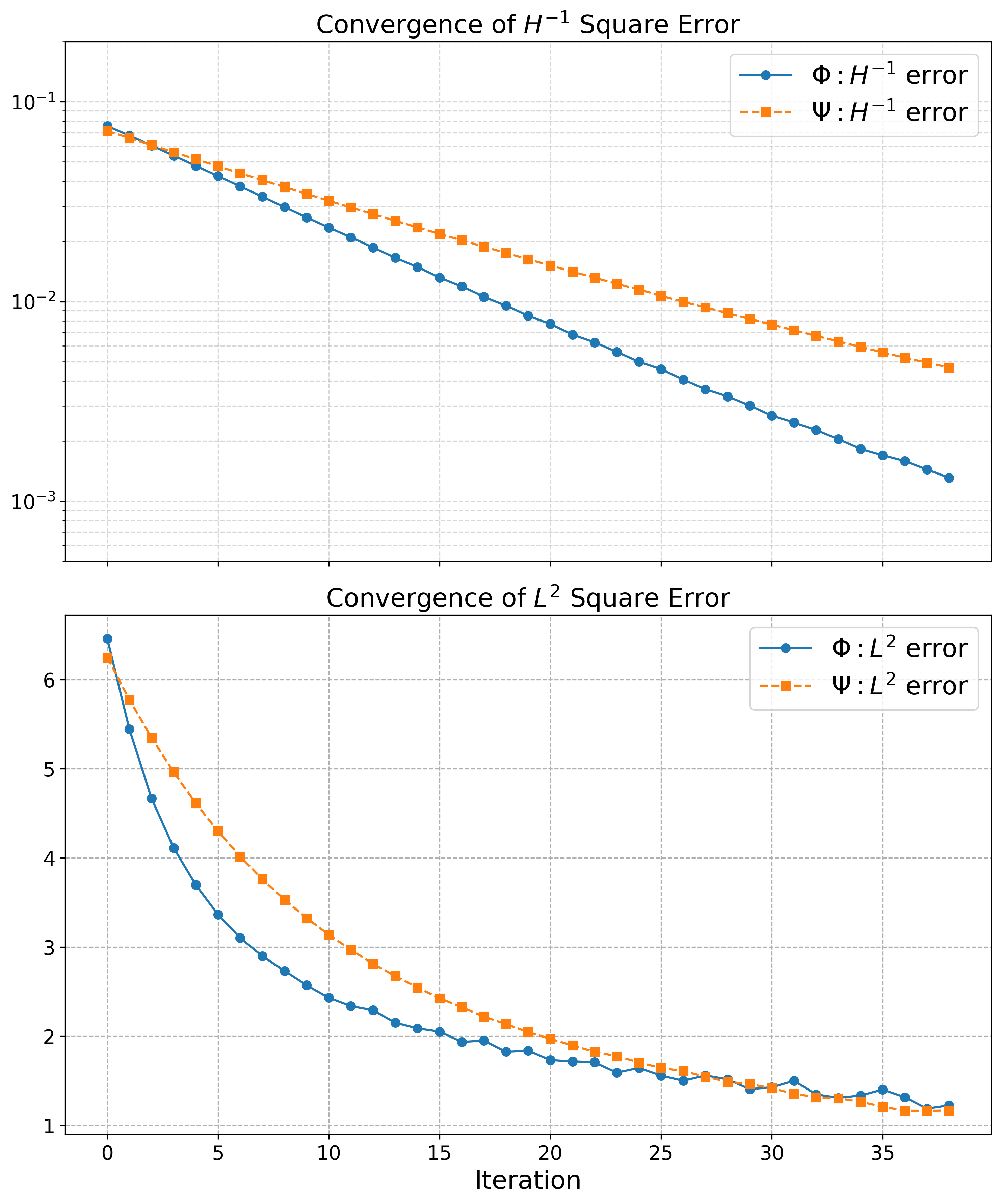}
        \caption{Annulus}
    \end{subfigure}\hfill
    \begin{subfigure}[t]{0.20\textwidth}
        \centering
        \includegraphics[width=\linewidth]{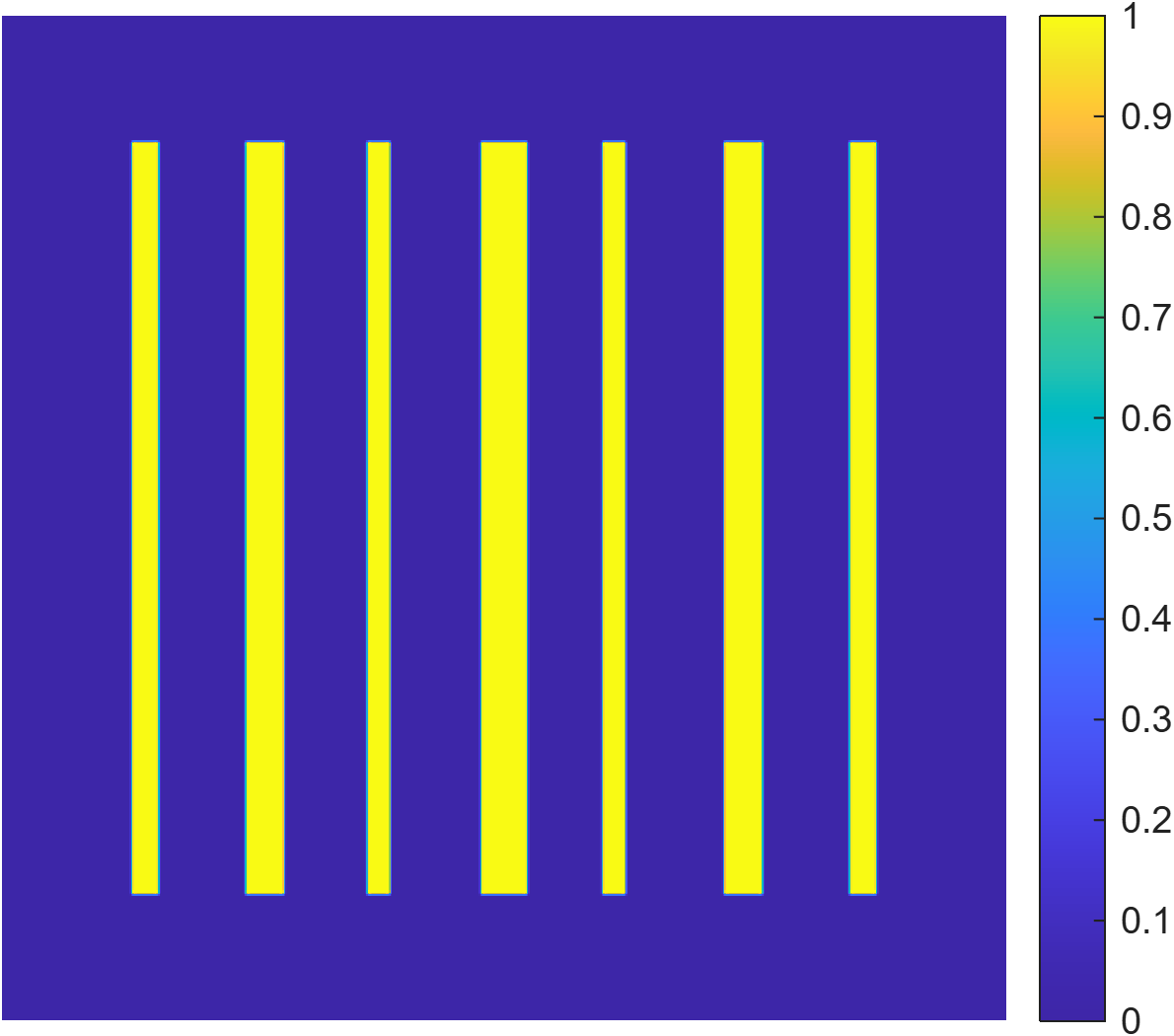}
        \vspace{0.35em}
        \includegraphics[width=\linewidth]{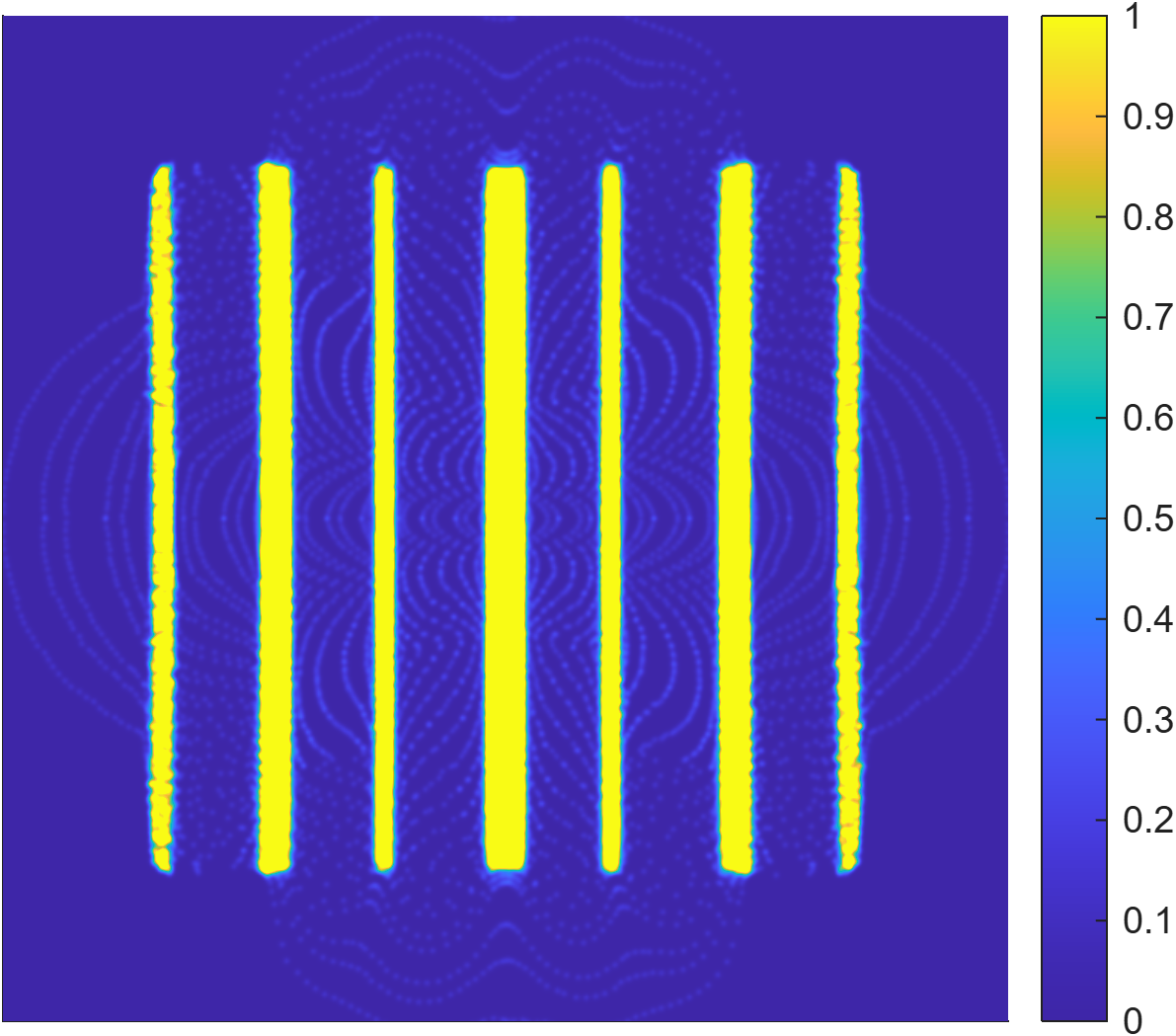}
        \vspace{0.35em}
        \includegraphics[width=\linewidth]{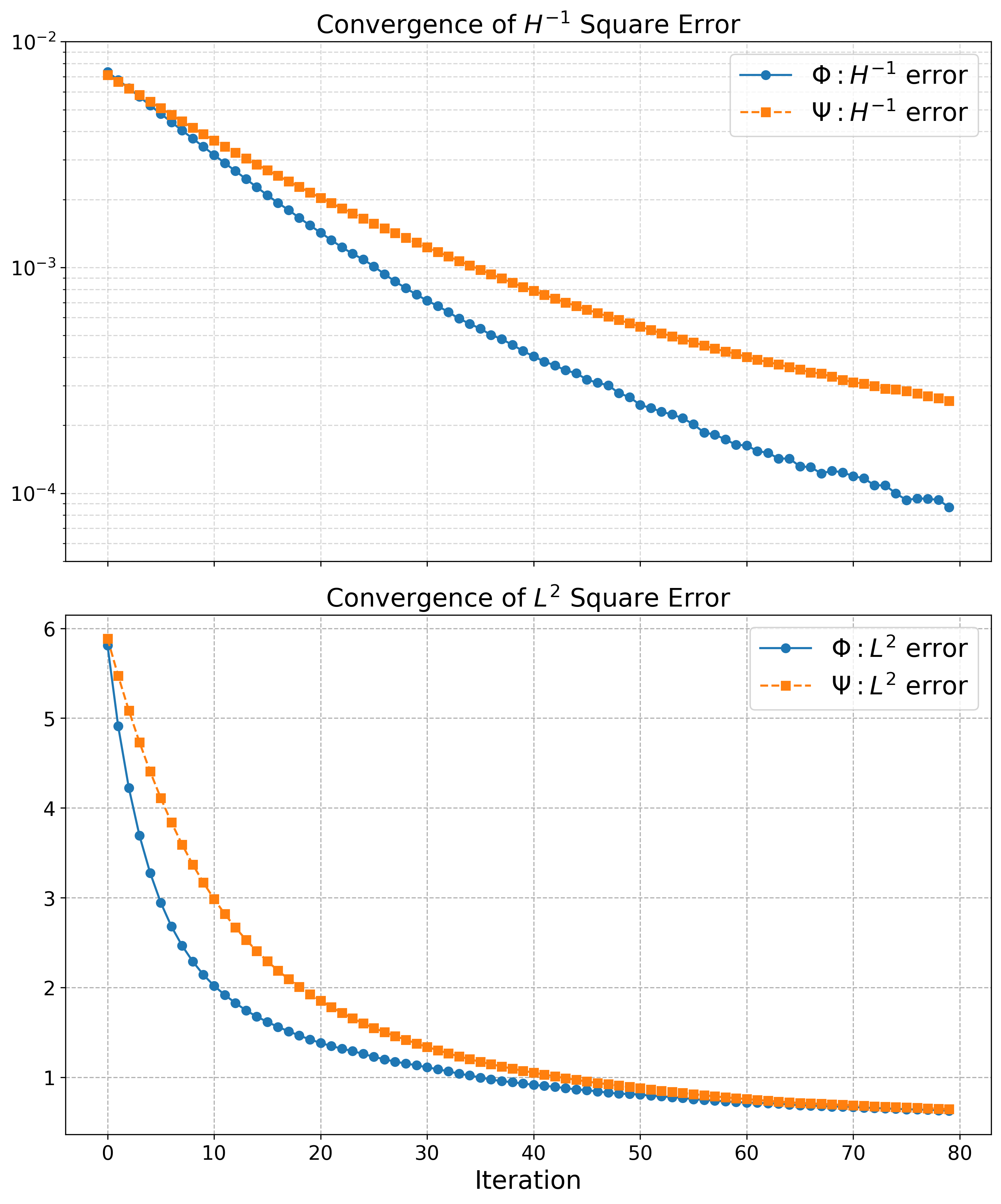}
        \caption{Thin bars}
    \end{subfigure}\hfill
    \begin{subfigure}[t]{0.20\textwidth}
        \centering
        \includegraphics[width=\linewidth]{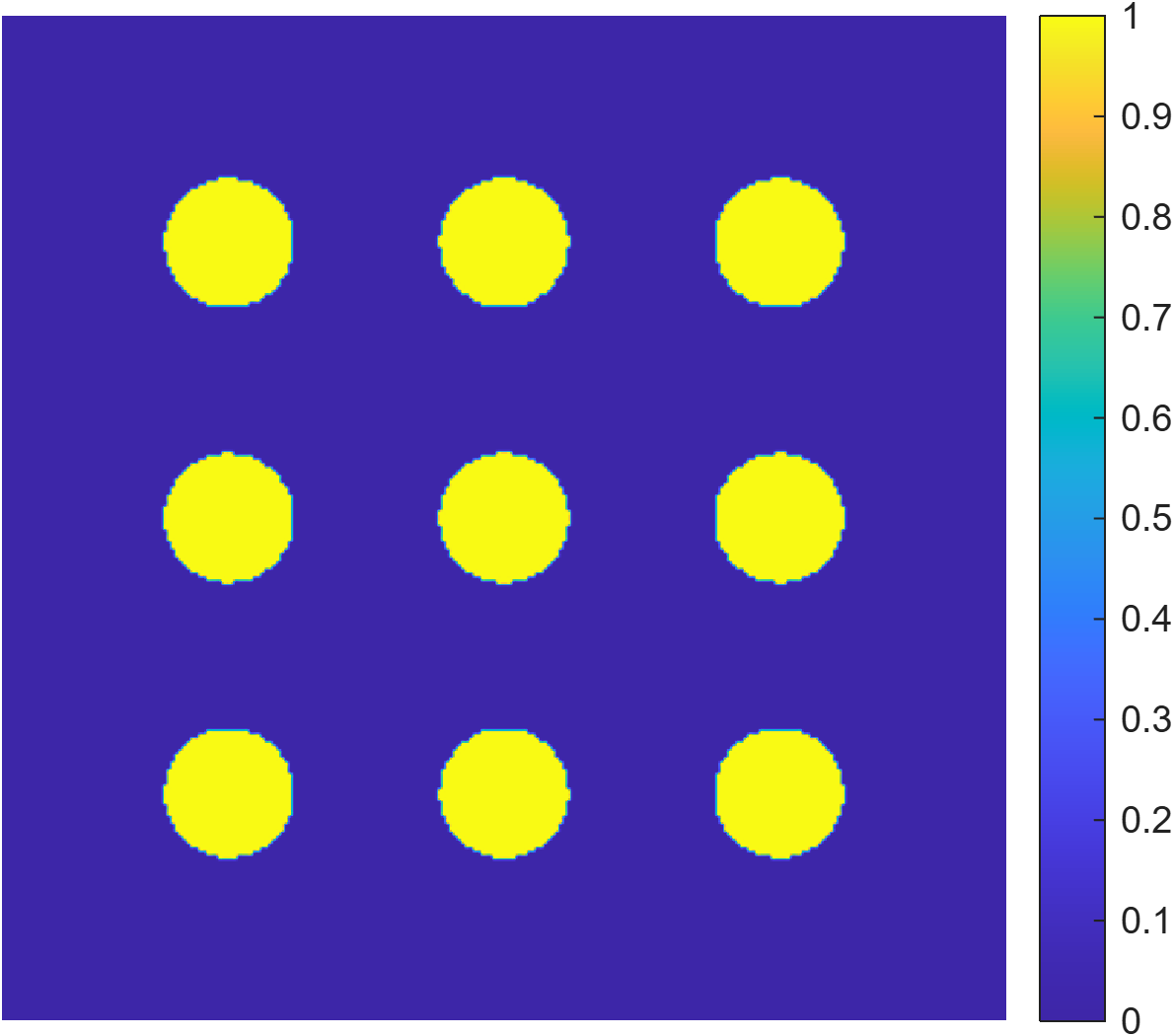}
        \vspace{0.35em}
        \includegraphics[width=\linewidth]{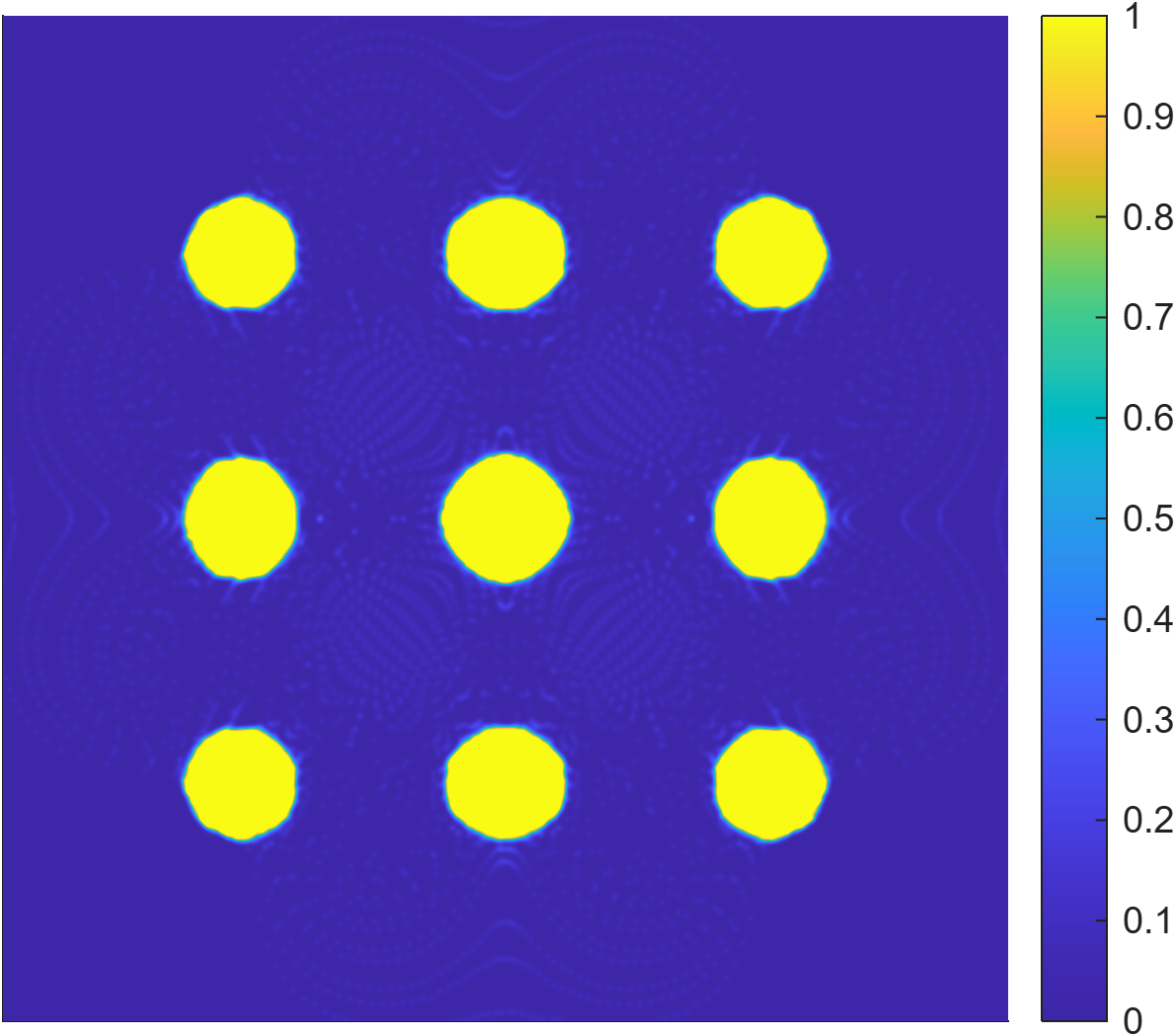}
        \vspace{0.35em}
        \includegraphics[width=\linewidth]{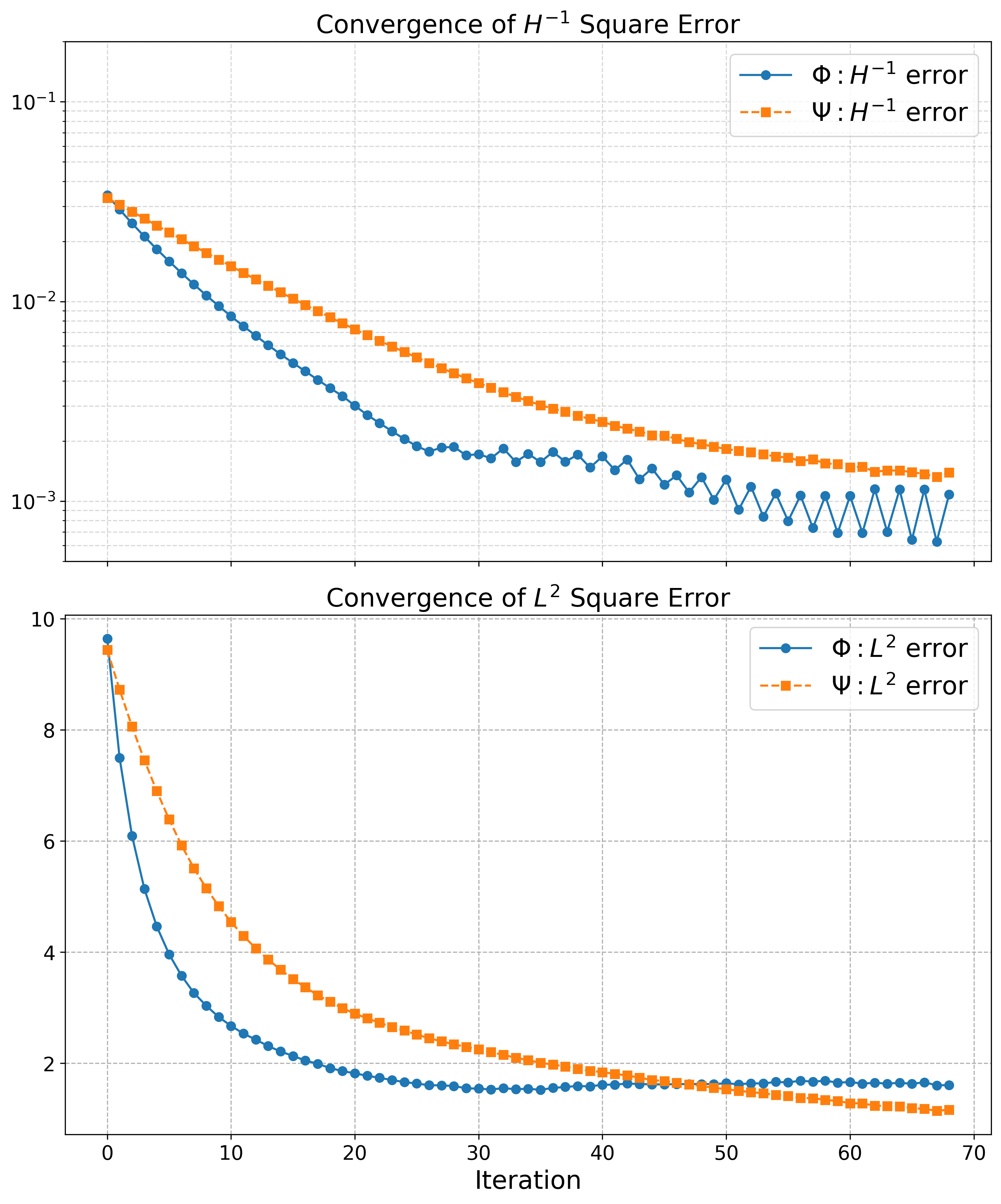}
        \caption{Separated disks}
    \end{subfigure}\hfill
    \begin{subfigure}[t]{0.20\textwidth}
        \centering
        \includegraphics[width=\linewidth]{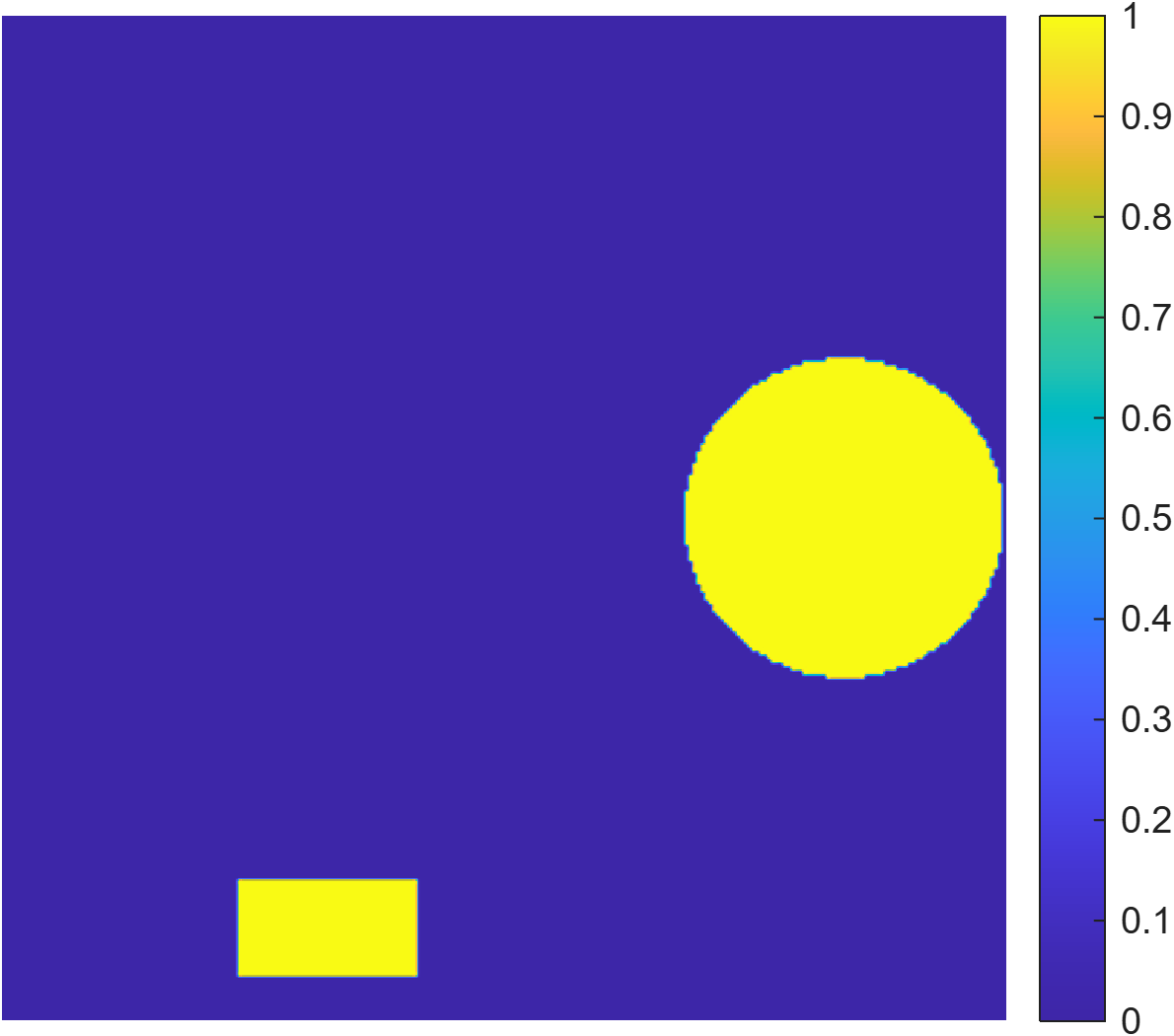}
        \vspace{0.35em}
        \includegraphics[width=\linewidth]{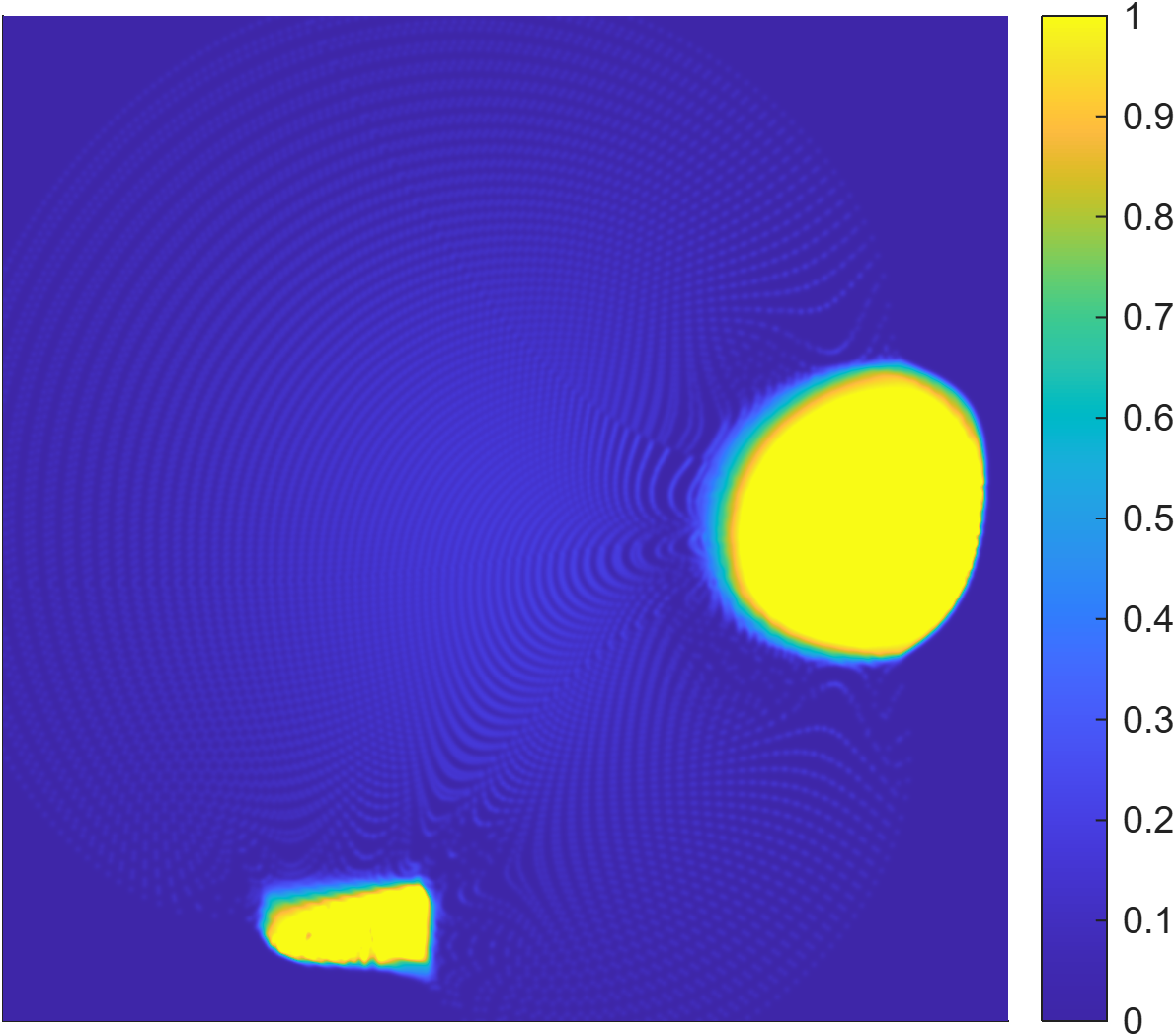}
        \vspace{0.35em}
        \includegraphics[width=\linewidth]{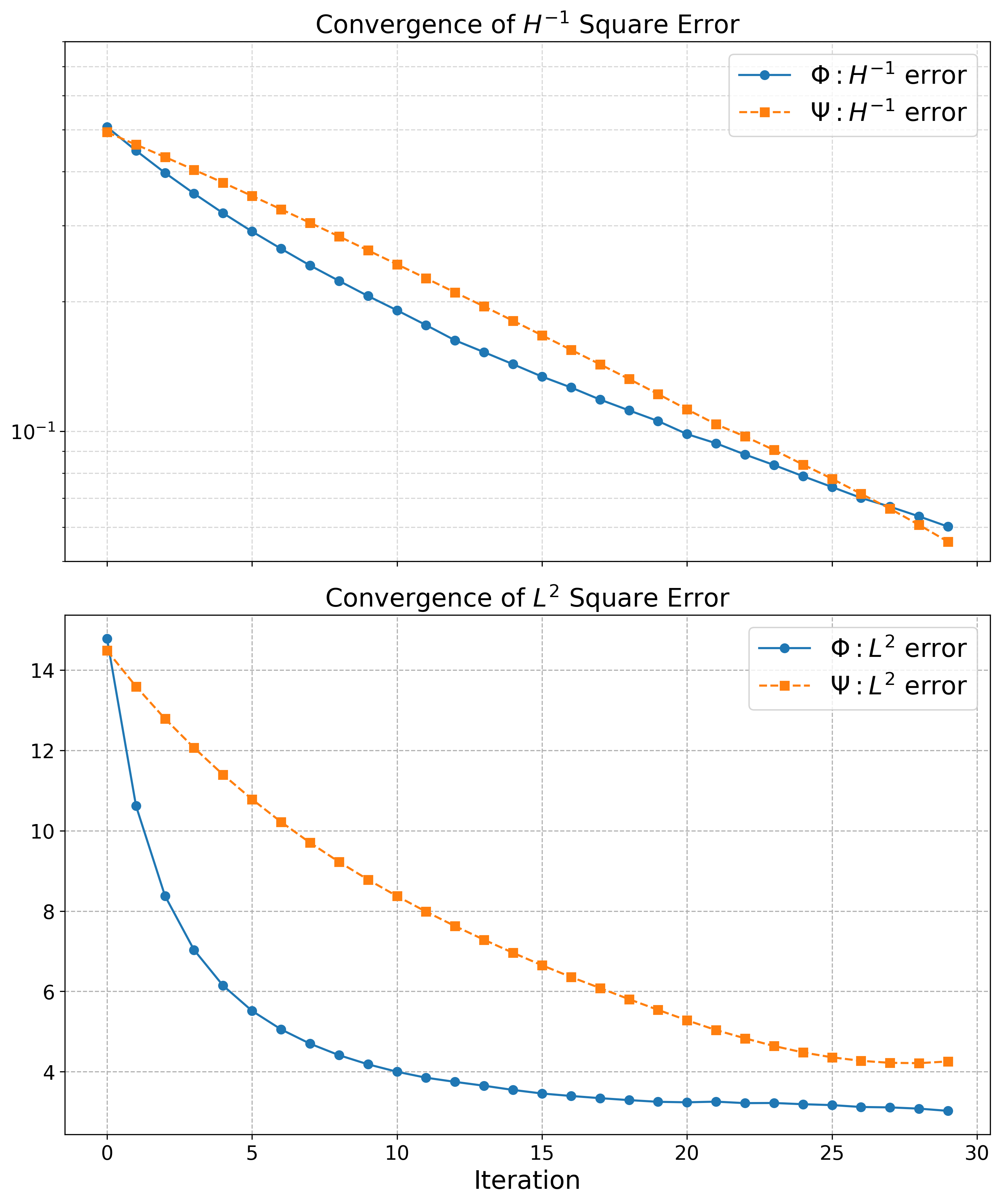}
        \caption{Near boundary}
    \end{subfigure}
    \caption{Influence of target type on the reconstruction results. The four columns correspond to annular, thin-bar, separated-disk, and near-boundary targets. In each column, the top panel is the prescribed target, the middle panel is the ray-tracing result obtained by the proposed algorithm, and the bottom panel is the $H^{-1}$ and $L^2$ residual curve.}
    \label{fig:target-type}
\end{figure}

The comparison in Fig.~\ref{fig:target-type} shows how different geometric singularities affect the reconstruction. For the annular target, the method preserves the central zero-intensity hole while concentrating the transported energy on the relatively thin illuminated ring. Some weak leakage remains near the outer boundary, but the main annular structure is maintained and the residuals decay stably. For the thin-bar target, the narrow illuminated strips and the dark gaps between neighboring bars are both well resolved, indicating that the method can preserve thin supports without excessive numerical diffusion.

For the separated-disk target, the algorithm maintains the relative positions and separation of the disconnected components without creating artificial connections through the dark region. The residuals again exhibit stable decay. These three examples show that holes, thin structures, and disconnected supports can be treated within the same computational framework.

The near-boundary target is more challenging. Although the main target shape is recovered, larger artifacts and residuals appear near \(\partial\Omega^*\). In this case, the desired transport is farther from the initial map and requires a larger global redistribution of mass. Moreover, the transport map undergoes stronger local deformation near the boundary, reducing the accuracy of the element-image approximation in \eqref{tri_approx}. The result therefore indicates that targets close to the boundary remain more demanding for the present implementation.

\subsection{Comparison of Two Discrete \texorpdfstring{$c$}{c}-Transforms}

We compare the linear fast $c$-transform with the version incorporating the local quadratic correction described in Section~\ref{High-Order Local Refinement}. All other parameters, including the mesh scale, time step, and stopping criterion, are kept unchanged so that the observed differences primarily reflect the $c$-transform step.

\begin{figure}[H]
    \centering
    \captionsetup[subfigure]{font=small, justification=centering, skip=2pt}
    \begin{subfigure}[t]{0.20\textwidth}
        \centering
        \includegraphics[width=\linewidth]{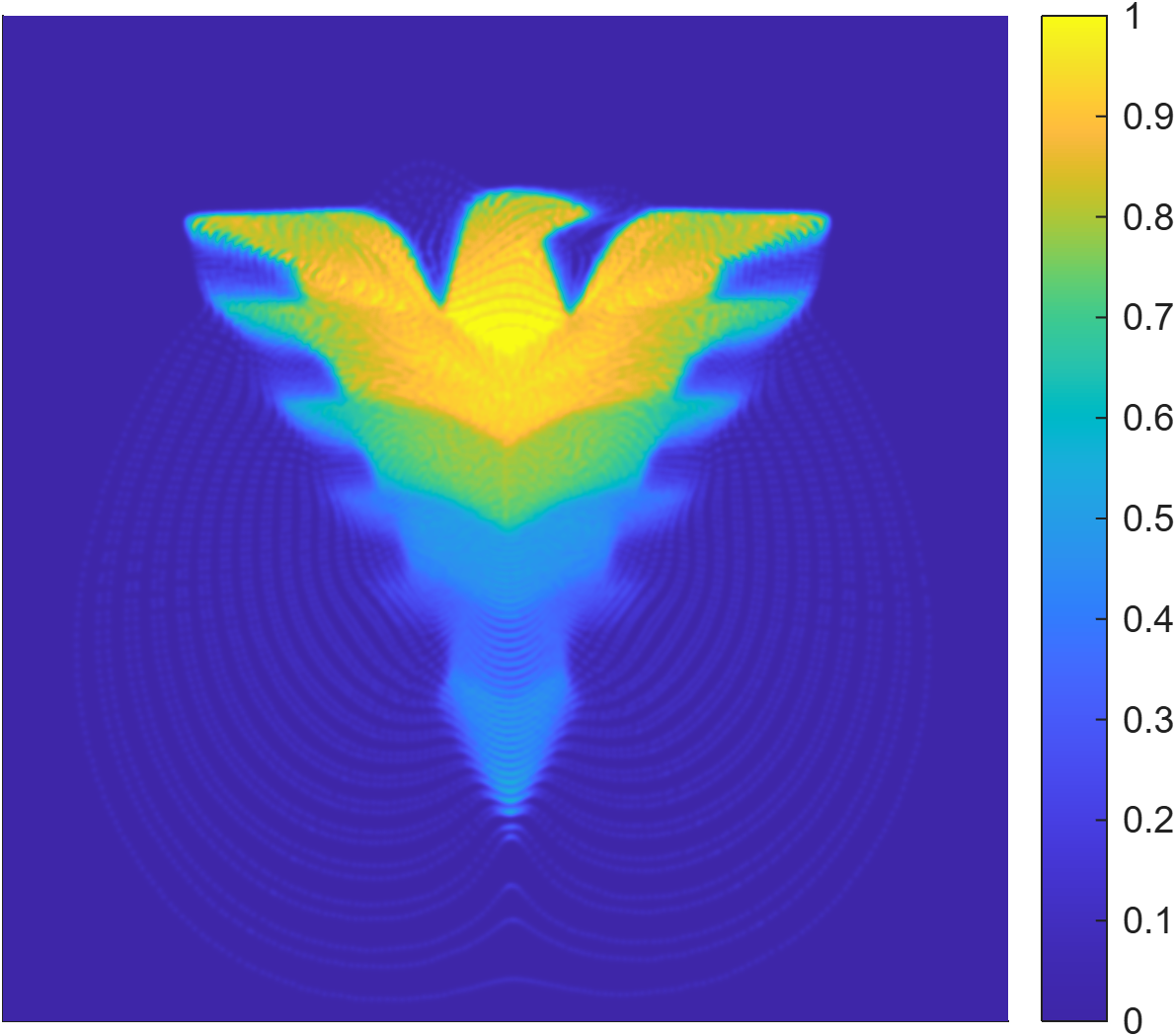}
        \caption{Linear}
    \end{subfigure}\hfill
    \begin{subfigure}[t]{0.20\textwidth}
        \centering
        \includegraphics[width=\linewidth]{simulation_gray.png}
        \caption{Quadratic}
    \end{subfigure}\hfill
    \begin{subfigure}[t]{0.20\textwidth}
        \centering
        \includegraphics[width=\linewidth]{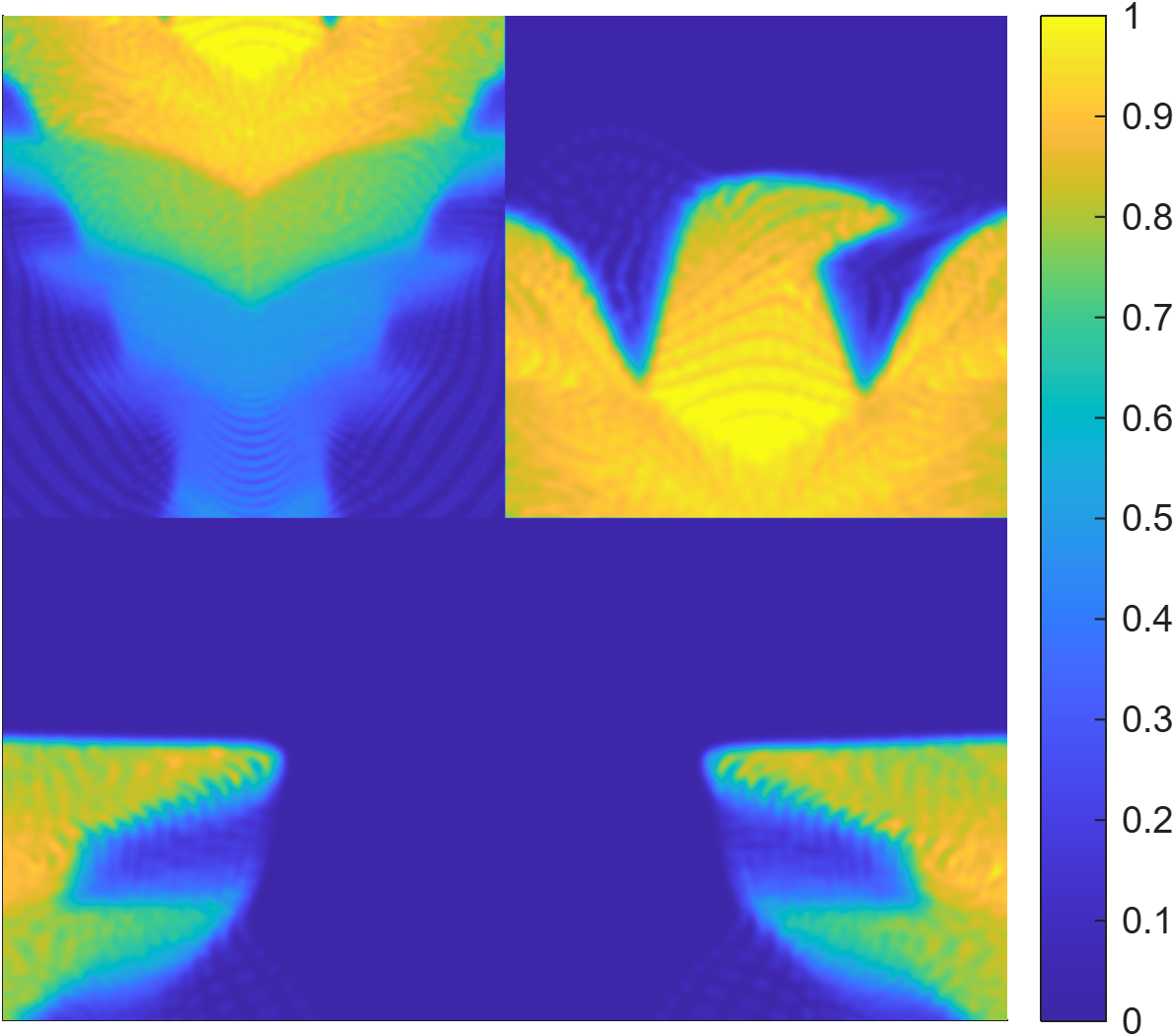}
        \caption{Linear zooms}
    \end{subfigure}\hfill
    \begin{subfigure}[t]{0.20\textwidth}
        \centering
        \includegraphics[width=\linewidth]{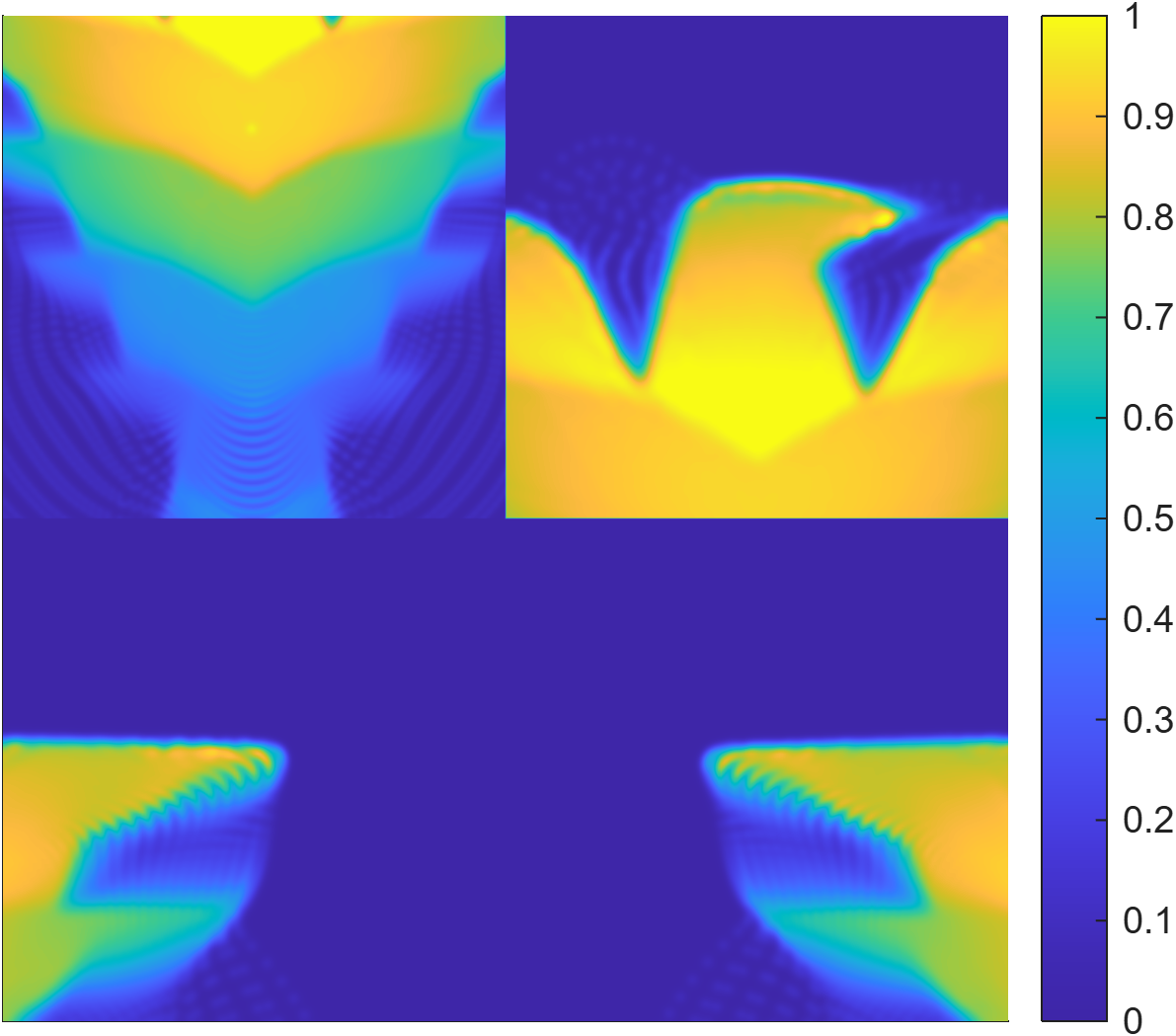}
        \caption{Quadratic zooms}
    \end{subfigure}
    \captionsetup{width=0.9\textwidth}
    \caption{Comparison between the linear fast $c$-transform and the quadratic
    fast $c$-transform algorithms. The panels show, from left to right, the
    linear result, the quadratic result, selected zoomed views of the linear
    result, and the corresponding zoomed views of the quadratic result.}
    \label{fig:ctransform-comparison}
\end{figure}

The local quadratic correction suppresses the interior fluctuations observed in the linear reconstruction and produces smoother intensity transitions. The improvement is particularly clear in the zoomed regions, where the local grayscale structure is reproduced more faithfully.

\textbf{Acknowledgments} The work was supported in part by National Key Research and Development Program of China (No. 2023YFA1009100), National Natural Science Foundation of China (No. U21A20425, No. 12501582), and a Key Laboratory of Zhejiang Province.

\end{document}